\documentclass[final,hidelinks,onefignum,onetabnum]{siamart220329}

\newcommand{\wmin}{\ensuremath{w_{\min}}}
\newcommand{\wmax}{\ensuremath{w_{\max}}}

\usepackage[]{colortbl}

\usepackage{lipsum}
\usepackage{amsfonts}
\usepackage{graphicx}
\usepackage{epstopdf}
\usepackage{algorithmic}
\ifpdf
  \DeclareGraphicsExtensions{.eps,.pdf,.png,.jpg}
\else
  \DeclareGraphicsExtensions{.eps}
\fi

\newcommand{\abs}[1]{\left|#1\right|} 
\newcommand{\norm}[1]{\left\|#1\right\|}
\usepackage[colorinlistoftodos, textsize=tiny]{todonotes}
\usepackage{mathtools}
\usepackage{tabularx}
\usepackage{booktabs}
\usepackage{enumitem}

\usepackage{thmtools}
\usepackage{thm-restate}

\renewcommand{\d}{\mathrm{d}}
\renewcommand{\Pr}[1]{\mathrm{Pr}\left[#1\right]}

\newcommand{\Expected}[1]{\mathbb{E}\left[#1\right]}
\newcommand{\Var}[1]{\mathrm{Var}\left[#1\right]}
\newcommand{\vol}{\ensuremath{\nu}}
\newcommand{\maxvar}{\ensuremath{s}}

\newcommand{\cliqueprobk}{\ensuremath{q_k}}
\DeclareMathOperator{\IRG}{IRG}
\DeclareMathOperator{\GIRG}{GIRG}

\newcommand{\Event}[1]{\mathbf{E}_{#1}}

\newcommand{\F}[1]{\mathbf{#1}}

\newcommand{\weightsequence}[1]{\{w_i\}_i^{#1}}
\newcommand{\nsim}{\ensuremath{\not\sim}}

\makeatletter
    \DeclareRobustCommand{\bfseries}%
    {%
        \not@math@alphabet\bfseries\mathbf%
        \fontseries\bfdefault\selectfont%
        \boldmath%
    }
\makeatother

\usepackage{amsmath}
\usepackage{cleveref}

\usepackage{makecell}
\usepackage{dsfont}

\newsiamremark{remark}{Remark}
\newsiamremark{hypothesis}{Hypothesis}
\crefname{hypothesis}{Hypothesis}{Hypotheses}
\newsiamthm{claim}{Claim}

\headers{Cliques in High-Dimensional GIRGs}{Friedrich, Göbel, Katzmann and Schiller}

\title{Cliques in High-Dimensional Geometric Inhomogeneous Random Graphs\thanks{Submitted to the editors DATE.
\funding{Andreas Göbel was funded by the project PAGES (project No. 467516565) of the German Research Foundation (DFG).}}}

\author{Tobias Friedrich\thanks{Hasso Plattner Institute, University of Potsdam, Germany 
  (\email{tobias.friedrich@hpi.de}, \email{andreas.goebel@hpi.de}).}
\and Andreas Göbel\footnotemark[2]
\and Maximilian Katzmann\thanks{Karlsruhe Institute of Technology, Karlsruhe, Germany (\email{maximilian.katzmann@kit.edu}).}
\and Leon Schiller\thanks{ETH Zürich, Switzerland
(\email{leon.schiller@inf.ethz.ch}).}
}

\usepackage{amsopn}

\ifpdf
\hypersetup{
  pdftitle={Cliques in High-Dimensional Geometric Inhomogeneous Random Graphs},
  pdfauthor={T. Friedrich, A. Göbel, M. Katzmann, L. Schiller}
}
\fi

\begin{document}

\maketitle

\begin{abstract}
  A recent trend in the context of graph theory is to bring
  theoretical analyses closer to empirical observations, by focusing
  the studies on random graph models that are used to represent
  practical instances.
  There, it was observed that geometric inhomogeneous random graphs
  (GIRGs) yield good representations of complex real-world networks,
  by expressing edge probabilities as a function that depends on
  (heterogeneous) vertex weights and distances in some underlying
  geometric space that the vertices are distributed in.
  While most of the parameters of the model are understood well, it
  was unclear how the dimensionality of the ground space affects the
  structure of the graphs.

  In this paper, we complement existing research into the dimension of
  geometric random graph models and the ongoing study of determining
  the dimensionality of real-world networks, by studying how the
  structure of GIRGs changes as the number of dimensions increases.
  We prove that, in the limit, GIRGs approach non-geometric
  inhomogeneous random graphs and present insights on how quickly the
  decay of the geometry impacts important graph structures.
  In particular, we study the expected number of cliques of a given
  size as well as the clique number and characterize phase
  transitions at which their behavior changes fundamentally.
  Finally, our insights help in better understanding previous results
  about the impact of the dimensionality on geometric random graphs.
\end{abstract}



\section{Introduction}

Networks are a powerful tool to model all kinds of processes that we
interact with in our day-to-day lives.  From connections between
people in social networks, to the exchange of information on the
internet, and on to how our brains are wired, networks are everywhere.
Consequently, they have been in the focus of computer science for
decades.  There, one of the most fundamental techniques used to model
and study networks are \emph{random graph models}.  Such a model
defines a probability distribution over graphs, which is typically
done by specifying a random experiment on how to construct the graph.
By analyzing the rules of the experiment, we can then derive
structural and algorithmic properties of the resulting graphs.  If the
results match what we observe on real-world networks, i.e., if the
model represents the graphs we encounter in practice well, then we can
use it to make further predictions that help us understand real graphs
and utilize them more efficiently.

The quest of finding a good model starts several decades ago, with the
famous Erdős-Rényi (ER) random graphs~\cite{Erdos_Renyi,g-rg-59}.
There, all edges in the graph exist independently with the same
probability.  Due to its simplicity, this model has been studied
extensively.  However, because the degree distribution of the
resulting graphs is rather homogeneous and they lack clustering (due
to the independence of the edges), the model is not considered to
yield good representations of real graphs.  In fact, many networks we
encounter in practice feature a degree distribution that resembles a
power-law~\cite{asvw-h-20, scm-t-21, vhhk-s-19} and the clustering
coefficient (the probability for two neighbors of a vertex to be
adjacent) is rather high~\cite{n-c-01, ws-c-98}.  To overcome these
drawbacks, the initial random graph model has been adjusted in several
ways.

In \emph{inhomogeneous random graphs (IRGs)}, often referred to as
\emph{Chung-Lu random graphs}, each vertex is assigned a weight and
the probability for two vertices to be connected by an edge is
proportional to the product of the weights~\cite{acl-rgmplg-01,
  cl-adrgged-02, cl-ccrgg-02}.  As a result, the expected degrees of
the vertices in the resulting graphs match their weight.  While
assigning weights that follow a power-law distribution yields graphs
that are closer to the complex real-world networks, the edges are
still drawn independently, leading to vanishing clustering
coefficients.

A very natural approach to facilitate clustering in a graph model is
to introduce an underlying geometry.  This was done first in
\emph{random geometric graphs (RGGs)}, where vertices are distributed
uniformly at random in the Euclidean unit square and any two are
connected by an edge if their distance lies below a certain threshold,
i.e., the neighborhood of a vertex lives in a disk centered at that
vertex~\cite{p-rgg-03}.  Intuitively, two vertices that connect to a
common neighbor cannot be too far away from each other, increasing the
probability that they are connected by an edge themselves.  In fact,
random geometric graphs feature a non-vanishing clustering
coefficient~\cite{dc-r-02}.  However, since all neighborhood disks
have the same size, they all have roughly the same expected degree,
again, leading to a homogeneous degree distribution.

To get a random graph model that features a heterogeneous degree
distribution \emph{and} clustering, the two mentioned adjustments were
recently combined to obtain \emph{geometric inhomogeneous random
  graphs (GIRGs)}~\cite{Keusch_2018}.  There, vertices are assigned a
weight \emph{and} a position in some underlying geometric space and
the probability for two vertices to connected increases with the
product of the weights but decreases with increasing geometric
distance between them.  As a result, the generated graphs have a
non-vanishing clustering coefficient and, with the appropriate choice
of the weight sequence, they feature a power-law degree distribution.
Additionally, recent empirical observations indicate that GIRGs
represent real-world networks well with respect to certain structural
and algorithmic properties~\cite{bf-evacaga-22}.

We note that GIRGs are not the first model that exhibits a
heterogeneous degree distribution and clustering.  In fact,
\emph{hyperbolic random graphs (HRGs)}~\cite{kpk-h-10} feature these
properties as well and have been studied extensively before (see,
e.g.,~\cite{Blasius_Friedrich_Krohmer_2018, fm-l-18, fhms-c-21,
  fk-dhrg-18, gpp-rhg-12}).  However, in the pursuit of finding good
models to represent real-world networks, GIRGs introduce a parameter
that sets them apart from prior models: the choice of the underlying
geometric space and, more importantly, the dimensionality of that
space.

Unfortunately, this additional parameter that sets GIRGs apart from
previous models, has not gained much attention at all.  In fact, it
comes as a surprise that, while the underlying dimensionality of
real-world networks is actively researched~\cite{abs-d-22,
  bgk-dsnume-14, dkbh-d-11, gtar-p-21, lb-mleid-04} and there is a
large body of research examining the impact of the dimensionality on
different homogeneous graph models~\cite{dc-r-02,
  Devroye_Gyoergy_Lugosi_Udina_2011,
  Erba_Ariosto_Gherardi_Rotondo_2020} with some advancements being
made on hyperbolic random graphs~\cite{yr-hdhgcn-20}, the effects of
the dimension on the structure of GIRGs have only been studied
sparsely.  For example, while it is known that GIRGs exhibit a
clustering coefficient of $\Theta(1)$ for any fixed
dimension~\cite{Keusch_2018}, it is not known how the hidden constants
scale with the dimension.

In this paper, we initiate the study of the impact of the
dimensionality on GIRGs.  In particular, we investigate the influence
of the underlying geometry as the dimensionality increases, proving
that GIRGs converge to their non-geometric counterpart (IRGs) in the
limit.  With our results we are able to explain seemingly disagreeing
insights from prior research on the impact of dimensionality on
geometric graph models.  Moreover, by studying the clique structure of
GIRGs and its dependence on the dimension $d$, we are able to quantify
how quickly the underlying geometry vanishes.  In the following, we
discuss our results in greater detail.  We note that, while we give
general proof sketches for our results, the complete proofs are
deferred to the full version \cite{FGKS23Arxiv}.

\subsection{(Geometric) Inhomogeneous Random Graphs}

Before stating our results in greater detail, let us recall the
definitions of the two graph models we mainly work with throughout the
paper.

\subparagraph*{Inhomogeneous Random Graphs (IRGs)}  The model of
inhomogeneous random graphs was introduced by Chung and
Lu~\cite{acl-rgmplg-01, cl-adrgged-02, cl-ccrgg-02} and is a natural
generalization of the Erdős-Rényi model.  Starting with a vertex set
$V$ of $n$ vertices, each $v \in V$ is assigned a weight $w_v$.  Each
edge $\{u, v\} \in \binom{V}{2}$ is then independently present with
probability
\begin{align*}
  \Pr{u \sim v} = \min\left\{ 1, \frac{\lambda w_uw_v}{n} \right\},
\end{align*}
for some constant $\lambda > 0$ controlling the average degree of the
resulting graph.  Note that assigning the same weight to all vertices
yields the same connection probability as in Erdős-Rényi random
graphs.  For the sake of simplicity, we define
$\kappa_{uv} = \min\{ \lambda w_u w_v, n\}$ such that
$\Pr{u \sim v} = \kappa_{uv}/n$.  Additionally, for a set of vertices
$U_k = \{v_1, \ldots, v_k\}$ with weights $w_1, \ldots, w_k$, we
introduce the shorthand notation $\kappa_{ij}= \kappa_{v_i v_j}$ and
write $\{\kappa\}^{(k)} = \{ \kappa_{ij} \mid 1 \le i < j \le k \}$.

Throughout the paper, we mainly focus on inhomogeneous random graphs that
feature a power-law degree distribution in expectation, which is
obtained by sampling the weights accordingly.  More precisely, for
each $v \in V$, we sample a weight $w_v$ from the Pareto distribution
$\mathcal{P}$ with parameters $1 - \beta, w_0$ and distribution
function
\begin{align*}
  \Pr{w_v \le x} = 1 - \left(\frac{x}{w_0}\right)^{1-\beta}.
\end{align*}
Then the density of $w_v$ is
$\rho_{w_v}(x) = (\beta - 1)x^{-\beta}/w_0^{1-\beta}$.  Here,
$w_0 > 0$ is a constant that represents a lower bound on the weights
in the graph and $\beta$ denotes the power-law exponent of the
resulting degree distribution.  Throughout the paper, we assume
$\beta > 2$ such that a single weight has finite expectation (and thus
the average degree in the graph is constant), but possibly infinite
variance.  We denote a graph obtained by utilizing the above weight
distribution and connection probabilities with $\IRG(n, \beta, w_0)$.
For a fixed weight sequence $\{w\}_1^n$, we denote the corresponding
graph by $\IRG(\{w\}_1^n)$.

\subparagraph*{Geometric Inhomogeneous Random Graphs (GIRGs)}
Geometric inhomogeneous random graphs are an extension of IRGs, where
in addition to the weight, each vertex $v$ is also equipped with a
position $\F{x}_v$ in some geometric space and the probability for
edges to form depends on their weights and the distance in the
underlying space~\cite{Keusch_2018}.  While, in its raw form, the GIRG
framework is rather general, we align our paper with existing analysis
on
GIRGs~\cite{Blaesius_Friedrich_Katzmann_Meyer_Penschuck_Weyand_2019,
  kl-bpgirg-16, Michielan_Stegehuis_2022} and consider the
$d$-dimensional torus $\mathbb{T}^d$ equipped with $L_\infty$-norm as
the geometric ground space, whereby $\mathbb{T}^d$ can be described as the $d$-dimensional hypercube $[0,1]^d$ with opposite bounddaries identifed with each other.  
More precisely, in what we call the
\emph{standard} GIRG model, the positions $\F{x}$ of the vertices are
drawn independently and uniformly at random from $\mathbb{T}^d$,
according to the standard Lebesgue measure.  We denote the $i$-th
component of $\F{x}_v$ by~$\F{x}_{vi}$.  Additionally, the geometric
distance between two points $\F{x}_u$ and $\F{x}_v$, is given by
\begin{align*}
  d(\F{x}_u, \F{x}_v) = \|\F{x}_u - \F{x}_v\|_\infty = \max_{1 \le i \le d}\{|\F{x}_{ui} - \F{x}_{vi}|_C\},
\end{align*}
where $|\cdot|_C$ denotes the distance on the circle, i.e,
\begin{align*}
  |\F{x}_{ui} - \F{x}_{vi}|_C = \min\{|\F{x}_{ui} - \F{x}_{vi}|, 1 - |\F{x}_{ui} - \F{x}_{vi}|\}.
\end{align*}
In a standard GIRG, two vertices $u \neq v$ are adjacent if and only
if their distance $d(\F{x}_u, \F{x}_v)$ in the torus is less than or
equal to a \emph{connection threshold} $t_{uv}$, which is given by
\begin{align*}
  t_{uv} = \frac{1}{2}\left( \frac{\lambda w_uw_v}{n} \right)^{1/d} = \left( \frac{w_uw_v}{\tau n} \right)^{1/d},
\end{align*}
where $\tau = 2^d/\lambda$.  Using $L_\infty$ is motivated by the fact
that it is the most widely used metric in the literature because it is
arguably the most natural metric on the torus. In particular, it has
the ``nice'' property that the ball of radius $r$ is a cube and
``fits'' entirely into $\mathbb{T}^d$ for all $0 \le r \le 1$.

Note that, as a consequence of the above choice, the marginal
connection probability $\Pr{u \sim v}$ is the same as in the IRG
model, i.e., $\Pr{u \sim v} = \kappa_{uv}/n$.  However, while the
probability that any given edge is present is the same as in the IRG
model, the edges in the GIRG model are \emph{not} drawn independently.
We denote a graph obtained by the procedure described above with
$\GIRG(n, \beta, w_0, d)$.  As for IRGs, we write
$\GIRG(\{w\}_1^n, d)$ when considering standard GIRGs with a fixed
weight sequence $\{w\}_1^n$

\medskip

As mentioned above, the standard GIRG model is a commonly used
instance of the more general GIRG framework~\cite{Keusch_2018}.
There, different geometries and distance functions may be used.  For
example, instead of $L_\infty$-norm, any $L_p$-norm for
$1 \le p < \infty$ may be used. Then, the distance between two
vertices $u, v$ is measured as
\begin{align*}
  \|\F{x}_u - \F{x}_v\|_p \coloneqq \begin{cases}
                                      \left( \sum_{i=1}^d |\F{x}_{ui} - \F{x}_{vi}|^p \right)^{1/p} & \text{if } p < \infty\\
                                      \max_{1 \le i \le d} \{ |\F{x}_{ui} - \F{x}_{vi}| \} & \text{otherwise.}
                                    \end{cases}
\end{align*}
With this choice, the volume (Lebesgue measure) of the ball $B_p(r)$
of radius $r$ under $L_p$-norm is equal to the probability that a
vertex $u$ falls within distance at most $r$ of $v$ (if $r = o(1)$).
We denote this volume by $\vol(r)$.  We call the corresponding graphs
\emph{standard GIRGs with any $L_p$-norm} and note that some of our
results extend to this more general model.  Finally, whenever our
insights consider an even broader variant of the model (e.g., variable
ground spaces, distances functions, weight distributions), we say that
they hold for \emph{any GIRG} and mention the constraints explicitly.

\subsection{Asymptotic Equivalence}

Our first main observations is that large values of $d$ diminish the
influence of the underlying geometry until, at some point, our model
becomes strongly equivalent to its non-geometric counterpart, where
edges are sampled independently of each other. We prove that the
\emph{total variation distance} between the distribution over all
graphs of the two models tends to zero as $n$ is kept fixed and
$d \rightarrow \infty$.  We define the total variation distance of two
probability measures $P$ and $Q$ on the measurable space
$(\Omega, \mathcal{F})$ as

\begin{align*}
  \| P, Q \|_{\text{TV}} = \sup_{A \in \mathcal{F}} | P(A) - P(B) | = \frac{1}{2} \sum_{\omega \in \Omega} |P(\omega) - Q(\omega)|,
\end{align*}
where the second equality holds if $\Omega$ is countable.  In our
case, $\Omega$ is the set $\mathcal{G}(n)$ of all possible graphs on
$n$ vertices, and $P, Q$ are distributions over these graphs. If
$G_1, G_2$ are two random variables mapping to $\Omega$, we refer to
$\| G_1, G_2\|_{\text{TV}}$ as the total variation distance of the
induced probability measures by $G_1$ and $G_2$, respectively.
Informally, this measures the maximum difference in the probability
that any graph $G$ is sampled by $G_1$ and $G_2$.

\begin{restatable}[]{theorem}{asymptoticchunglueq}\label{thm:chunglueq}
  Let $\mathcal{G}(n)$ be the set of all graphs with $n$ vertices, let
  $\{w\}_1^n$ be a weight sequence, and consider
  $G_{\text{IRG}} = \IRG(\{w\}_1^n) \in \mathcal{G}(n)$ and a standard
  GIRG $G_\text{GIRG} = \GIRG(\{w\}_1^n, d) \in \mathcal{G}(n)$ with
  any $L_p$-norm.  Then,
  \begin{equation*}
    \lim_{d \rightarrow \infty} \|G_{\text{GIRG}}, G_{\text{IRG}}\|_{\text{TV}} = 0.
  \end{equation*}
\end{restatable}

We note that this theorem holds for arbitrary weight sequences that do
not necessarily follow a power law and for arbitrary $L_p$-norms used
to define distances in the ground space.  For $p \in [1, \infty)$, the
proof is based on the application of a multivariate central limit
theorem~\cite{Raic_2019}, in a similar way as used to prove a related
statement for \emph{spherical random geometric graphs (SRGGs)}, i.e.,
random geometric graphs with a hypersphere as ground
space~\cite{Devroye_Gyoergy_Lugosi_Udina_2011}.  Our proof generalizes
this argument to arbitrary $L_p$-norms and arbitrary weight sequences.
For the case of $L_\infty$-norm, we present a proof based on the inclusion-exclusion principle and the bounds we develop in the full version \cite[Section~4]{FGKS23Arxiv}.

Remarkably, while a similar behavior was previously established for
SRGGs, there exist works indicating that RGGs on the hypercube do not
converge to their non-geometric counterpart~\cite{dc-r-02,
  Erba_Ariosto_Gherardi_Rotondo_2020} as $d \rightarrow \infty$.  We
show that this apparent disagreement is due to the fact that the torus
is a homogeneous space while the hypercube is not. In fact, our proof
shows that GIRGs on the hypercube \emph{do} converge to a
non-geometric model in which edges are, however, not sampled
independently.  This lack of independence is because, on the
hypercube, there is a positive correlation between the distances from
two vertices to a given vertex, leading to a higher tendency to form
clusters, as was observed
experimentally~\cite{Erba_Ariosto_Gherardi_Rotondo_2020}.  Due to the
homogeneous nature of the torus, the same is not true for GIRGs and
the model converges to the plain IRG model with independent edges.

\subsection{Clique Structure}

To quantify for which dimensions $d$ the graphs in the GIRG model
start to behave similar to IRGs, we investigate the number and size of
cliques.  Previous results on SRGGs indicate that the dimension of the
underlying space heavily influences the clique structure of the
model~\cite{Avrachenkov_Bobu_2020,Devroye_Gyoergy_Lugosi_Udina_2011}.
However, it was not known how the size and the number of cliques
depends on~$d$ if we use the torus as our ground space, and how the
clique structure in high-dimensions behaves for inhomogeneous weights.

We give explicit bounds on the expected number of cliques of a given
size $k$, which we afterwards turn into bounds on the \emph{clique
  number} $\omega(G)$, i.e., the size of the largest clique in the
graph $G$.  While the expected number of cliques in the GIRG model was
previously studied by Michielan and
Stegehuis~\cite{Michielan_Stegehuis_2022} when the power-law exponent
of the degree distribution satisfies $\beta \in (2, 3)$, to the best
of our knowledge, the clique number of GIRGs remains unstudied even in
the case of constant (but arbitrary) dimensionality.  We close this
gap, reproduce the existing results, and extend them to the case
$\beta \ge 3$ and the case where $d$ can grow as a function of the
number of vertices $n$ in the graph.  Furthermore, our bounds for the
case $\beta \in (2, 3)$ are more explicit and complement the work of
Michielan and Stegehuis, who expressed the (rescaled) asymptotic
number of cliques as converging to a non-analytically solvable
integral.  Furthermore, we show that the clique structure in our model
eventually behaves asymptotically like that of an IRG if the dimension
is sufficiently large.  In summary, our main contributions are
outlined in Tables~\ref{tab:expectedcliques},
\ref{tab:expectedtriangles}, and~\cref{tab:clique-number}.

\begin{table}[t] 
    \centering
    \caption{Asymptotic behavior of the expected number of
      $k$-cliques. Results marked with * were
      previously known for constant $k$. For all depicted regimes, $K_k$ concentrates well around its expectation, i.e., $K_k/\Expected{K_k}$ converges in probability to $1$ if $k$ is sufficiently small: for cells marked in light-gray, this holds for $k = o(n^{(3-\beta)/4})$; for cells marked in dark-gray, it holds for $k = o(\log(n)/(\log\log(n) + d))$; for white cells, it holds for all $k$. (cf. \cref{thm:concentration})}\label{tab:expectedcliques}
    \begin{tabularx}{\textwidth}[t]{lccc}
    \toprule
    \multicolumn{4}{c}{$\Expected{K_k}$ for $k \ge 4$} \\
    \midrule
    \hspace*{3cm}& $d = \Theta(1)$ & \makecell{$d = o(\log(n))$}  & $d = \omega(\log(n))$ \\
    \midrule
    $2 < \beta < 3, k > \frac{2}{3-\beta}$ & \cellcolor{gray!25} $n^{\frac{k}{2}(3-\beta)}\Theta(k)^{-k}$* & \cellcolor{gray!25} $n^{\frac{k}{2}(3-\beta)}\Theta(k)^{-k}$ & \cellcolor{gray!25} $n^{\frac{k}{2}(3-\beta)}\Theta(k)^{-k}$ \\
    $2 < \beta < 3, k < \frac{2}{3-\beta}$ & \cellcolor{gray!50} $n\Theta(k)^{-k}$* & \cellcolor{gray!50} $ n  e^{-\Theta(1)dk} \Theta(k)^{-k} $ & \cellcolor{gray!25} $n^{\frac{k}{2}(3-\beta)}\Theta(k)^{-k} $ \\
    $\beta > 3$ & \cellcolor{gray!50} $n\Theta(k)^{-k}$ & \cellcolor{gray!50} $ n  e^{-\Theta(1)dk}\Theta(k)^{-k}$ &  $o(1)$ \\
    \midrule
    & \makecell{equivalent to\\ HRGs \cite{Blasius_Friedrich_Krohmer_2018}}  &                                        & \makecell{equivalent to \\IRGs \cite{Janson_Luczak_Norros_2010}}  \\
\bottomrule
    \end{tabularx} 
\end{table}

We observe that the structure of the cliques undergoes three phase
transitions in the size of the cliques $k$, the dimension $d$, and the
power-law exponent $\beta$.

\subparagraph*{Transition in $k$}  When $\beta \in (2, 3)$ and
$d \in o(\log(n))$, the first transition is at
$k = \frac{2}{3-\beta}$, as was previously observed for hyperbolic
random graphs~\cite{Blasius_Friedrich_Krohmer_2018} and for GIRGs of
constant dimensionality~\cite{Michielan_Stegehuis_2022}.  The latter
work explains this behavior by showing that for
$k < \frac{2}{3-\beta}$, the number of cliques is strongly dominated
by ``geometric'' cliques forming among vertices whose distance is of
order $n^{-1/d}$ regardless of their weight.  For
$k > \frac{2}{3-\beta}$, on the other hand, the number of cliques is
dominated by ``non-geometric'' cliques forming among vertices with
weights in the order of $\sqrt{n}$.  This behavior is in contrast to
the behavior of cliques in the IRG model, where this phase transition
does not exist and where the expected number of $k$ cliques is
$\Theta\left( n^{\frac{k}{2}(3-\beta)}\right)$ for all $k \ge 3$ (if
$\beta \in (2, 3)$)~\cite{Daly_Haig_Shneer_2020}.

\subparagraph*{Transition in $d$}  Still assuming $\beta \in (2, 3)$,
the second phase transition occurs as $d$ becomes superlogarithmic.
More precisely, we show that in the high-dimensional regime, where
$d = \omega(\log(n))$, the
phase transition in $k$ vanishes, as the expected number of cliques of
size $k \ge 4$ behaves asymptotically like its
counterpart in the IRG model. Nevertheless, we can still
differentiate the two models as long as $d = o(\log^{3/2}(n))$, by
counting triangles among low degree vertices as can be seen in \cref{tab:expectedtriangles}. 

The reason for this behavior is that the number of cliques in
the case $d = \omega(\log(n))$ is already dominated by cliques forming
among vertices of weight close to $\sqrt{n}$. For those, the
probability that a clique is formed already behaves like in an IRG
although, for vertices of small weight, said probability it is still
significantly larger as long as $d = o(\log(n)^{3/2})$.

Regarding the clique number, in the case $\beta > 3$, we observe a
similar phase transition in $d$.  For constant $d$, the clique number
of a GIRG is $\Theta(\log(n)/\log\log(n)) = \omega(1)$.  We find that
this asymptotic behavior remains unchanged if
$d = \mathcal{O}(\log\log(n))$.  However, if $d = \omega(\log\log(n))$
but $d = o(\log(n))$, the clique number scales as $\Theta(\log(n)/d)$,
which is still superconstant.  Additionally if $d = \omega(\log(n))$,
we see that, again, GIRGs show the same behavior as IRGs. That is,
there are asymptotically no cliques of size larger than 3.

\begin{table}[t] 
    \centering
    \caption{Asymptotic behavior of the expected number of
      triangles. The case $\beta = \infty$ refers to the case of
      constant weights. While in the case $\beta < 3$, the number of
      triangles already behaves like that of the IRG model if
      $d = \omega(\log(n))$, in the case $\beta > 3$, the number of
      triangles remains superconstant as long as
      $d =
      o\left(\log^{3/2}(n)\right)$. }\label{tab:expectedtriangles}
    \begin{tabularx}{\textwidth}[t]{lccc}
    \toprule
    \multicolumn{4}{c}{ \hspace*{2cm} Expected number of triangles $\Expected{K_3}$ } \\
    \midrule
    \hspace{3cm} & \makecell{$d = o(\log(n))$}  & $d = \omega(\log(n))$ & $d = \omega(\log^{2}(n))$ \\
    \midrule
    $2 < \beta < \frac{7}{3}$ & $n^{\frac{3}{2}(3-\beta)} \Theta\left(1 \right)$ & $n^{\frac{3}{2}(3-\beta)} \Theta\left(1 \right)$ & $n^{\frac{3}{2}(3-\beta)} \Theta\left(1 \right)$ \\
    $\frac{7}{3} < \beta < 3$ & $n e^{-\Theta(1)d}\Theta\left(1 \right) $ & $n^{\frac{3}{2}(3-\beta)} \Theta\left(1 \right) $ & $n^{\frac{3}{2}(3-\beta)} \Theta\left(1 \right)$ \\
    $\beta > 3$ & $n e^{-\Theta(1)d}\Theta\left(1 \right)$ & $ \Omega\left(\exp\left(  \frac{\ln^3(n)}{d^2} \right) \right) $ & $\Theta(1)$ \\
    \midrule
    $\beta = \infty$ & $ n e^{-\Theta(1)d}\Theta\left(1 \right)$ & $ \Theta\left(\exp\left(  \frac{\ln^3(n)}{d^2} \right) \right) $ & $\Theta(1)$ \\
    \bottomrule
    \end{tabularx} 
\end{table}

\subparagraph*{Transition in $\beta$}  The third phase transition in the high-dimensional case
occurs at $\beta = 3$, which is in line
with the fact that networks with a power-law exponent
$\beta \in (2, 3)$ contain with high probability (w.h.p., meaning with
probability $1 - O(1/n)$) a densely connected ``heavy core'' of
$\Theta\left(n^{\frac{1}{2}(3-\beta)} \right)$ vertices with weight
$\sqrt{n}$ or above, which vanishes if $\beta$ is larger than $3$.
This heavy core strongly dominates the number of cliques of sufficient
size and explains why the clique number is
$\Theta\left(n^{\frac{1}{2}(3-\beta)} \right)$ regardless of $d$ if
$\beta \in (2,3)$. As $\beta$ grows beyond $3$, the core disappears
and leaves only very small cliques. Accordingly for $\beta > 3$ IRGs
contain asymptotically almost surely (a.a.s., meaning with probability
$1 - o(1)$) no cliques of size greater than $3$. In contrast to that,
for GIRGs of dimension $d = o(\log(n))$ (and HRGs), the clique number
remains superconstant and so does the number of $k$-cliques for any
constant $k \ge 3$. If $d = \omega(\log(n))$, there are no cliques of
size greater than 3 like in an IRG. However, as noted before, GIRGs
feature many more triangles than IRGs as long as
$d = o(\log^{3/2}(n))$.

\begin{table}[t]
  \centering
  \caption{Asymptotic behavior of the clique number of $G$ for
    different values of $d$ in the GIRG model. The behavior of the
    first column is the same as in hyperbolic random graphs
    established in \cite{Blasius_Friedrich_Krohmer_2018}, and the
    behavior in the third column is the same as that of IRG graphs
    established in \cite{Janson_Luczak_Norros_2010}. All results
    hold a.a.s. and under $L_\infty$-norm.}\label{tab:clique-number}
  \begin{tabularx}{\textwidth}[t]{lccc}
  \toprule
  \multicolumn{4}{c}{$\omega(G)$}                                                                                                                                       \\
  \toprule
   \hspace{2cm}           & $d = \mathcal{O}(\log\log(n))$                           & $d = o(\log(n))$                       & $d = \omega(\log(n))$                               \\
  \midrule
  $\beta < 3$ & $\Theta\left( n^{(3-\beta)/2}\right)$                    & $\Theta\left( n^{(3-\beta)/2}\right)$  & $\Theta\left(n^{(3-\beta)/2}\right)$                \\
  $\beta = 3$ & $\Theta \left(\frac{\log(n)}{\log\log(n)}\right)$        & $\Omega\left(\frac{\log(n)}{d}\right)$ & $\mathcal{O}\left( 1 \right)$                            \\
  $\beta > 3$ & $\Theta \left(\frac{\log(n)}{\log\log(n)}\right)$        & $\Theta\left(\frac{\log(n)}{d}\right)$ & 
  $\le 3$                                                   \\
  \midrule
              & equivalent to HRGs \cite{Blasius_Friedrich_Krohmer_2018} &                                        & equivalent to IRGs \cite{Janson_Luczak_Norros_2010} \\
  \bottomrule
  \end{tabularx}
\end{table}%

\medskip

\paragraph*{Characterizing the typical Clique}

Our analysis also yields insights into where cliques typically form within the graph. Previoulsy known in this regard was that -- for constant $d, \beta \in (2,3)$ and constant $k$ -- cliques of size $k > \frac{2}{3-\beta}$ form dominantly among vertices of weight in the order of $\sqrt{n}$, whereas for $ k < \frac{2}{3 - \beta}$, they form among vertices of pairwise distance in the order of $n^{-1/d}$ as shown in \cite{Michielan_Stegehuis_2022}. We extend these results to the case where $k$ and $d$ are allowed to be superconstant, where $\beta > 3$ and we provide a characterization in terms of the weights of the vertices associated to a clique that extend the known results even for the previously studied parameter regimes.

To be more precise, we denote by $\wmin, \wmax$ the minimal and maximal vertex weights associated to a randomly chosen clique and study  for which weights $w$, $\wmin$ and $\wmax$ are arbitrarily likely to be in the order of $w$. To this end, we define the following.

\begin{definition} \label{def:M}
  For any $w\in \mathbb{R}$ and $\varepsilon > 0$ define
  \begin{align*}
    M_\varepsilon^{(+)}(w) \coloneqq \{ x \in \mathbb{R} \mid x \le w / \varepsilon \} \text{  and  } M_\varepsilon^{(-)}(w) \coloneqq \{ x \in \mathbb{R} \mid x \ge \varepsilon w \}.
  \end{align*} Furthermore, define \begin{align*}
    M_\varepsilon(w) \coloneqq M_\varepsilon^{(+)}(w) \cap M_\varepsilon^{(-)}(w).
  \end{align*}
\end{definition}  We proceed by studying for which $w$, we can make the conditional probabilities that $\wmin$ or $\wmax$ are in $M_\varepsilon^{(+)}(w), M_\varepsilon^{(-)}(w)$ or $M_\varepsilon(w)$ arbitrarily large by adjusting $\varepsilon$. 

\begin{table}[t]
  \centering
  \caption{Dominant regimes for the minimum/maximum vertex weight associated to a clique. An entry of $M_\varepsilon(w)$ (as defined in \cref{def:M}) means that for every $p \in (0,1)$, there is an $\varepsilon > 0$ such that $\Pr{\wmin \in M_\varepsilon(w) \mid U_k \text{ is clique}} \ge p$ (resp. $\Pr{\wmax \in M_\varepsilon(w) \mid U_k \text{ is clique}} \ge p$) where $U_k$ is a set of $k$ vertices chosen u.a.r.}\label{tab:typical-clique-min}
  \begin{tabularx}{\textwidth}[t]{lccc}
  \toprule
  \multicolumn{4}{c}{Dominant Regimes for $\wmin$}                                                                                                                                       \\
  \toprule
   \hspace{4.5cm}           & $d = \Theta(1)$                           & $d = o(\log(n))$                       & $d = \omega(\log(n))$                               \\
  \midrule
  $2 < \beta < 3, k > \frac{2}{3-\beta}$ & $M_\varepsilon(\sqrt{n})$                    & $M_\varepsilon(\sqrt{n})$  & $M_\varepsilon(\sqrt{n})$                \\
  $2 < \beta < 3, k < \frac{2}{3-\beta}$ & $M_\varepsilon(1)$                    & $M^{(+)}_\varepsilon(n^{o(1)})$  & $M_\varepsilon(\sqrt{n})$                \\
  $\beta > 3$ & $M_\varepsilon(1)$                    & $M^{(+)}_\varepsilon(n^{o(1)})$  & \\
  \bottomrule                
  \end{tabularx}

  \vspace{.3cm}

  \begin{tabularx}{\textwidth}[t]{lccc}
    \toprule
    \multicolumn{4}{c}{Dominant Regimes for $\wmax$}                                                                                                                                       \\
    \toprule
     \hspace{4cm}           & $d = \Theta(1)$                           & $d = o(\log(n))$                       & $d = \omega(\log(n))$                               \\
    \midrule
    $2 < \beta < 3, k > \frac{2}{3-\beta}$ & $M_\varepsilon(k^{\frac{1}{\beta-1}}\sqrt{n})$                    & $M_\varepsilon(k^{\frac{1}{\beta-1}}\sqrt{n})$  & $M_\varepsilon(k^{\frac{1}{\beta-1}}\sqrt{n})$                \\
    $2 < \beta < 3, k < \frac{2}{3-\beta}$ & $M_\varepsilon(k^{\frac{1}{\beta-2}})$                    & $M^{(+)}_\varepsilon(k^{\frac{1}{\beta-2}}n^{o(1)})$  & $M_\varepsilon(k^{\frac{1}{\beta-1}}\sqrt{n})$                \\
    $\beta > 3$ & $M_\varepsilon(k^{\frac{1}{\beta-2}})$                    & $M^{(+)}_\varepsilon(k^{\frac{1}{\beta-2}}n^{o(1)})$  & $-$                \\
    \bottomrule
    \end{tabularx}
\end{table}%

Our results are summarized in \cref{tab:typical-clique-min}.  The central result of these two tables is that assuming $d = o(\log(n))$, if $k < \frac{2}{3-\beta}$ or if $k$ is arbitrary and $\beta > 3$, then cliques dominantly form among vertices of very small weight, more precisely among vertices of weight at most $k^{\frac{1}{2-\beta}} n^{o(1)}$, which (if you take our results on the clique number into account) is $n^{o(1)}$ in total for cliques of all sizes that appear in the model with non-vanishing probability. On the other hand, if $d = \omega(\log(n))$ and $\beta \in (2,3)$, then cliques of all sizes dominantly form among very high-weight vertices, more precisely among vertices of weight at least in the order of $\sqrt{n}$. Again, this is the same behavior as in IRGs. We formalize this result in the following theorem.

\begin{theorem}
  Let $U_k$ be a set of $k$ randomly chosen vertices. If $\beta \in (2,3)$, $k < \frac{2}{3-\beta}$ and $d = o(\log(n))$ then there is a function $f(n) = e^{\Theta(1)d} = n^{o(1)}$ such that for all $p \in (0,1)$, there is an $\varepsilon > 0$ such that \begin{align*}
    \Pr{\wmax \le f(n) / \varepsilon \mid U_k \text{ is clique} } \ge p.
  \end{align*} If $\beta \in (2,3)$ and $d = \omega(\log(n))$, then for all (potentially superconstant) $k \ge 3$ and all $p \in (0,1)$, there is a $\varepsilon > 0$ such that \begin{align*}
    \Pr{ \wmin \ge \varepsilon \sqrt{n} \mid U_k \text{ is clique}} \ge p.
  \end{align*}
\end{theorem}

Moreover, it is worth noting that for $d, k$ constant, cliques of size $k < \frac{2}{3-\beta}$, if $\beta \in (2,3)$, and cliques of size $k \ge 3$, if $\beta > 3$, dominantly form among vertices of only constant weight. Additionally, we remark that the dependence of $\wmax$ on $k$ as given in \cref{tab:typical-clique-min} is the same as one would expect in a star centered at the vertex of minimal weight $v_{\min}$. That is, if the weight of $v_{\min}$ is much smaller than $\sqrt{n}$, then each neighbor $u$ of $v_{\min}$ has a weight that is essentially a sample of a Pareto distribution with exponent $\beta - 1$ instead of $\beta$, since conditioning on $u \sim v_{\min}$ imposes a bias towards a higher weight of $u$. Since there are $\Theta(k)$ neighbors, the maximum weight among these is essentially the maximum of $\Theta(k)$ independent samples from this distribution, which is typically of order $k^{\frac{1}{\beta - 2}}$. If $\wmin$ is already of order $\sqrt{n}$, the situation is similar, however, conditioning on $u \sim v_{\min}$  no longer imposes a bias towards a higher weight of $u$ as vertices with weight in this range are all connected with probability $\Omega(1)$. Thus, the weight of the neighbors of $v_{\min}$ continue to follow a Pareto distribution with exponent $\beta$ and the maximal weight among them is of order $k^{\frac{1}{\beta - 1}}$. Our results show that this known behavior for stars remains essentially true for cliques, that is, conditioning on having a clique does not induce a bias towards much larger weights than conditioning on having a star.

\paragraph*{Concentration Bounds} 

The above analysis does not only give insights into where cliques dominantly form, but also allows us to establish concentration bounds on the number of cliques in a similar way as done in \cite{Michielan_Stegehuis_2022}. More precisely, it allows us to establish that $K_k$ rescaled by its expectation converges in probability to $1$ for almost all the regimes we consider and almost all relevant sizes of $k$. We write $K_k/\Expected{K_k} \rightarrow_p 1$ to denote convergence in probability and formalize in our statement in the following theorem.
\begin{restatable}[]{theorem}{concentration}\label{thm:concentration}
  We have $K_k/\Expected{K_k} \rightarrow_p 1$, that is for all $\delta > 0$, \begin{align*}
    \Pr{ \left| \frac{K_k}{\Expected{K_k}} - 1 \right| \ge \delta } = o(1)
  \end{align*} if one of following conditions holds. \begin{enumerate}
    \item[(i)] $d = o(\log(n))$, $\beta \in (2,3)$, $k \neq \frac{2}{3-\beta}$, and $k = o(n^{(3-\beta)/4})$
    \item[(ii)] $d = \omega(\log(n))$, $\beta \in (2,3)$, $k = o(n^{(3-\beta)/4})$
    \item[(iii)] $d = o(\log(n))$, $\beta > 3$ and $k = o\left( \log(n)/(\log\log(n) + d) \right)$.
  \end{enumerate}
\end{restatable}
We remark that even for values of $k$ larger as the ones stated above, our results imply (slightly weaker) concentration bounds. General bounds on the variance of cliques are given in \Cref{sec:variance}

\medskip

\paragraph*{Proof Techniques}

The proofs of our results (i.e., the ones in the above tables) are
mainly based on bounds on the probability that a set of $k$ randomly
chosen vertices forms a clique. To obtain concentration bounds on the
number of cliques as needed for deriving bounds on the clique number,
we use the second moment method and Chernoff bounds.

For the case of $d = \omega(\log(n))$, many of our results are derived
from the following general insight. We show that for and all
$\beta > 2$, the probability that a set of vertices forms a clique
already behaves similar as in the IRG model if the weights of the
involved nodes are sufficiently large. For $d = \omega(\log(n)^2)$,
this holds in the entire graph, that is, regardless of the weights of
the involved vertices. In fact our statement holds even more
generally. That is, the described behavior not only applies to the
probability that a clique is formed but also to the probability that
any set of edges (or a superset thereof) is created.  

\begin{restatable}[]{theorem}{cliqueprobabilityhighdim}
  \label{thm:chung_lu_eq}
  Let $G$ be a standard GIRG and let $k \ge 3$ be a constant.
  Furthermore, let $U_k = \{v_1, \ldots, v_k\}$ be a set of vertices
  chosen uniformly at random and let
  $\{\kappa\}^{(k)} = \{\kappa_{ij} \mid 1 \le i, j \le k\}$ describe the
  pairwise product of weights of the vertices in $U_k$. Let $E(U_k)$
  denote the (random) set of edges formed among the vertices in
  $U_k$. Then, for any set of
  edges $\mathcal{A} \subseteq \binom{U_k}{2}$,
  
  \begin{align*}
    \Pr{E(U_k) \supseteq \mathcal{A} \mid \{\kappa\}^{(k)}} = \begin{cases} (1\pm o(1)) \prod_{\{i,j\}\in \mathcal{A}} \frac{\kappa_{ij}}{n} &\text{if } d = \omega(\log^2(n))\\
      (1\pm o(1)) \prod_{\{i,j\}\in \mathcal{A}} \left( \frac{\kappa_{ij}}{n} \right)^{1 \mp \mathcal{O}\left( \frac{\log(n)}{d} \right)} &\text{if } d = \omega(\log(n)).
    \end{cases}
  \end{align*}
\end{restatable}
For the proof we derive elementary bounds on the probability of the
described events and use series expansions to investigate their
asymptotic behavior.  Remarkably, in contrast to our bounds for the
case $d = o(\log(n))$, the high-dimensional case requires us to pay
closer attention to the topology of the torus.

We leverage the above theorem 
to prove that GIRGs eventually become equivalent to IRGs with respect
to the total variation distance.  \Cref{thm:chung_lu_eq} already
implies that the expected number of cliques in a GIRG is
asymptotically the same as in an IRG for all $k \ge 3$ and all
$\beta > 2$ if $d = \omega(\log^2(n))$. However, we are able to show
that the expected number of cliques for $\beta \in (2,3)$ actually
already behaves like that of an IRG if $d = \omega(\log(n))$. The
reason for this is that the clique probability among high-weight
vertices starts to behave like that of an IRG earlier than it is the
case for low-weight vertices and cliques forming among these
high-weight vertices already dominate the number of cliques. Moreover,
the clique number behaves like that of an IRG if $d = \omega(\log(n))$
for all $\beta > 2$. However, the number of triangles among vertices
of constant weight asymptotically exceeds that of an IRG as long as
$d = o(\log^{3/2}(n))$, which we prove 
by deriving even sharper bounds on the expected number of
triangles. Accordingly, convergence with respect to the total
variation distance cannot occur before this point (this holds for all
$\beta > 2$).

In contrast to this, for the low-dimensional case (where
$d = o(\log(n))$), the underlying geometry still induces strongly
notable effects regarding the number of sufficiently small cliques for
all $\beta > 2$.  However, even here, the expected number of such
cliques decays exponentially in $dk$.  The main difficulty in showing
this is that we have to handle the case of inhomogeneous weights,
which significantly influence the probability that a set of $k$
vertices chosen uniformly at random forms a clique.  To this end, we
prove the following theorem that bounds the probability that a clique
among $k$ vertices is formed if the ratio of the maximal and minimal
weight is at most $c^d$. Note that the vertices forming a star is necessary for a clique to form. For this reason we consider the event $\Event{\text{star} }^c$ of the vertices forming a star centered at the lowest weight vertex. The theorem generalizes a result of
Decreusefond et
al.~\cite{Decreusefond_Ferraz_Randriambololona_Vergne_2014}.

\begin{restatable}{theorem}{sharpcliquebounds}\label{thm:sharpcliquebounds}
  Let $G$ be a standard GIRG and consider $k \ge 3$.  Furthermore, let
  $U_k = \{ v_1, v_2, \ldots, v_k \}$ be a set of vertices chosen
  uniformly at random and assume without loss of generality that
  $w_1 \le \ldots \le w_k$.  Let $\Event{\text{star} }^c$ be the event
  that $v_1$ connects to all vertices in $U_k \setminus \{v_1\}$ and
  that $w_k \le c^d w_1$ for some constant $c \ge 1$ with
  $c^2 \left( w_1^2/(\tau n) \right)^{1/d} \le 1/4$.  Then, the
  probability that $U_k$ is a clique conditioned on
  $\Event{\text{star}}^c$ fulfills
  \begin{align*}
    \left(\frac{1}{2}\right)^{d(k-1)}k^d \le \Pr{U_k \text{ is clique} \mid \Event{\text{star}}^c} \le c^{d(k-2)} \left(\frac{1}{2}\right)^{d(k-1)} k^d.
  \end{align*} 
\end{restatable}

Building on the variant by Decreusefond et
al.~\cite{Decreusefond_Ferraz_Randriambololona_Vergne_2014}, we
provide an alternative proof of the original statement, showing that
the clique probability conditioned on the event
$\Event{\text{star}}^c$ is monotonous in the weight of all other
vertices.  Remarkably, this only holds if we condition on the event
that the center of our star is of minimal weight among the vertices in
$U_k$.

We apply \cref{thm:sharpcliquebounds} to bound the clique probability
in the whole graph (where the ratio of the maximum and minimum weight
of vertices in $U_k$ is not necessarily bounded). 
Afterwards, 
we additionally use Chernoff bounds and the second moment method to
bound the clique number.

\subsection{Relation to Previous Analyses}

In the following, we discuss how our results compare to insights
obtained on similar graph models that (apart from not considering
weighted vertices) mainly differ in the considered ground space.  We
note that, in the following, we consider GIRGs with uniform weights in
order to obtain a valid comparison.

\subparagraph*{Random Geometric Graphs on the Sphere}  Our results
indicate that the GIRG model on the torus behaves similarly to the
model of Spherical Random Geometric Graphs (SRGGs) in the
high-dimensional case. In this model, vertices are distributed on the
surface of a $d-1$ dimensional sphere and an edge is present whenever
the Euclidean distance between two points (measured by their inner
product) falls below a given threshold. Analogously to the behavior of
GIRGs, when keeping $n$ fixed and considering increasing
$d \rightarrow \infty$, this model converges to its non-geometric
counterpart, which in their case is the Erdős–Rényi
model~\cite{Devroye_Gyoergy_Lugosi_Udina_2011}.  It is further shown
that the clique number converges to that of an Erdős–Rényi graph (up
to a factor of $1 + o(1)$) if $d = \omega(\log^3(n))$.

Although the overall behavior of SRGGs is similar to that of GIRGs,
the magnitude of $d$ in comparison to $n$ at which non-geometric
features become dominant seems to differ. In fact, it is shown in
\cite[proof of Theorem 3]{Bubeck_Ding_Eldan_Racz_2015} that the
expected number of triangles in sparse SRGGs still grows with $n$ as
long as $d = o(\log^3(n))$, whereas its expectation is constant in the
non-geometric, sparse case (as for Erdős–Rényi graphs). On the other
hand, in the GIRG model, we show that the expected number of triangles
in the sparse case converges to the same (constant) value as that of
the non-geometric model if only $d = \omega(\log^{3/2}(n))$.  This
indicates that, in the high-dimensional regime, differences in the
nature of the underlying geometry result in notably different
behavior, although in the case of constant dimensionality, the models
are often assumed to behave very similarly.

\subparagraph*{Random Geometric Graphs on the Hypercube}
\label{apx:torusvscube}

The work of Dall and Christensen~\cite{dc-r-02} and the recent work of
Erba et al.~\cite{Erba_Ariosto_Gherardi_Rotondo_2020} show that RGGs
on the hypercube do \emph{not} converge to Erdős–Rényi graphs as $n$
is fixed and $d \rightarrow \infty$. However, our results imply that
this is the case for RGGs on the torus. These apparent disagreements
are despite the fact that Erba et al. use a similar central limit
theorem for conducting their calculations and
simulations~\cite{Erba_Ariosto_Gherardi_Rotondo_2020}.

The tools established in our paper yield an explanation for this
behavior.  Our proof of \cref{thm:chunglueq} relies on the fact that,
for independent zero-mean variables $Z_1, \dots, Z_d$, the covariance
matrix of the random vector $Z = \sum_{i=1}^d Z_i$ is the identity
matrix.  This, in turn, is based on the fact that the torus is a
\emph{homogeneous space}, which implies that the probability measure
of a ball of radius $r$ (proportional to its Lebesgue measure or
volume, respectively) is the same, regardless of where this ball is
centered.  It follows that the random variables $Z_{(u,v)}$ and
$Z_{(u, s)}$, denoting the normalized distances from $u$ to $s$ and
$v$, respectively, are independent.  As a result their covariance is
$0$ although both ``depend'' on the position of $u$.

For the hypercube, this is not the case. Although one may analogously
define the distance of two vertices as a sum of independent, zero-mean
random vectors over all dimensions just like we do in this paper, the
random variables $Z_{(u,v)}$ and $Z_{(u, s)}$ do \emph{not} have a
covariance of $0$.

\subsection{Conjectures \& Future Work}

While making the first steps towards understanding GIRGs and sparse
RGGs on the torus in high dimensions, we encountered several questions
whose investigation does not fit into the scope of this paper.  In the
following, we give a brief overview of our conjectures and possible
starting points for future work.

\paragraph*{Noisy GIRGs}

It would be interesting to extend our results to the \emph{temperate} version of GIRGs, where the threshold is softened using a \emph{temperature} parameter.  That is, while the probability for an edge to exist still decreases with increasing distance, we can now have longer edges and shorter non-edges with certain probabilities.  The motivation of this variant of GIRGs is based on the fact that real data is often noisy as well, leading to an even better representation of real-world
graphs. In this regard, we remark that most of our proof techniques for the case of constant dimension carry over quite directly to temperate GIRGs. However, when considering non-constant dimension, having an additional temperature parameter seems to complicate things by a lot, which is the reason why we concentrate on threshold GIRGs in this work. Nevertheless, we note that both temperature and dimensionality affect the influence of the underlying geometry, so it would be interesting to further investigate whether a sufficiently high temperature has additional impact on how quickly GIRGs converge to IRGs.

\paragraph*{Testing Thresholds for Detecting Underlying Geometry}

Another crucial question is under which circumstances the underlying
geometry of our model remains detectable by means of statistical
testing, and when (i.e., for which values of $d$) our model converges
in total variation distance to its non-geometric counterpart.  A large
body of work has already been devoted to this question for SRGGs~\cite{Devroye_Gyoergy_Lugosi_Udina_2011,
  Bubeck_Ding_Eldan_Racz_2015, Brennan_Bresler_Nagaraj_2020,
  Liu_Racz_2021, Liu_Mohanty_Schramm_Yang_2021} and recently also for
random intersection graphs~\cite{Brennan_Bresler_Nagaraj_2020}.  While
the question when these graphs lose their geometry in the dense case
is already largely answered, it remains open for the sparse case
(where the marginal connection probability is proportional to $1/n$)
and progress has only been made
recently~\cite{Brennan_Bresler_Nagaraj_2020,
  Liu_Mohanty_Schramm_Yang_2021}.  It would be interesting to study this question for our model, both for the case of
constant and for the case of inhomogeneous weights. Our work indicates that GIRGs and RGGs on the torus might lose their geometry earlier than SRGGs as the number of triangles is in expectation already the same as in an
Erdős-Rényi graph if $d = \omega(\log^{3/2}(n))$, while for SRGGs
this only happens if $d = \omega(\log^3(n))$~\cite{Bubeck_Ding_Eldan_Racz_2015}.

Furthermore, it remains to investigate dense RGGs on the torus in this regard, where the marginal connection probability of any pair of vertices is constant and does not decrease with $n$. For dense SRGGs, an analysis
of the high-dimensional case has shown that the underlying geometry
remains detectable as long as $d = o(n^3)$ while for sparse SRGGs it is conjectured that the respective threshold is only at $d = \log(n)^3$. In the dense case, this is accomplished by counting so-called \emph{signed triangles}~\cite{Bubeck_Ding_Eldan_Racz_2015}. Although for the sparse case, signed triangles have no advantage over ordinary triangles, they are much more powerful in the dense case and might prove useful for analyzing dense RGGs on the torus as well. Additionally, as GIRGs contain both very sparse and very dense parts, it is an interesing question whether inhomogeneous weights actually result in different testing thresholds somewhere between that of the dense and the sparse case.

\section{Preliminaries}

We let $G = (V, E)$ be a (random) graph on $n$ vertices.  We let
$\binom{U}{2}$ be the set of all \emph{possible} edges among vertices
of a subset $U \subseteq V$ and denote the \emph{actual} set of edges
between them by $E(U) = E \cap \binom{U}{2}$.  A $k$-clique in $G$ is
a complete induced subgraph on~$k$ vertices of $G$.  We let $K_k$
denote the random variable representing the number of $k$-cliques in
$G$.  We typically use $U_k = \{v_1, \ldots, v_k\}$ to denote a set of
$k$ vertices chosen independently and uniformly at random from the
graph.  For the sake of brevity, we simple write that $U_k$ is a set
of \emph{random vertices} to denote that $U_k$ is obtained that way.
Further, we write $w_1, \ldots, w_k$ for the weights of the vertices
in $U_k$.  The probability that $U_k$ forms a clique is denoted by
$\cliqueprobk$.  Additionally, $\Event{\text{star}}$ is the event
that~$U_k$ is a star with center $v_1$.

We use standard Landau notation to describe the asymptotic behavior of
functions for sufficiently large $n$. That is, for functions $f,g$, we
write $f(n) = \mathcal{O}(g(n))$ if there is a constant $c > 0$ such
that for all sufficiently large $n$, $f(n) \le c g(n)$. Similarly, we
write $f(n) = \Omega(g(n))$ if $f(n) \ge c g(n)$ for sufficiently
large $n$. If both statements are true, we write
$f(n) = \Theta(g(n))$. Regarding our study of the clustering
coefficient, some results make a statement about the asymptotic
behavior of a function with respect to a sufficiently large $d$. These
are marked by
$\mathcal{O}_d(\cdot), \Omega_d(\cdot), \Theta_d(\cdot)$,
respectively.

\subsection{Spherical Random Geometric Graphs (SRGGs)}

In this model, $n$ vertices are distributed uniformly on the
$d$-dimensional unit sphere $\mathcal{S}^{d-1}$ and vertices $u,v$
connected whenever their $L_2$-distance is below the connection
threshold $t_{uv}$, which is again chosen such that the connection
probability of $u,v$ is fixed. This model thus differs from the GIRG
model in its ground space (sphere instead of torus) and the fact that
it uses homogeneous weights, i.e., the marginal connection probability
between each pair of vertices is the same. We mainly use this model as
a comparison, since its behavior in high-dimensions was extensively
studied previously~\cite{Avrachenkov_Bobu_2020,
  Bubeck_Ding_Eldan_Racz_2015, Devroye_Gyoergy_Lugosi_Udina_2011}.

\subsection{Useful Bounds and Concentration Inequalities}

Throughout this paper, we use the following approximation of the
binomial coefficient.
\begin{lemma}\label{lem:binomapprox}
  For all $n \ge 1$ and all $1 \le k \le \frac{1}{2}n$, we
  have \begin{align*} \binom{n}{k} = n^k \Theta(k)^{-k}.
\end{align*} That is, there a are constants $c_1, c_2 > 0$ such that for all $n \ge 1$, \begin{align*}
    n^k \left( c_1 k \right)^{-k} \le \binom{n}{k} \le n^k \left( c_2 k \right)^{-k}.
\end{align*}
\end{lemma} \begin{proof}
We start with the upper bound and immediately get that for all $n, k$, \begin{align*}
    \binom{n}{k} \le \frac{n^k}{k!}.
\end{align*} From Stirling's approximation, we get for all $k \ge 1$ that 
\begin{align*}
    \sqrt{2\pi k} \left( \frac{k}{e} \right)^k e^{\frac{1}{12k + 1}} \le k!.
\end{align*} Because $k \ge 1$, the left side is lower bounded by $\left( \frac{k}{e} \right)^k$ and hence, \begin{align*}
    \binom{n}{k} \le n^k \left( e^{-1}k \right)^{-k}.
\end{align*}

For the lower bound, we observe that
\begin{align*}
  \binom{n}{k} \ge \left( \frac{(n-k)}{k} \right)^k = (n-k)^k k^{-k}.
\end{align*}
We claim that there is a constant $c > 0$ such that
$(n - k)^k \ge (cn)^{k}$, which is equivalent to
$1 - \frac{k}{n} \ge c$. As $k \le \frac{1}{2}n$, this inequality is
true for all $c < \frac{1}{2}$, which finishes the proof.
\end{proof}

We use the following well-known concentration bounds.
\begin{theorem}[Theorem 2.2 in \cite{Keusch_2018}, Chernoff-Hoeffding Bound]\label{thm:chernoff-hoeffding}
  For $1 \le i \le k$, let $X_i$ be independent random variables
  taking values in $[0,1]$, and let $X \coloneqq \sum_{i=1}^k
  X_i$. Then, for all $0 < \varepsilon < 1$,
  \begin{enumerate}[label=(\roman*),leftmargin=2\parindent]
  \item $\Pr{X > (1 + \varepsilon)\Expected{X}} \le \exp \left(-\frac{\varepsilon^2}{3}\Expected{X} \right)$.
  \item $\Pr{X < (1 - \varepsilon)\Expected{X}} \le \exp \left(-\frac{\varepsilon^2}{2}\Expected{X} \right)$.
  \item $\Pr{X \ge t} \le 2^{-t}$ for all $t \ge 2e\Expected{X}$.
  \end{enumerate} 
\end{theorem}

\section{Cliques in the Low-Dimensional Regime}

We start by proving the results from \cref{tab:expectedcliques},
\cref{tab:expectedtriangles}, and \cref{tab:clique-number} for the
case $d = o(\log(n))$. We remind the reader that the results in this
section hold for the standard GIRG model but remark that our bounds up
to (and including) \cref{sec:lowerboundspk} are also applicable if
norms other than $L_\infty$ are used.

\subsection{Bounds on the Clique Probability}

Recall that we denote by $K_k$ the random variable that represents the
number of cliques of size $k$ in $G$ and that $\cliqueprobk$ is the
probability that a set of $k$ vertices chosen uniformly at random
forms a clique. Then the expectation of $K_k$ is
\begin{equation*}
  \Expected{K_k} = \binom{n}{k} \cliqueprobk.
\end{equation*}
In the following, we derive upper and lower bounds on
$\cliqueprobk$. Our bounds here are very general and remain valid
regardless of how the dimension scales with $n$ and which $L_p$-norm
is used. One may also easily extend them to the non-threshold version
of the weight sampling model. Although our bounds are asymptotically
tight for constant $d$, they become less meaningful if~$d$ scales with
$n$. We therefore derive sharper bounds in \cref{sec:highdim} for the
case $d = \omega(\log(n))$.
 
\paragraph{An Upper Bound on $\cliqueprobk$}\label{sec:upperboundspk} 

In this section, we derive an upper bound on $\cliqueprobk$ by
considering the event that a set of~$k$ random vertices forms a star
centered around the vertex of minimal weight. As this is necessary to
form a clique, it gives us an upper bound on $\cliqueprobk$ that is
very general and independent of $d$. To get sharper upper bounds, we
combine this technique with \cref{thm:sharpcliquebounds} in
Section~\ref{sec:expectationbounds}.

Recall that $U_k = \{v_1, \ldots, v_k\}$ is a set of $k$ random
vertices with (random) weights $w_1, \ldots, w_k$.  In the following,
we assume without loss of generality that $v_1$ is of minimal weight
among all vertices in $U_k$. We start by analyzing how the minimal
weight $w_1$ is distributed.
\begin{lemma}\label{lem:distributions}
  Let $G$ be any GIRG with a power-law weight distribution with
  exponent $\beta > 2$.  Furthermore, let $U_k = \{v_1, \ldots, v_k\}$
  be a set of $k$ random vertices and assume that $v_1$ is of minimal
  weight $w_1$ among $U_k$.  Then, $w_1$ is distributed according to
  the density function
  \begin{align*}
    \rho_{w_1}(x) = \frac{(\beta-1)k}{w_0^{(1-\beta)k}} \cdot x^{(1-\beta)k - 1}
  \end{align*}
  in the interval $[w_0, \infty]$.  Conditioned on the weight $w_1$,
  the weight $w_i$ for all $2 \le i \le k$ is distributed
  independently as
  \begin{align*}
    \rho_{w_i \mid w_1}(x) = \frac{\beta-1}{w_1^{1-\beta}}x^{-\beta}.
  \end{align*} 
\end{lemma}
\begin{proof}
    Recall that the weight $w$ of each vertex is independently sampled from the Pareto distribution such that \begin{align}\label{eq:prob-pareto}
        \Pr{w \le x} = 1 - \left( \frac{x}{w_0} \right)^{1-\beta}.
    \end{align}
    Accordingly, the probability that the minimal weight $w_1$ is at most $x$ is \begin{align*}
        \Pr{w_1 \le x} = 1-\Pr{w \ge x}^k = 1 - \left( \frac{x}{w_0} \right)^{k(1-\beta)}.
    \end{align*}
    To find the density function of $w_1$, we differentiate this term and get \begin{align*}
        \rho_{w_1}(x) = \frac{ \d \Pr{w_v \le x} }{\d x} = \frac{\d}{\d x} \left(1- \left(\frac{x}{w_0}\right)^{(1-\beta)k}\right) \nonumber
        =  \frac{(\beta-1)k}{w_0^{(1-\beta)k}} \cdot x^{(1-\beta)k - 1}.
    \end{align*} The conditional density function $\rho_{w_i \mid w_1}(x)$ of $w_i$ is then \begin{align*}
        \rho_{w_i \mid w_1}(x) = \frac{\rho_w(x)}{\int_{w_1}^\infty \rho_w(x)\d x} = \frac{x^{-\beta}}{\int_{w_1}^\infty x^{-\beta}\d x} = \frac{\beta-1}{w_1^{1-\beta}} x^{-\beta}
    \end{align*} where $\rho_w(x) = \frac{\beta-1}{w_0^{1-\beta}}x^{-\beta}$ is the (unconditional) density function of a single weight.
\end{proof}

We proceed by bounding the probability of the event
$\Event{\text{star}}$ that $U_k$ is a star with center~$v_1$. We start
with the following lemma.

\begin{lemma}\label{lem:starbound}
  Let $G$ be any GIRG with a power-law weight distribution with
  exponent $\beta > 2$, let $U_k = \{v_1, \ldots, v_k\}$ be a set of
  $k$ random vertices, and assume that $v_1$ is of minimal weight
  $w_1$ among $U_k$.  Furthermore, let $\Event{\text{star}}$ be the
  event that $U_k$ is a star with center $v_1$ and let
  $w_-, w_+ \ge w_0$. Then, \begin{align*} \Pr{\Event{\text{star}}
      \cap w_- \le w_1 \le w_+} \le \mathcal{C}
    \left(\frac{\lambda}{n}\right)^{k-1} \left( \frac{\beta -1}{\beta
        -2} \right)^{k-1} \left( w_+^{k(3-\beta)-2} - w_-^{k(3-\beta)
        - 2}\right)
    \end{align*} 
    with \begin{align*}
        \mathcal{C} \coloneqq \frac{(\beta-1)k w_0^{-(1-\beta)k}}{(3-\beta)k-2}.
    \end{align*}
\end{lemma} \begin{proof}
    Define $P \coloneqq \Pr{\Event{\text{star}} \cap w_- \le w_1 \le w_+}$. As the marginal connection probability of two vertices $u,v$ with weights $w_u, w_v$ is $\min\{\lambda w_uw_v, 1\}$ we get  
    \begin{align*}
        P &\le \int_{w_-}^{w_+}\! \int_{w_1}^\infty\!\!\! \dots\! \int_{w_1}^\infty \frac{\lambda^{k-1} w_1^{k-1}w_2\ldots w_k}{ n^{k-1}} \rho_{w_1}(w_1)\rho_{w_2\mid w_1}(w_2) \ldots \rho_{w_k\mid w_1}(w_k) \d w_{k} \ldots \d w_1\\
        &= \left(\frac{\lambda}{n}\right)^{k-1} \int_{w_-}^{w_+} w_{1}^{k-1} \rho_{w_1}(w_{1}) \left( \int_{w_1}^{\infty} w_2 \rho_{w_2 \mid w_1}(w_2) \d w_2 \right)^{k-1} \d w_{1}.
    \end{align*}
    By \cref{lem:distributions}, we have \begin{align}
        \int_{w_1}^{\infty} w_2 \rho_{w_2\mid w_1}(w_2) \d w_2 = \frac{\beta-1}{w_1^{1-\beta}} \int_{w_1}^{\infty} w_2^{1-\beta} \d w_2 = \frac{\beta -1}{\beta -2} w_1,
    \end{align} 
    and therefore, our expression for $P$ simplifies to \begin{align*}
    P &\le \left(\frac{\lambda}{n}\right)^{k-1} \int_{w_-}^{w_+} w_{1}^{k-1} \rho_{w_1}(w_{1}) \left( \frac{\beta -1}{\beta -2} w_1 \right)^{k-1} \d w_{1}\\
    &= \left(\frac{\lambda}{n}\right)^{k-1} \left( \frac{\beta -1}{\beta -2} \right)^{k-1} \frac{(\beta-1)k}{w_0^{(1-\beta)k}} \int_{w_-}^{w_+} w_{1}^{k(3-\beta) - 3} \d w_1\\
    &= \mathcal{C} \left(\frac{\lambda}{n}\right)^{k-1} \left( \frac{\beta -1}{\beta -2} \right)^{k-1} \left( w_+^{k(3-\beta)-2} - w_-^{k(3-\beta) - 2}
    \right) .
    \end{align*}
    as desired. 
\end{proof}

\begin{corollary}\label{cor:upperbound}
  Let $G$ be any GIRG with a power-law weight distribution with
  exponent $\beta > 2$, let $U_k = \{v_1, \ldots, v_k\}$ be a set of
  $k$ random vertices, and assume that $v_1$ is of minimal weight
  $w_1$ among $U_k$.  Furthermore, let $\Event{\text{star}}$ be the
  event that $U_k$ is a star with center $v_1$.  Then,
  \begin{align*}
    \cliqueprobk \le \Pr{\Event{\text{star}}} \le \begin{cases}
                                                    \Theta(1)^k n^{\frac{k}{2}(1-\beta)} & \text{if } k > \frac{2}{3-\beta} \text{ and } 2 < \beta < 3 \\
                                                    \Theta(1)^k n^{1-k} & \text{otherwise. }
                                                  \end{cases}
  \end{align*}
\end{corollary}
\begin{proof}
    Set $w_+ = \sqrt{n/\lambda}$ and observe that \begin{align*}
        \Pr{\Event{\text{star}}} = \Pr{\Event{\text{star}} \cap w_0 \le w_1 \le w_+ } + \Pr{\Event{\text{star}} \cap w_1 \ge w_+}.
    \end{align*} Note that, if $w_1 \ge w_+ = \sqrt{n/\lambda}$, the formation of a clique (and thus a star) is guaranteed and hence, \begin{align*}
        \Pr{\Event{\text{star}} \cap w_1 \ge w_+} = \Pr{w_1 \ge w_+} = \left(\frac{w_+}{w_0}\right)^{k(1-\beta)} = \Theta(1)^k n^{\frac{k}{2}(1-\beta)}.
    \end{align*} To bound $\Pr{\Event{\text{star}} \cap w_0 \le w_1 \le w_+ }$, we use \cref{lem:starbound} and observe that 
    \[\mathcal{C} = \frac{(\beta-1)k w_0^{-(1-\beta)k}}{(3-\beta)k-2}\]
    is positive if and only if $k > \frac{2}{3-\beta}$ and $2 < \beta < 3$. In this case \cref{lem:starbound} implies that $\Pr{\Event{\text{star}}} \le \Theta(1)^k n^{\frac{k}{2}(1-\beta)}$ as desired. Otherwise, we get $\Pr{\Event{\text{star}} \cap w_0 \le w_1 \le w_+ } \le \Theta(1)^k n^{1-k}$, which dominates $\Pr{\Event{\text{star}} \cap w_1 \ge w_+} = \Theta(1)^k n^{\frac{k}{2}(1-\beta)}$.
\end{proof}

The above lemma shows that there is a phase transition at
$k = \frac{2}{3-\beta}$ if $2 < \beta < 3$ as previously observed by
Michielan and Stegehuis \cite{Michielan_Stegehuis_2022}. We remark
that our bound is independent of the geometry and also works in the
temperate variant of the model, i.e., in the non-threshold case.

\paragraph{A Lower Bound on $\cliqueprobk$}\label{sec:lowerboundspk}
To obtain a matching lower bound, we employ a similar strategy that
yields the following lemma.
\begin{lemma}\label{lem:starupperbound}
  Let $G = \GIRG(n, \beta, w_0, d)$ be a standard GIRG with any
  $L_p$-norm and let $w_0 \le w_+ \le \sqrt{n/\lambda}$.  Then,
  \begin{align*}
    \Pr{(U_k \text{ is clique}) \cap (w_0 \le w_1 \le w_+) } \ge 2^{-d(k-1)} \mathcal{C} \left(\frac{\lambda}{n}\right)^{k-1}\!\! \left( w_+^{k(3-\beta)-2} - w_0^{k(3-\beta) - 2} \right).
  \end{align*}
  with
  \begin{align*}
    \mathcal{C} \coloneqq \frac{(\beta-1)k w_0^{-(1-\beta)k}}{(3-\beta)k-2}.
  \end{align*}
\end{lemma}
\begin{proof}
    To get a lower bound on $\Pr{U_k \text{ is clique} \cap w_0 \le w_1 \le w_+ }$, we consider the event that every $u \in U_k \setminus \{ v_1 \}$ is placed at a distance of at most $t_{uv_1}/2$ from $v_1$. Then, by the triangle inequality, for any $u, v \in U_k \setminus \{ v_1 \}$, we may bound the distance $d(u,v)$ as \begin{align*}
        d(u,v) \le \frac{1}{2}t_{uv_1}+\frac{1}{2}t_{vv_1} \le t_{uv}
    \end{align*} because $w_1 \le w_v, w_u$. Hence, $u$ and $v$ are adjacent. The probability that a random vertex $u$ is placed at distance of at most $t_{uv}/2$ of $v$ is equal to the volume $\vol(r)$ of the ball of radius $r = t_{uv}/2$ (but at most $1$), i.e., \begin{align*}
        \min\left\{1, \vol(t_{uv}/2)\right\} = \min\left\{1, 2^{-d}\vol(t_{uv})\right\} = \min\left\{1, 2^{-d} \lambda \frac{w_uw_v}{n} \right\}.
    \end{align*} We remark that it is easy to verify that the above term is also a valid lower bound for the probability of the described event if we are working with some $L_p$-norm where $p < \infty$. Conditioned on a value of $w_1$ smaller than $\sqrt{n/\lambda}$, the probability that a vertex $u \in U_k\setminus \{v_1\}$ is placed within distance $t_{uv_1}/2$ of $v_1$ is thus at least $2^{-d}\lambda w_1^2/n$. Thus, \begin{align*}
        \Pr{(U_k \text{ is clique}) \cap (w_0 \le w_1 \le w_+)} &\ge \int_{w_0}^{w_+} \frac{2^{-d(k-1)}\lambda^{k-1} w_1^{2(k-1)}}{n^{k-1}} \rho_{w_1}(w_1) \d w_{1}\\
        &= \left(\frac{2^{-d}\lambda}{n}\right)^{k-1} \frac{(\beta-1)k}{w_0^{(1-\beta)k}} \int_{w_0}^{w_+} w_1^{k(3-\beta) - 3}\d w_1\\
        &= 2^{-d(k-1)} \mathcal{C} \left(\frac{\lambda}{n}\right)^{k-1}\!\! \left( w_+^{k(3-\beta)-2} - w_0^{k(3-\beta) - 2}
        \right),
    \end{align*}
    where the first equality is due to \cref{lem:distributions}.
\end{proof}

\begin{corollary}\label{cor:lowerbound}
  Let $G = \GIRG(n, \beta, w_0, d)$ be a standard GIRG with any
  $L_p$-norm.  Then,
  \begin{align*}
    \cliqueprobk \ge \begin{cases}
                       \Theta(1)^k n^{\frac{k}{2}(1-\beta)} & \text{ if } k > \frac{2}{3-\beta} \text{ and } 2 < \beta < 3 \\
                       \Theta(1)^k 2^{-dk} n^{1-k} & \text{otherwise. }
                     \end{cases}
  \end{align*}
\end{corollary}
\begin{proof}
    The proof is identical to that of \cref{cor:upperbound} with \cref{lem:starupperbound} instead of \cref{lem:starbound}.
\end{proof}

Hence, the asymptotic behavior of our lower bound for $\cliqueprobk$ is the same as that of the upper bound up to a factor of $2^{-d(k-1)} \Theta(1)^k$. We remark that our bounds are easily adaptable to the non-threshold version of the GIRG model as here, it is still guaranteed that a pair of vertices placed within its respective connection threshold is adjacent with a constant probability.

\subsubsection{A Sharper Upper Bound on $q_k$ for Bounded Weights}\label{sec:betterupperboundforboundedweights}

Our upper and lower bounds for the cases $k < \frac{2}{3-\beta}$ and
$\beta \ge 3$ still differ by a factor that is exponential in $d$. In
this section, we prove \cref{thm:sharpcliquebounds}, which we restate
for the sake of readability, and thereby narrow this gap down under
the assumption that the weights of the vertices in $U_k$ are
bounded. While this condition is not always met with high probability
if we choose $U_k$ at random from all vertices, we show how to
leverage it to obtain a better bound on the number of cliques in the
entire graph in the \Cref{sec:expectationbounds}.

\sharpcliquebounds*

Recall that $\tau = 2^d/\lambda$ is a parameter controlling the average degree by influencing the connection threshold. 
We require the condition $c^2 \left( w_1^2/(\tau n) \right)^{1/d} \le 1/4$ to ensure that the maximal connection threshold of any pair of vertices in $U_k$ is so small that we can measure the distance between two neighbors of a given vertex as we would in $\mathbb{R}^d$, i.e., without having to take the topology of the torus into account. We remark that this condition is asymptotically fulfilled as long as $d = o(\log(n))$ and $w_1 = \mathcal{O}(n^{1/2-\varepsilon})$ for any $\varepsilon > 0$.

In the following, we let $\weightsequence{k} = \{w_1, \ldots, w_k\}$ be the sequence of weights of the vertices in $U_k = \{v_1, \ldots, v_k\}$ whereby we assume without loss of generality that $w_1 \le \ldots \le w_k$. We denote by $U_k \text{ is star}$ the event that $U_k$ is a star centered at the vertex of minimum weight (which is $v_1$). In order to prove \Cref{thm:sharpcliquebounds}, we start by showing that $\Pr{U_k \text{ is clique} \mid U_k \text{ is star}, \weightsequence{k}}$ is monotonically increasing in $w_i$ for all $2 \le i \le k$. We remark that this property only holds if we condition on having a star centered at the vertex of minimal weight in $U_k$. With this statement, we may subsequently assume all $u \in U_k\setminus\{v_1\}$ to have a weight of $w_{1}$ and $c^d  w_{1}$ for deriving a lower and upper bound on $\Pr{U_k \text{ is clique} \mid \Event{star}^{c}}$, respectively.

\begin{lemma}\label{lem:monotonicity}
  Let $G = \GIRG(n, \beta, w_0, d)$ be a standard GIRG and denote by $U_k$ is star the event that $U_k$ forms a star centered at the vertex of minimal weight in $U_k$. Let further $U_k = \{v_1, \ldots, v_k\}$ be a set of $k$ random vertices with $w_1 \le \dots \le w_k$.  Then, the conditional probability $\Pr{U_k \text{ is clique} \mid U_k \text{ is star}, \weightsequence{k}}$ is monotonically increasing in $w_2, \ldots, w_k$.
\end{lemma}
\begin{proof}
    For any $1 \le i, j \le k$, denote by $t_{ij}$ the connection threshold $t_{v_iv_j} = \left(\frac{w_iw_{j}}{\tau n}\right)^{1/d}$. In the following, we abbreviate $t_{i1}$ by $t_i$ for $2 \le i \le k$. Note that, since we assume the use of $L_\infty$-norm and condition on $\Event{\text{star}}^c$, the vertex $v_i$ is uniformly distributed in the cube of radius $t_{i}$ around $v_1$ for all $2 \le i \le k$. Thus, all components of $\F{x}_{v_i}$ are independent and uniformly distributed random variables in the interval $[-t_i, t_i]$ (we choose our coordinate system such that $\F{x}_{v+1}$ is the origin). The probability that $U_k$ is a clique conditioned on $\Event{\text{star}}^c$ is hence equal to the probability that the distance between each pair of points is below their respective connection threshold in every dimension. If we denote by $p$ the probability that this event occurs in one fixed dimension, we get that $\Pr{U_k \text{ is clique} \mid \Event{\text{star}}^c, \weightsequence{k}} = p^d$ because all dimensions are independent. Therefore, it suffices to show the desired monotonicity only for $p$. In the following, we therefore only consider one fixed dimension and denote the coordinate of the vertex $v_i$ in this dimension with $x_i$.

    Note that the probability $p$ is equal to \begin{align*}
        p = \int_{-t_2}^{t_2}\int_{-t_3}^{t_3}\cdots \int_{-t_{k}}^{t_{k}} \rho(x_2)\ldots\rho(x_{k}) \mathds{1}(x_2, \ldots, x_{k}) \d x_k \ldots \d x_{2},
    \end{align*}
    where $\rho(x_i) = \frac{1}{2t_i}$ is the density function of $x_i$ as $x_i$ is uniformly distributed in the interval $[-t_i, t_i]$, and $\mathds{1}(x_2, \ldots, x_{k})$ is an indicator function that is $1$ if and only if for all $2 \le i,j \le k$, we have $|x_i - x_j| \le t_{ij}$. We show that $p$ is monotonically increasing in $w_i$ for all $2 \le i \le k$. For this, assume without loss of generality that we increase the weight of $v_2$ by a factor $\xi > 1$. This weight change increases the threshold $t_2$ by a factor of $\xi^{1/d}$ and we denote the connection threshold between $v_i$ and $v_j$ after the weight change by $\tilde{t}_{ij}$. The connection probability $\tilde{p}$ after this weight increases is \begin{align*}
        \tilde{p} = \int_{-\xi^{1/d} t_2}^{\xi^{1/d} t_2}\int_{-t_3}^{t_3}\cdots \int_{-t_{k}}^{t_{k}} \tilde{\rho}(x_2) \rho(x_3)\ldots\rho(x_{k}) \tilde{\mathds{1}}( x_2, \ldots, x_{k}) \d x_k \ldots \d x_{2},
    \end{align*}
    where $\tilde{\rho}(x_2) = \frac{1}{2\xi^{1/d}t_2} = \rho(x_2)/\xi^{1/d}$, and where $\tilde{\mathds{1}}$ is defined like $\mathds{1}$ with the only difference that it uses the new weight of $v_2$, i.e., $\tilde{\mathds{1}}$ is 1 if and only if $|x_i - x_j| \le \tilde{t}_{ij}$ for all $2 \le i,j \le k$.
    Substituting $x_2 = \xi^{1/d} y$, we get \begin{align*}
        \tilde{p} &= \xi^{1/d}\int_{-t_2}^{t_2}\int_{-t_3}^{t_3}\cdots \int_{-t_{k}}^{t_{k}} \tilde{\rho}(\xi^{1/d}y) \rho(x_3) \ldots\rho(x_{k}) \tilde{\mathds{1}}(\xi^{1/d} y, \ldots, x_{k}) \d y \ldots \d x_{k}\\
        &= \int_{-t_2}^{t_2}\int_{-t_3}^{t_3}\cdots \int_{-t_{k}}^{t_{k}} \rho(y) \rho(x_3) \ldots\rho(x_{k}) \tilde{\mathds{1}}(\xi^{1/d} y, \ldots, x_{k}) \d y \ldots \d x_{k}.
    \end{align*}
    
    We claim that $\tilde{p}\geq p$, which we show by proving that $\mathds{1}(y, \ldots, x_{k}) = 1$ implies $\tilde{\mathds{1}}(\xi^{1/d} y, \ldots, x_{k}) = 1$. For this, assume that $y, x_3, \ldots, x_k$ are such that $\mathds{1}(y, \ldots, x_{k}) = 1$. Note that it suffices to show that for all $3 \le i \le k$, if $|x_i - y| \le t_{2i}$, then $|x_i - \xi^{1/d} y| \le \xi^{1/d}t_{2i}$. More formally, we have to show that $d_i \coloneqq |x_i - y| \le t_{2i}$ implies $d_i' \coloneqq |x_i - \xi^{1/d} y| \le \xi^{1/d} t_{2i}$. 
    
    We note that $|y - \xi^{1/d}y| \le \xi^{1/d} t_2 - t_2 = t_2(\xi^{1/d} - 1)$, as $|y|$ is at most $t_2$. Hence, the distance between $v_i$ and $v_2$ increases by at most $t_2(\xi^{1/d} - 1)$ as well. Furthermore, recall that $t_2 \le t_{2i}$ as we assume $w_i \ge w_{1}$.  Accordingly, \begin{align*}
        d_i' &\le d_i + (\xi^{1/d} - 1) t_{2}\\
               &\le t_{2i} + (\xi^{1/d} - 1) t_{2}\\
               &\le t_{2i} + (\xi^{1/d} - 1) t_{2i}\\
               &= \xi^{1/d} t_{2i},
    \end{align*} which finishes the proof.
\end{proof}

Before proceeding with the proof of \Cref{thm:sharpcliquebounds}, we remark that the above statement implies that the entire clique probability (conditional on a given weight sequence) is monotonically increasing in the involved weights. This will be useful in the next section.

\begin{corollary}
  Let $\weightsequence{w} = \{w_1, \ldots, w_k\}$ be a weight sequence with $w_1 \le \ldots \le w_k$. 
  Then, the  probability $\Pr{U_k \text{ is clique} \mid \weightsequence{w}}$ is monotonically increasing in $w_1, \ldots, w_k$.
\end{corollary}\begin{proof}
  Recall that ``$U_k \text{ is star}$'' denotes the event that $U_k$ is a star centered at the vertex of minimal weight in $U_k$ and note that \begin{align*}
    \Pr{U_k \text{ is clique} \mid \weightsequence{w}} = \Pr{U_k \text{ is clique} \mid U_k \text{ is star}, \weightsequence{w}}\Pr{U_k \text{ is star} \mid \weightsequence{w}}.
  \end{align*}
  Hence, if we increase any of the weights $w_2, \ldots, w_k$, then both of the above terms on the right hand side are increasing. If we increase $w_1$, then note that \begin{align*}
    &\Pr{U_k \text{ is clique} \mid \weightsequence{w}} \\
    & \hspace{1cm} = \Pr{U_k \text{ is clique} \mid U_k \setminus \{ v_1 \} \text{ is clique}, \weightsequence{w}} \Pr{U_k \setminus \{v_1\} \text{ is clique} \mid \weightsequence{w}}.
  \end{align*} Here, the second factor remains the same if we change $w_1$ and the first factor can only increase if we increase $w_1$ as -- no matter how $v_2, \ldots, v_k$ arrange to form a clique -- increasing $w_1$ only increases the probability that $v_1$ is adjacent to all of them.
\end{proof}

We go on with calculating $\Pr{U_k \text{ is clique} \mid \Event{\text{star}}^c}$ under the assumption of uniform weights, which afterwards implies \Cref{thm:sharpcliquebounds}.
\begin{lemma}\label{lem:uniformbound}
  Let $G = \GIRG(n, \beta, w_0, d)$ be a standard GIRG, let $U_k$ and
  $\Event{\text{star}}^c$ be defined as in
  \cref{thm:sharpcliquebounds}, and assume that all vertices in
  $u \in U_k \setminus \{v_1\}$ have weight $w_u \ge w_{1}$.
  Then,
  \begin{align*}
    \Pr{U_k \text{ is clique} \mid \Event{\text{star}}^c} = \left(\frac{w_u}{w_1}\right)^{k-2} \left(2^{-(k-1)} \left( (2 - k) \left( \frac{w_u}{w_1} \right)^{1/d} + 2(k-1) \right)\right)^d.
  \end{align*}
\end{lemma}
\begin{proof}
    Again, we only consider one fixed dimension and denote the coordinate of a vertex $u \in U_k$ in this dimension by $x_u$. Note that by \cref{lem:monotonicity}, we may assume that the vertices $v_2, \ldots, v_k$ all have the same weight. We further refer to the connection threshold between any $u \in U_k \setminus \{v_1\}$ and $v_1$ as $t_0$ and to the connection threshold between two vertices in $U_k \setminus \{v_1\}$ as $t_u$. This enables us to set the origin of our coordinate system such that $x_u$ takes values in $[0, 2t_0]$ for all $u \in U_k$, i.e., such that $x_{v_1} = t_0$. Recall that for all $u \in U_k \setminus \{v_1\}$, $x_u$ is a uniformly distributed random variable. 

    We refer to the event that the pairwise distance in the coordinates of all $u,v \in U_k \setminus \{v_1\}$ in our fixed dimension is below the connection threshold $t_{uv}$ as $U_k$ being a \emph{1-D clique}. We calculate $p \coloneqq \Pr{U_k \text{ is 1-D clique} \mid \Event{\text{star}}^c}$ by integrating over the conditional probability $\Pr{U_k \text{ is 1-D clique} \mid \Event{\text{star}}^c, x_{max}}$ where $x_{max}$ is the coordinate of the rightmost vertex (the one with largest coordinate) in $U_k \setminus \{v_1\}$. Note that $U_k$ is a 1-D clique if and only if we have $|x_u - x_{max}| \le t_u$ for all $u \in U_k\setminus \{v_1\}$. We further note that \begin{align*}
            \Pr{U_k \text{ is 1-D clique} \mid \Event{\text{star}}^c, x_{max}} = \begin{cases}
                1 & \text{if } x_{max} \le t_u\\
                \left(\frac{t_u}{x_{max}}\right)^{k-2} & \text{otherwise}.
            \end{cases}
    \end{align*}
    It remains to derive the distribution of $x_{max}$. For this, we derive its density function as \begin{align*}
        \rho(x_{max}) &= \frac{\d \Pr{x_{max} \le x} }{\d x} \\
        &= \frac{\d }{\d x} \left( \Pr{y \le x}^{k-1} \right) \\
        &= \frac{\d }{\d x} \left( \left(\frac{x}{2t_0}\right)^{k-1} \right)\\
        &= \frac{k-1}{(2t_0)^{k-1}}x^{k-2}.
    \end{align*}
    With this, we may deduce \begin{align*}
        \Pr{U_k \text{ is 1-D clique} \mid \Event{\text{star}}^c} &= \int_{0}^{2t_0} \rho(x_{max}) \Pr{U_k \text{ is 1-D clique} \mid x_{max}} \d x_{max}\\
        &= \frac{k-1}{(2t_0)^{k-1}} \left( \int_{0}^{t_u} x^{k-2} \d x + t_u^{k-2} \int_{t_u}^{2t_0} 1 \d x \right)\\
        &= 2^{-(k-1)}\!\left(  \left(\frac{t_u}{t_0}\right)^{k-1} + (k-1)\frac{t_u^{k-2}}{t_0^{k-1}} (2t_0 - t_u) \right) \\
        &= 2^{-(k-1)}\!\left(\!\left(\frac{t_u}{t_0}\right)^{k-1}\!\!\! \mkern-10mu + 2(k-1)\!\left(\frac{t_u}{t_0}\right)^{k-2}\!\!\! \mkern-10mu - (k-1)\!\left(\frac{t_u}{t_0}\right)^{k-1} \right)\\
        &= 2^{-(k-1)}\! \left(\frac{w_u}{w_1}\right)^{\frac{k-2}{d}} \left( (2 - k) \left(\frac{w_u}{w_1}\right)^{1/d} + 2(k-1) \right).
    \end{align*}
    Since $\Pr{U_k \text{ is clique} \mid \Event{\text{star}}^c} = \Pr{U_k \text{ is 1-D clique} \mid \Event{\text{star}}^c}^d$, this finishes the proof.
\end{proof} 

\begin{proof}[Proof of \cref{thm:sharpcliquebounds}]
    We use the monotonicity of $\Pr{U_k \text{ is clique} \mid \Event{\text{star}}^c}$ obtained by \cref{lem:uniformbound}. 
    Due to this monotonicity, it is sufficient to assume that $w_u = w_1$ for all $u \in U_k$ for deriving the lower bound. Hence, we have $w_u/w_1 = 1$ and the expression in \cref{lem:uniformbound} simplifies to \begin{align*}
        2^{-d(k-1)}k^d.
    \end{align*}

    For the upper bound, we instead assume $w_u = c^d w_1$ implying that $ (w_u/w_1)^{1/d} = c$. \cref{lem:uniformbound} implies that \begin{align*}
        \Pr{U_k \text{ is clique} \mid \Event{\text{star}}^c} &\le c^{d(k-2)} 2^{-d(k-1)} \left((2-k)c +2(k - 1)\right)^d\\
        &\le c^{d(k-2)} 2^{-d(k-1)} k^d = \frac{1}{c^d} \left(\frac{c}{2}\right)^{d(k-1)} k^d,
    \end{align*} where we used that $(2-k)c +2(k - 1) \le k$ for all $c \ge 1$ and $k \ge 2$.
\end{proof}

\subsection{Characterizing Cliques by Vertex Weights}\label{sec:cliquesbyvertexweightslowdim}

After establishing bounds on the clique probability in the whole graph, we now turn to characterizing the clique probability in specific parts of the graph in order to prove the statements in \cref{tab:typical-clique-min}. Note that the proofs for the regime $k < \frac{2}{3-\beta}$ and $d = \omega(\log(n))$ are in \cref{sec:clique-weight-high-dim}, all the rest is proven here. 

Recall that $\wmin$ and $\wmax$ are the minimum and maximum weight among $U_k$. Furthermore we assume that $U_k = \{v_1, \ldots, v_k\}$ with associated weights $w_1 \le w_2 \le \ldots \le w_k$. Note that $w_1 = \wmin, w_k = \wmax$. We start by showing that cliques of size $k > \frac{2}{3-\beta}$ dominantly form within the heavy core.

\begin{lemma}\label{lem:sqrtnboundsufficientlylargek}
  Let $k > \frac{2}{3-\beta}$, $\beta \in (2,3)$. Then for any $p \in (0,1)$, there is an $\varepsilon > 0$ such that \begin{align*}
    \Pr{\wmin \in M_\varepsilon(\sqrt{n}) \mid U_k \text{ is clique }} \ge p.
  \end{align*} 
\end{lemma} \begin{proof}
  In the first part, we show that $ \Pr{\wmin < \varepsilon \sqrt{n} \mid U_k \text{ is clique }} \le 1-p$ for some $\varepsilon > 0$. To this end, recall that $U_k = \{ v_1, v_2, \ldots, v_k \}$ is a set of $k$ random vertices with weights $w_1, \ldots, w_k$ whereby we assume without loss of generality that $w_1 \le w_2 \le \cdots \le w_k$. We have \begin{align*}
    \Pr{\wmin < \varepsilon w \mid U_k \text{ is clique }} &= \frac{\Pr{(\wmin < \varepsilon w) \cap (U_k \text{ is clique })}}{\Pr{U_k \text{ is clique }}}
  \end{align*} By \Cref{cor:lowerbound}, we get that there is a constant $c_1 > 0$ such that $\Pr{U_k \text{ is clique}} \ge c_1^k n^{\frac{k}{2}(1-\beta)}$. By \Cref{lem:starbound}, where we set $w_-=w_0$, we have \begin{align*}
    \begin{split}
    &\Pr{(\wmin < \varepsilon w) \cap (U_k \text{ is clique})} \\ 
    &  \hspace{1cm}\le \frac{(\beta - 1)kw_0^{-(1-\beta)}}{(3-\beta)k - 2} \left( \frac{\lambda (\beta - 1)w_0^{-(1-\beta)}}{\beta - 2} \right)^{k-1}n^{1-k} (\varepsilon w)^{k(3-\beta) - 2}\\
    &  \hspace{1cm} \le c_2^k n^{1-k} (\varepsilon w)^{k(3-\beta) - 2}
    \end{split}
  \end{align*} for some constant $c_2 > 0$ as $\frac{(\beta - 1)k}{(3-\beta)k - 2}$ is at most a constant for all (potentially superconstant) $k > \frac{2}{3-\beta}$. Hence, \begin{align*}
    \Pr{\wmin < \varepsilon w \mid U_k \text{ is clique }} &\le \frac{c_2^k n^{1-k} (\varepsilon w)^{k(3-\beta) - 2}}{c_1^k n^{\frac{k}{2}(1-\beta)}} \\ &= \varepsilon^{k(3-\beta)-2}\left(\frac{c_2}{c_1}\right)^k \frac{n^{1-k}w^{k(3-\beta)-2}}{n^{\frac{k}{2}(1-\beta)}}.
  \end{align*} Setting $w = \sqrt{n}$ yields \begin{align*}
    \Pr{\wmin < \varepsilon \sqrt{n} \mid U_k \text{ is clique }} \le \varepsilon^{k(3-\beta)-2}\left(\frac{c_2}{c_1}\right)^k = \varepsilon^{-2}\left(\frac{c_2\varepsilon^{3-\beta}}{c_1}\right)^k.
  \end{align*} If $k$ is a constant, we can see that (since $k(3-\beta)-2 > 0$ by our assumption on $k$), choosing an $\varepsilon > 0$ small enough the above probability is at most $1 - p$ as desired. For $k = \omega(1)$, choosing any $\varepsilon < \left( c_1/c_2 \right)^{\frac{1}{3-\beta}}$ yields that the above probability is $o(1)$ and thus shows that our statement holds for sufficiently large $n$.

  It remains to show that also $ \Pr{\wmin > \sqrt{n}/\varepsilon \mid U_k \text{ is clique }} \le 1-p$ for some $\varepsilon > 0$. Here, it suffices to observe that for $\varepsilon < \lambda$ implies that a clique is formed if $\wmin > \sqrt{n}/\varepsilon$. Thus, if $\varepsilon < \lambda$, \begin{align*}
    \Pr{\wmin > \sqrt{n}/\varepsilon \mid U_k \text{ is clique }} &\le \frac{\Pr{\wmin > \sqrt{n}/\varepsilon \cap U_k \text{ is clique }}}{\Pr{U_k \text{ is clique}}} \\
    &\le \frac{\Pr{\wmin \ge \sqrt{n}/\varepsilon}}{\Pr{\wmin \ge \sqrt{n}/\lambda}}\\
    &= (\lambda / \varepsilon)^{k(1-\beta)},
  \end{align*} which approaches $0$ as $\varepsilon \rightarrow 0$.
\end{proof}

The next lemma proves the claimed bounds for $\wmax$ based on the previous result. Note that this also proves what we want for the entire regime $d = \omega(\log(n)), \beta \in (2,3)$ after we establish suitable bounds for $\wmin$ in \cref{sec:clique-weight-high-dim}.

\begin{lemma}
  Assume that for every $p \in (0,1)$ there is some $\delta > 0$ such that \begin{align*}
    \Pr{\wmin \in M_\delta(\sqrt{n}) \mid U_k \text{ is clique }} \ge p.
  \end{align*}
  Then, for every $p \in (0,1)$ there is also some $\varepsilon > 0$ such that \begin{align*}
    \Pr{\wmax \in M_\varepsilon( \sqrt{n} k^{\frac{1}{\beta-1}}) \mid U_k \text{ is clique }} \ge p.
  \end{align*}
\end{lemma} \begin{proof}
  Bound \begin{align*}
    &\Pr{\wmax > \sqrt{n} k^{\frac{1}{\beta-1}} / \varepsilon \mid U_k \text{ is clique }}\\ 
    & \hspace{1cm} \le \Pr{\wmax > \sqrt{n} k^{\frac{1}{\beta-1}} / \varepsilon \mid U_k \text{ is clique } \cap \wmin \ge \delta \sqrt{n}} \\
    &\hspace{1.5cm} + \Pr{\wmin \le \delta \sqrt{n} \mid U_k \text{ is clique}}.
  \end{align*} From the assumption in our statement, we know that the second term is upper bounded by some function $f(\delta)$ that approaches $0$ as $\delta \rightarrow 0$. We bound the first term as follows \begin{align*}
    &\Pr{\wmax > \sqrt{n} k^{\frac{1}{\beta-1}} / \varepsilon \mid U_k \text{ is clique } \cap \wmin \ge \delta \sqrt{n}} \\
    & \hspace{1cm} \le \sum_{i=2}^k \Pr{w_i > \sqrt{n} k^{\frac{1}{\beta-1}} / \varepsilon \mid U_k \text{ is clique } \cap \wmin \ge \delta \sqrt{n}} \\
    & \hspace{1cm} \le k \cdot \frac{\Pr{w_2 > \sqrt{n} k^{\frac{1}{\beta-1}} / \varepsilon \cap U_k  \text{ is clique } \mid \wmin \ge \delta \sqrt{n} }}{\Pr{U_k  \text{ is clique } \mid \wmin \ge \delta \sqrt{n} }}.
  \end{align*}
  Conditional on $\wmin \ge \delta \sqrt{n}$, the $U_k$ is a clique if $U_{k} \setminus \{v_2\}$ is a clique and if $w_2$ is at least $\sqrt{n} / (\lambda \delta)$ because then the connection probability of every pair of vertices is $1$. Hence,
  \begin{align*}
    &\Pr{\wmax > \sqrt{n} k^{\frac{1}{\beta-1}} / \varepsilon \mid U_k \text{ is clique } \cap \wmin \ge \delta \sqrt{n}} \\
      & \hspace{1cm} \le k \cdot \frac{\Pr{U_k \setminus \{v_2\} \text{ is clique } \cap w_2 > \sqrt{n} k^{\frac{1}{\beta-1}} / \varepsilon \mid \wmin \ge \delta \sqrt{n} }}{\Pr{U_k \setminus \{v_2\} \text{ is clique } \cap w_2 \ge \sqrt{n} / (\lambda \delta) \mid \wmin \ge \delta \sqrt{n} }}\\
    & \hspace{1cm} = k \cdot \frac{\Pr{w_2 > \sqrt{n} k^{\frac{1}{\beta-1}} / \varepsilon \mid \wmin \ge \delta \sqrt{n} }}{\Pr{w_2 \ge \sqrt{n} / (\lambda \delta) \mid \wmin \ge \delta \sqrt{n} }}
  \end{align*} where the last step is due to independence. By the definition of the Pareto distribution, this is at most
  \begin{align*}
    k \cdot \left( \frac{\sqrt{n} k^{\frac{1}{\beta-1}} / \varepsilon}{\sqrt{n}/(\delta \lambda)} \right)^{1-\beta} = (\lambda \delta / \varepsilon)^{1-\beta},
  \end{align*} so in total \begin{align*}
    \Pr{\wmax > \sqrt{n} k^{\frac{1}{\beta-1}} / \varepsilon \mid U_k \text{ is clique }} \le (\lambda \delta / \varepsilon)^{1-\beta} + f(\delta)
  \end{align*} and setting $\delta = \sqrt{\varepsilon}$ yields that this functiontends to zero as $\varepsilon \rightarrow 0$ and the proof of the first part is finished.

  For the second part, bound \begin{align*}
    &\Pr{\wmax < \varepsilon \sqrt{n} k^{\frac{1}{\beta-1}} \mid U_k \text{ is clique}} \le \\
    & \hspace*{1cm} \Pr{\wmax < \varepsilon \sqrt{n} k^{\frac{1}{\beta-1}} \mid U_k \text{ is clique} \cap \wmin \ge \delta \sqrt{n} } \\
    & \hspace*{1cm} + \Pr{ \wmin < \delta \sqrt{n} \mid U_k \text{ is clique} }.
  \end{align*} Again, by the assumption in our statement, we can make the second term arbitrarily small by adjusting $\delta$. For the first term, we bound \begin{align*}
    &\Pr{\wmax < \varepsilon \sqrt{n} k^{\frac{1}{\beta-1}} \mid U_k \text{ is clique} \cap \wmin \ge \delta \sqrt{n} } \\
    &\hspace{1cm} = \Pr{\wmax < \varepsilon \sqrt{n} k^{\frac{1}{\beta-1}} \mid \wmin \ge \delta \sqrt{n} } \\
    &\hspace{2cm} \cdot \frac{ \Pr{U_k \text{ is clique} \mid (\wmax < \varepsilon \sqrt{n} k^{\frac{1}{\beta-1}}) \cap (\wmin \ge \delta \sqrt{n}) } }{\Pr{U_k \text{ is clique} \mid \wmin \ge \delta \sqrt{n} }}.
  \end{align*} It follows by a coupling argument and the fact that the clique probability is monotonically increasing in the weights of the involved vertices that the fraction above is at most $1$. Furthermore, from the definition of the Pareto definition we have\begin{align*}
    \Pr{\wmax < \varepsilon \sqrt{n} k^{\frac{1}{\beta-1}} \mid \wmin \ge \delta \sqrt{n} } &\le \left(1 - \Omega \left( \frac{\varepsilon k^{\frac{1}{\beta-1}}}{\delta}  \right)^{1-\beta}  \right)^{k-1}\\
    & = \exp\left( - \Omega(\varepsilon / \delta)^{1-\beta}  \right).
  \end{align*} Hence in total, \begin{align*}
    \Pr{\wmax < \varepsilon \sqrt{n} k^{\frac{1}{\beta-1}} \mid U_k \text{ is clique}} \le \exp\left( - \Omega(\varepsilon / \delta)^{1-\beta}  \right) + f(\delta)
  \end{align*} where $f$ is a function that tends to $0$ as $\delta \rightarrow 0$. Setting $\delta = \sqrt{\varepsilon}$ yields that this holds for the entire right hand side and finishes the proof. 
\end{proof}

We turn to the regime $d = o(\log(n))$ and $k \le \frac{2}{3-\beta}$ or $\beta > 3$. We start with the following lemma which tells us that here at least one vertex of small weight, i.e., weight in the order of $\exp(\mathcal{O}(1)d) = n^{o(1)}$, is involved in a clique. We afterwards extend this statement to the other vertices involved in a clique.

\begin{lemma}\label{lem:minweightbound} 
  Let $d = o(\log(n))$ and $U_k$ be a set of $k$ random vertices. Let $w_{\min}$ be the minimum weight among $U_k$. If $\beta > 3$ or $k < \frac{2}{3-\beta}$, there is a constant $c > 0$ (independent of $k$) such that for all $p \in (0,1)$ there is an $\varepsilon > 0$ such that \begin{align*} 
    \Pr{w_{\min} \le e^{cd} / \varepsilon \mid U_k \text{ is clique }} \ge p
  \end{align*} 
\end{lemma} \begin{proof}
  Similarly as in the proof of \cref{lem:sqrtnboundsufficientlylargek}, we use \Cref{lem:starbound} to obtain \begin{align*}
    \begin{split}
    &\Pr{(w_1 \ge w / \varepsilon)  \cap (U_k \text{ is clique })} \\
    & \hspace{1cm} \le \frac{(1-\beta)kw_0^{-(1-\beta)}}{(3-\beta)k - 2} \left( \frac{\lambda (\beta - 1)w_0^{-(1-\beta)}}{\beta - 2} \right)^{k-1}n^{1-k} (w/\varepsilon)^{k(3-\beta) - 2}\\
    & \hspace{1cm} \le c_2^k n^{1-k} (w/\varepsilon)^{k(3-\beta) - 2}
    \end{split}
  \end{align*} for some constant $c_2 > 0$.
  By \Cref{cor:lowerbound}, we get that there is a constant $c_1$ such that $\Pr{U_k \text{ is clique}} \ge (c_1 2^{-d})^k n^{1-k}$. Define $\alpha = k(3-\beta) - 2$ and note how this implies \begin{align*}
    \Pr{w_1 \ge w/\varepsilon \mid U_k \text{ is clique }} \le (w/\varepsilon)^{\alpha} \left( \frac{c_2}{c_1 2^{-d}} \right)^{k} = (w/\varepsilon)^{-2} \left( \frac{c_2 (w/\varepsilon)^{3 - \beta}}{c_1 2^{-d}} \right)^{k}.
  \end{align*} Note that due to our assumptions on $\beta$ and $k$, we have $\alpha < 0$. If $\beta \in (2,3)$, we only have to consider the case $k < \frac{2}{3-\beta} = \text{const}$. Hence $\Pr{w_1 \ge w/\varepsilon \mid U_k \text{ is clique }} = c_3 (w/\varepsilon)^{\alpha} 2^{dk}$ for some constant $c_3 > 0$ and setting $w \ge 2^{ \frac{dk}{-\alpha}}$ yields that the above probability is at most $c_3/\varepsilon^\alpha$ as desired (recall that $\alpha < 0$). Note that $c_3$ depends on $k$, however, since we only consider a constant number of different values of $k$ (namely all integers between $3$ and $\frac{2}{3-\beta}$), we may as well choose it independent of $k$ by taking the maximum over all these $k$.

  If $\beta > 3$, let us choose $w =  2^{\max \left\{ 1, \frac{1}{\beta - 3} \right\}d }$. Then, as $w \ge 1$, for any $\varepsilon < \left(c_2/c_1\right)^{\frac{1}{3-\beta}}$, we have $\Pr{w_1 \ge \varepsilon w \mid U_k \text{ is clique }} \le \varepsilon^{-2} c_4^k$ for some constant $c_4 < 1$ and the proof is finished.
\end{proof}

We now proceed by bounding the maximum weight associated to a clique. To this end, we relate the probability that $U_k$ is a clique (assuming $\wmin$ is small) to the probability that $U_{k-1}$ is a clique in the following two lemmas.

\begin{lemma}\label{lem:uppermaxweightbound}
  If $d = o(\log(n))$ and $\beta > 3$, there are constants $a, t_0 > 0$ such that for all $t \ge t_0, w \ge w_0$ and all $k \ge 3$, \begin{align*}
    \Pr{ U_k \text{ is clique} \cap w_\text{max} \ge t \mid w_{\min} \le w} \le \Pr{ U_{k-1} \text{ is clique} \mid w_{\min} \le w } \cdot \frac{a w k t^{2 - \beta}}{n}.
  \end{align*} 
\end{lemma} \begin{proof}
  Note that if we condition on any $w_{\min}$, all vertices in $U_k\setminus \{v_{\min}\}$ are distributed as independent Pareto random variables with parameters $\beta, w_{\min}$. By a union bound, \begin{align*}
    \Pr{ U_k \text{ is clique} \cap w_\text{max} \ge t \mid w_{\min}} &\le \sum_{v \in U_k \setminus \{v_{\min}\}} \Pr{ U_k \text{ is clique} \cap w_v \ge t \mid w_{\min}}.
  \end{align*} Now for any $v \in U_k \setminus \{v_{\min}\}$, \begin{align*}
    \begin{split}
      &\Pr{ U_k \text{ is clique} \cap w_v \ge t \mid w_{\min}}\\ 
      & \hspace{.5cm} = \Pr{ U_k \text{ is clique} \cap w_v \ge t \mid U_k \setminus \{v\} \text{ is clique}, w_{\min}} \Pr{U_k \setminus \{v\} \text{ is clique} \mid w_{\min}} \\
      & \hspace{.5cm} = \Pr{ U_k \text{ is clique} \cap w_v \ge t \mid U_k \setminus \{v\} \text{ is clique}, w_{\min}} \Pr{U_{k-1} \text{ is clique} \mid w_{\min}}.
    \end{split}
  \end{align*} Hence, it remains to bound the first factor above. To this end, we consider the necessary event that $w_v$ is adjacent to $v_{\min}$ and obtain \begin{align*}
    \Pr{ U_k \text{ is clique} \cap w_v \ge t \mid U_k \setminus \{v\} \text{ is clique}, w_{\min}} &\le \Pr{ v \sim v_{\min} \cap w_v \ge t \mid w_{\min}}\\
    &\le \frac{\lambda w_{\min}}{n} \int_{t}^\infty c w^{1-\beta} \d w\\
    &= \frac{a w_{\min} t^{2-\beta}}{n}
  \end{align*} where $c, a $ are constants. Summing over all $v$ yields the desired statement.
\end{proof}

\begin{lemma}\label{lem:lowermaxweightbound}
  Let $d = o(\log(n))$. Then there are constants $a, c> 0$ such that for any $w \ge w_0$, 
  \begin{align*}
    \Pr{U_k \text{ is clique} \mid w_{\min} \le w } \ge \Pr{U_{k-1} \text{ is clique} \mid w_{\min} \le w } \cdot \frac{a w^{1-\beta} e^{-cd}}{n}.
  \end{align*}
\end{lemma} \begin{proof}
  Fix any $v \in U_k \setminus \{v_{\min}\}$ and observe \begin{align*}
    \begin{split}
      \Pr{U_k \text{ is clique} \mid w_{\min} \le w } 
      & = 
      \Pr{U_k \text{ is clique} \mid U_{k} \setminus \{v\} \text{ is clique} \cap w_{\min} \le w } \\
      & \hspace{1cm} \cdot \Pr{U_{k} \setminus \{v\} \text{ is clique }\mid w_{\min} \le w},
    \end{split}
  \end{align*} 
  so it remains to find a lower bound for the first factor. We consider the event that $w_k$ is sufficiently large such that it is sufficiently likely that $v_k$ is placed close enough to $v_\text{min}$ to be connected to all the other vertices by the triangle inequality. We note that, if $v$ is placed within distance 
  \begin{align*}
    \phi \coloneqq \left(\frac{w_v\wmin}{\tau n} \right)^{1/d} - \left(\frac{\wmin^2}{\tau n} \right)^{1/d} = 
    \left(\frac{ w_\text{min}}{\tau n}\right)^{1/d} \left( w_v^{1/d} - w_\text{min}^{1/d} \right)  
  \end{align*}
  of $v_{\min}$ then it must also be adjacent to all the other vertices. Assuming that $w_v \ge \alpha w_\text{min}$ for some $\alpha > 0$, this event occurs with probability 
  \begin{align*}
    \phi^d = \frac{\lambda w_\text{min}}{n} \left( w_k^{1/d} - w_\text{min}^{1/d} \right)^d \ge \frac{\lambda w_0^2}{n} \left( \alpha^{1/d} - 1 \right)^d \ge \frac{\lambda w_0^2}{n} \left( \frac{\ln(\alpha)}{d} \right)^{d}.
  \end{align*} where in the last step we used the inequality $e^x \ge 1 + x$. Choosing $\alpha = e^d$ yields that the above probability is at least $\lambda w_0^2/n$. Furthermore, the probability that $w_v \ge \alpha w_\text{min} = e^d\wmin$ is 
  \begin{align*}
    \left(\frac{\alpha w_\text{min}}{w_0} \right)^{1-\beta} \ge \left( \frac{we^d}{w_0} \right)^{1-\beta}. 
  \end{align*}
  In total, \begin{align*}
    \Pr{U_k \text{ is clique} \mid U_{k} \setminus \{v\} \text{ is clique}\cap w_{\min} \le w } \ge \frac{\lambda w_0^2}{n} \left( \frac{we^d}{w_0} \right)^{1-\beta} = \frac{a w^{1-\beta} e^{-cd}}{n}
  \end{align*} for some constants $a, c > 0$.
\end{proof}

Using the previous two lemmas, we can bound the maximum weight associated to a clique as follows.
\begin{lemma}\label{lem:lowweightdomination}
  If $d = o(\log(n))$ and $\beta > 3$, there is a constant $c > 0$ such that for any $p \in (0,1)$, there is an $\varepsilon > 0$ such that \begin{align*}
    \Pr{w_{\max} \le e^{c d } k^{\frac{1}{\beta - 2}} / \varepsilon \mid U_k \text{ is clique } } \ge p.
  \end{align*}
\end{lemma} \begin{proof}
  Observe that for any $\delta, \alpha$, \begin{align*}
    \begin{split}
    &\Pr{w_{\max} \ge e^{c d } k^{\frac{1}{\beta - 2}} / \varepsilon \mid U_k \text{ is clique } } \\ & \hspace{1cm} \le \Pr{w_{\max} \ge e^{c d} k^{\frac{1}{\beta - 2}} / \varepsilon \mid U_k \text{ is clique} \cap w_{\min} \le e^{\alpha d} / \delta } \\ 
    & \hspace{2cm} + \Pr{w_{\min} \ge e^{\alpha d} / \delta \mid U_k \text{ is clique }}
    \end{split}
  \end{align*} 
  Now by \cref{lem:minweightbound}, if we choose a suitable constant $\alpha$, we can make the second term arbitrarily small by choosing $\delta$ small enough. To bound the first term, observe further that \begin{align*}
    &\Pr{w_{\max} \ge e^{cd} k^{\frac{1}{\beta - 2}} / \varepsilon \mid U_k \text{ is clique} \cap w_{\min} \le  e^{\alpha d} / \delta } \\
    & \hspace{1cm} = \frac{\Pr{U_k \text{ is clique } \cap w_{\max} \ge e^{c d} k^{\frac{1}{\beta - 2}} / \varepsilon \mid w_{\min} \le e^{\alpha d} / \delta }}{\Pr{U_k \text{ is clique } \mid w_{\min} \le e^{\alpha d} / \delta}} 
  \end{align*} By \cref{lem:uppermaxweightbound}, there is a constant $a_1$ such that the numerator is bounded from above by \begin{align*}
    \frac{a_1}{n} \left( e^{\alpha d} / \delta \right) k \left( e^{cd} k^{\frac{1}{\beta - 2}}  / \varepsilon \right)^{2 - \beta} = \frac{a_1}{n} \frac{e^{\alpha d + (2-\beta) c d } }{\delta \varepsilon^{2 - \beta}}.
  \end{align*} Similarly, by \cref{lem:lowermaxweightbound}, the denominator is bounded from below by \begin{align*}
    \frac{a_2}{n} e^{-c'd} \left( e^{\alpha d} / \delta \right)^{1 - \beta}
  \end{align*} for some constants $a_2, c'$. Combining this, the fraction is bounded from above by \begin{align*}
    \frac{a_1}{a_2} \exp \left( \alpha d + c' d + (2 - \beta) cd \right) \delta^{-\beta} \varepsilon^{\beta - 2}.
  \end{align*} In total, there are constants $s, a_3, a_4 > 0$ such that \begin{align*}
    \Pr{w_{\max} \ge e^{c d } k^{\frac{1}{\beta - 2}} / \varepsilon \mid U_k \text{ is clique } } \le a_3 \delta^{-\beta} \varepsilon^{\beta - 2} + a_4 \delta^{s}.
  \end{align*}
  Setting $\delta = \varepsilon^{\gamma}$ for any $ 0 < \gamma < \frac{\beta - 2}{\beta}$ then yields that this term tends to $0$ as $\varepsilon$ tends to $0$ and implies the desired statement.
\end{proof}

Finally, we show that the minimum weight associated to a clique is at least of order $k^{\frac{1}{\beta - 2}}$. This essentially follows from the fact that this holds for stars centered at $v_{\min}$; we show additionally that conditioning on a clique only induces a bias towards even larger weights. This implies that our bounds are tight up to a factor $e^{\Theta(1)d} = n^{o(1)}$. 

\begin{lemma}
  Let $d = o(\log(n))$. Assume further that $k < \frac{2}{3-\beta}$ or that $\beta > 3$. Then for every $p \in (0, 1)$ there is an $\varepsilon > 0$ such that \begin{align*}
    \Pr{\wmax \in M_\varepsilon^{(-)}(k^{\frac{1}{\beta - 2}})} \ge p.
  \end{align*}
\end{lemma}\begin{proof}
  In the following -- again -- the event ``$U_k \text{ is star}$'' denotes the event that $U_k$ is a star centered at $v_{\min}$.
  Bound \begin{align*}
    &\Pr{\wmax < \varepsilon k^{\frac{1}{\beta -2}} \mid U_k \text{ is clique}} \\
    & \hspace{.1cm} = \frac{\Pr{\wmax < \varepsilon k^{\frac{1}{\beta - 2}} \mid U_k \text{ is star} } \Pr{ U_k \text{ is clique} \mid U_k \text{ is star} \cap (\wmax < \varepsilon k^{\frac{1}{\beta - 2}}) }}{ \Pr{ U_k \text{ is clique} \mid U_k \text{ is star}  }}.
  \end{align*} It follows by a coupling argument and the fact that the clique probability is monotonically incresing in $w_2, \ldots, w_k$ that \begin{align*}
    \frac{\Pr{ U_k \text{ is clique} \mid U_k \text{ is star} \cap (\wmax < \varepsilon k^{\frac{1}{\beta - 2}}) }}{ \Pr{ U_k \text{ is clique} \mid U_k \text{ is star}  }} \le 1.
  \end{align*} Moreover, by the Pareto distribution and the fact that the edges in a star are all independent, we get that there is a constant $c > 0$ such that \begin{align*}
    \Pr{\wmax < w \mid U_k \text{ is star} \cap (\wmin = x) } &= \left( \frac{ \int_{x}^{w} \frac{cxy}{n}y^{-\beta} \d y }{ \int_{x}^{\infty} \frac{cxy}{n}y^{-\beta} \d y } \right)^{k-1} \\ 
    &= \left( 1 - \Omega(w/x)^{2-\beta} \right)^{k-1}.
  \end{align*} Thus, \begin{align*}
    \Pr{\wmax < \varepsilon k^{\frac{1}{\beta - 2}} \mid U_k \text{ is star} \cap (\wmin = x) } &\le  \left( 1 - \Omega\left(\frac{\varepsilon k^{\frac{1}{\beta - 2}}}{x}\right)^{2-\beta} \right)^{k-1}\\
    &= \exp(-\Omega( (x/\varepsilon)^{\beta-2}) ).
  \end{align*} As $\wmin \ge w_0 $ deterministically, we get \begin{align*}
    \Pr{\wmax < \varepsilon k^{\frac{1}{\beta - 2}} \mid U_k \text{ is star} } &\le \exp(- \Omega(1/\varepsilon)^{\beta - 2}) ,
  \end{align*} which approaches $0$ as $\varepsilon \rightarrow 0$ as desired.
\end{proof}

\subsection{Bouding the Variance of Typical Cliques}\label{sec:variance}

After bounding the expected number of cliques and characterizing the clique probability by vertex weight, we use the gained insights for bounding the variance of cliques restricted to the dominant regimes identified previously. This, together with the results from \cref{tab:typical-clique-min} allows us further to derive concentration bounds on the total number of cliques later in \cref{sec:concentration-bounds}. 

In the following, given a set of admissible weights $M$, we denote by $K_k(M)$ the number of $k$-cliques with vertex weights that are in $M$ and we derive bounds on the variance of $K_k(M_\varepsilon(w))$ for suitable $w$ (recall the definition of $M_\varepsilon$ from \cref{def:M}). 

\begin{lemma}\label{lem:generalvarbound}
  For all  $k \ge 3$ and for any set of weights $M$, \begin{align*}
    \Var{K_k(M)} \le \Expected{K_k(M)} \sum_{\ell = 1}^k \binom{k}{\ell}\Expected{K_{k-\ell}(M)}
  \end{align*}
\end{lemma}
\begin{proof}
  For a set of vertices $A$, we write that ``$A$ is clique in $M$'' if $A$'s vertex weights are in $M$. Using this notation, we have\begin{align*}
    \Expected{K_k(M)^2} = \left( \sum_{A} \Pr{M_1\text{ is clique in }M} \right)
  \end{align*} where the sum goes over all $k$-element subsets of vertices. Accordingly,
  \begin{align*}
    &\Expected{K_k(M)^2} \le \\
        &\hspace{.5cm} \left(\sum_{A_1} \Pr{A_1\text{ is clique in }M} \right)^2 \\
        & \hspace{.5cm} + \sum_{A_1} \sum_{\substack{A_2 \\ |A_1 \cap A_2| \ge 1}} \Pr{(A_1\text{ is clique in }M) \cap (A_2\text{ is clique in }M) } \\
        &\le \Expected{K_k(M)}^2 \\
        & \hspace{.5cm} +  \sum_{A_1} \sum_{\substack{A_2 \\ |A_1 \cap A_2| \ge 1}}  \Pr{A_1\text{ is clique in }M} \Pr{A_2 \setminus A_1 \text{ is clique in }M }.\\
        &= \Expected{K_k(M)}^2 \\
        & \hspace{.5cm} + \binom{n}{k} \Pr{U_k \text{ is clique in }M} \sum_{\ell = 1}^k \binom{k}{\ell} \binom{n}{k-\ell} \Pr{U_\ell \text{ is clique in }M }\\
        &= \Expected{K_k(M)}^2 + \Expected{K_k(M)} \sum_{\ell = 1}^k \binom{k}{\ell}\Expected{K_{k-\ell}(M)}. 
  \end{align*} By $\Var{X} = \Expected{X^2} - \Expected{X}^2$, the proof is finished.
\end{proof}

We further need the following auxiliary lemma.
\begin{lemma}\label{lem:expectationbound}
  For $\beta \in (2,3)$ and all $\ell \ge 1, \ell = o(n)$.
  \begin{align*}
    \Expected{K_\ell(M_\varepsilon^{(-)}(\sqrt{n}))} \le \mathcal{O} \left( \frac{\ell \varepsilon^{1-\beta} }{n^{(3-\beta)/2}} \right) \Expected{K_{\ell + 1}(M_\varepsilon^{(-)}(\sqrt{n}))} 
  \end{align*}
\end{lemma}\begin{proof}
  Observe that \begin{align*}
    \Expected{K_{\ell + 1}(M_\varepsilon^{(-)}(\sqrt{n}))} &= \binom{n}{\ell+1} \Pr{U_{\ell + 1} \text{ is clique in }M_\varepsilon^{(-)}(\sqrt{n})}\\
    &\ge \frac{n-\ell}{\ell+1} \binom{n}{\ell} \Pr{U_{\ell} \text{ is clique in }M_\varepsilon^{(-)}(\sqrt{n})} \Omega\left( (\sqrt{n} / \varepsilon) ^{1-\beta} \right)\\
    &= \Omega \left( \frac{  n^{(3-\beta)/2}}{ \varepsilon^{1 - \beta} \ell} \right) \Expected{K_{\ell}(M_\varepsilon^{(-)}(\sqrt{n}))}. 
  \end{align*} Here, in the penultimate step, we used the fact that $U_{\ell + 1}$ is a clique if $U_\ell$ is a clique and the remaining $\ell+1$-th vertex has a weight of $\Omega\left( (\sqrt{n} / \varepsilon) \right)$. 
\end{proof}

Combining these statements yields that the number of cliques in the dominant regimes established earlier is self-averaging. That is, restricting the number of cliques to either vertices of very low or of very large weight yields a random variable that concentrates well around its expectation

\begin{lemma}\label{lem:concentration-bound-low-weight}
  For all $k \ge 3$ and for any $w > 1$, \begin{align*}
    \Var{K_k(M_\varepsilon^{(+)}(w))} = \mathcal{O}(w/\varepsilon )^{2k}\Expected{K_k(M_\varepsilon^{(+)}(w))}.
  \end{align*} Moreover, \begin{align*}
    \Var{K_k(M_\varepsilon^{(-)}(\sqrt{n}))} \le  \sum_{\ell = 1}^k \mathcal{O}\left( \frac{\varepsilon^{1 - \beta} k^2}{n^{(3-\beta)/2}} \right)^{\ell} \Expected{K_k(M_\varepsilon^{(-)}(\sqrt{n}))}^2.
  \end{align*}
\end{lemma}
\begin{proof}
As the maximum weight in $M_\varepsilon^{(+)}(w)$ is $w/\varepsilon$, the probability that $k-\ell$ vertices form a clique is $\mathcal{O}(w^2/(\varepsilon^2 n))^{k-\ell} $. Thus, the bound from \cref{lem:generalvarbound} implies \begin{align*}
  \Var{K_k(M_\varepsilon^{(+)}(w))} &\le \Expected{K_k(M_\varepsilon^{(+)}(w))} \sum_{\ell = 1}^k \binom{k}{\ell} \binom{n}{k-\ell} \mathcal{O}((w^2/(\varepsilon^2 n))^{k-\ell} \\
  &\le \Expected{K_k(M_\varepsilon^{(+)}(w))} \mathcal{O}(w/\varepsilon )^{2k} \sum_{\ell=1}^k \binom{k}{\ell}\\
  &\le  \Expected{K_k(M_\varepsilon^{(+)}(w))} \mathcal{O}(w/\varepsilon )^{2k}
\end{align*} as $\sum_{\ell=1}^k \binom{k}{\ell} \le 2^k$.

For the second part of the statement, note that by \cref{lem:expectationbound}, we have \begin{align*}
  \Expected{K_{k -\ell}(M_\varepsilon^{(+)}(w))} \le \mathcal{O}\left( \frac{\varepsilon^{1 - \beta} k}{n^{(3-\beta)/2}} \right)^{\ell}  \Expected{K_{k}(M_\varepsilon^{(+)}(w))}.
\end{align*} Accordingly, by \cref{lem:generalvarbound}, \begin{align*}
  \Var{K_k(M_\varepsilon^{(-)}(\sqrt{n}))} &\le \Expected{K_k(M_\varepsilon^{(-)}(\sqrt{n}))}^2 \sum_{\ell = 1}^k \mathcal{O}\left( \frac{\varepsilon^{1 - \beta} k^2}{n^{(3-\beta)/2}} \right)^{\ell}
\end{align*}
\end{proof}

\subsection{Bounds on the Expected Number of Cliques}\label{sec:expectationbounds}

We proceed by turning the results from the previous sections into bounds on the expected number of cliques. 
Recall that \begin{align}\label{eq:expectation}
    \Expected{K_k} = \binom{n}{k}\cliqueprobk = n^k \Theta(k)^{-k}  \cliqueprobk.
\end{align} Based on this relation and our bounds on $\cliqueprobk$, we prove the following.
\begin{theorem}\label{thm:expectationupperbound}
  Let $G = \GIRG(n, \beta, w_0, d)$ be a standard GIRG with dimensionality
  $d = o(\log(n))$.  Then, we have
  \begin{align*}
    \Expected{K_k} = \begin{cases}
                       n \exp\left( -\Theta(1)dk \right) \Theta(k)^{-k} + o(1) & \text{if } k < \frac{2}{3-\beta} \text{ or } \beta \ge 3 \\
                       n^{\frac{k}{2}(3-\beta)} \Theta(k)^{-k} & \text{otherwise},
                     \end{cases}
  \end{align*}
  whereby the statement holds for all (potentially superconstant)
  $k \ge 3$ if $\beta \neq 3$, and for $k = o(\log(n)/d)$ if
  $\beta = 3$.
\end{theorem} \begin{proof}
    Observe that \cref{cor:upperbound} and \cref{cor:lowerbound} directly imply $\Expected{K_k} = n^{\frac{k}{2}(3-\beta)}\Theta(k)^{-k}$. for the case $2 < \beta < 3$ and $k > \frac{2}{3-\beta}$. For the other case, note that \cref{cor:lowerbound} shows that $\Expected{K_k} \ge 2^{-dk} n  \Theta(k)^{-k}$ and so it only remains to derive an upper bound, i.e., to show that $\Expected{K_k} \le n \exp(-\Omega(1)dk) \Theta(k)^{-k} + o(1)$.

    To this end, we define the set $W_{\ell}$ to be the set of all $k$-element subsets of vertices among which the minimum weight is at most $\ell$ and we denote by $K_k(W_\ell)$ and $K_k(\overline{W_\ell})$ the random variables denoting the number of $k$-cliques in $W_{\ell}$ and $\overline{W_{\ell}}$, respectively. Clearly, for all $\ell$ and any $\varepsilon > 0$, we have $\Expected{K_k}  = \Expected{K_k(W_\ell)} + \Expected{K_k(\overline{W_\ell})}$. We fix $\ell = n^{1/2 - \varepsilon}$ and start by deriving a bound for $\Expected{K_k(W_\ell)}$.

    Recall that we assume $v_1$ to be of minimal weight $w_1$ among $U_k$ and that we denote by $\Event{\text{star}}$ the event that $U_k$ is a star with center $v_1$. We let $\Event{\text{star}}^{\ell}$ be the event $\Event{\text{star}} \cap \left(U_k \in W_\ell\right)$ and observe that \begin{align*}
        \Expected{K_k(W_\ell)} \le n^{k}\Theta(k)^{-k} \Pr{U_k \text{ is clique} \mid \Event{\text{star}}^\ell} \Pr{\Event{\text{star}}^\ell} .
    \end{align*} We note that $\Pr{\Event{\text{star}}^\ell} \le \Pr{\Event{\text{star}}} = \Theta\left( 1 \right)^k n^{-(k-1)}$ as established in \cref{cor:upperbound}. It thus remains to show $\Pr{U_k \text{ is clique} \mid \Event{\text{star}}^\ell} \le \mathcal{O}(1)^k\exp\left(-\Omega(1)dk\right)$.
    
    To show this, we apply \cref{thm:sharpcliquebounds}. However, for this, we need to ensure that the weights $w_2, \ldots, w_k$ are at most $w_1c^d$ for some constant $c > 1$, and that the maximal connection threshold $c^2(\frac{w_1^2}{\tau n})^{1/d}$ is at most $1/4$. As $w_1 \le \ell = n^{1/2-\varepsilon}$ and $d = o(\log(n))$, the second condition is fulfilled for large enough $n$ and every $\varepsilon > 0$. Regarding the first condition, however, we observe that some of the vertices in $U_k$ might have a weight larger than $w_1c^d$. Yet, it is very unlikely that this occurs for many vertices of $U_k$. We split $U_k$ into the two parts $U_{k}^{(1)}$ and $\overline{U_k^{(1)}}$, such that $U_k^{(1)}$ is the set of vertices in $U_k$ with weight at most $w_1c^d$ and $\overline{U_k^{(1)}} = U_k \setminus U_k^{(1)}$. We set $\maxvar =  \max\{3, \lfloor k/2 \rfloor\}$ and bound \begin{align*}
        \begin{split}
        &\Pr{U_k \text{ is clique} \mid \Event{\text{star}}^\ell} \le \hspace{4cm} \\
        & \hspace{1cm} \Pr{U_k^{(1)} \text{ is clique} \mid \Event{\text{star}}^\ell \cap (|U_k^{(1)}| \ge \maxvar)} \Pr{|U_k^{(1)}| \ge \maxvar \mid \Event{\text{star}}^\ell} + \\
        & \hspace{1cm} \Pr{U_k^{(1)} \text{ is clique} \mid \Event{\text{star}}^\ell \cap (|U_k^{(1)}| < \maxvar)} \Pr{|U_k^{(1)}| <  \maxvar \mid \Event{\text{star}}^\ell}.
    \end{split}
    \end{align*} As probabilities are at most one, we may bound \begin{align}\label{ieq:bound}
        \begin{split}
            &\Pr{U_k \text{ is clique} \mid \Event{\text{star}}^\ell} \le \hspace{4cm} \\
            & \hspace{1cm} \Pr{U_k^{(1)} \text{ is clique} \mid \Event{\text{star}}^\ell \cap |U_k^{(1)}| \ge \maxvar} + \Pr{|U_k^{(1)}| < \maxvar \mid \Event{\text{star}}^\ell}.
        \end{split}
    \end{align} 
    
    By \cref{thm:sharpcliquebounds}, the former probability is $\exp(-\Theta(1)d\maxvar)$ as there is a constant $\xi > 0$ such that $\maxvar \ge \xi k$ for all $k \ge 3$. We thus proceed by bounding the latter term in the above expression.
    
    We observe that, conditioned on $\Event{\text{star}}^\ell$, the weights $w_2, \ldots, w_k$ are i.i.d. random variables and so, $|\overline{U_k^{(1)}}|$ is distributed according to the binomial distribution $\text{Bin}(k-1, p)$, where $p = \Pr{w_2 \ge w_1c^d \mid \Event{\text{star}}^\ell}$. 
    \begin{claim}\label{clm:prob}
    There is a constant $b > 1$ such that $p \le b^{-d + \mathcal{O}(1)}$. 
    \end{claim}
    We defer the proof of this claim and proceed by showing how it helps in our main proof. Due to the binomial nature of $|\overline{U_k^{(1)}}|$, we get \begin{align*}
        \Pr{|U_k^{(1)}| < \maxvar \mid \Event{\text{star}}^\ell} &= \Pr{|\overline{U_k^{(1)}}| > k - \maxvar \mid \Event{\text{star}}^\ell}\\
        &= \sum_{i = k - \maxvar + 1}^{k-1} \binom{k-1}{i} p^{i}(1-p)^{k-1-i} \\
        &\le \sum_{i = k - \maxvar + 1}^{k-1} \binom{k-1}{i} b^{-di + \mathcal{O}(1)i}\\
        &\le b^{-d(k - s + 1) + \mathcal{O}(1)k } \sum_{i = k - \maxvar + 1}^{k-1} \binom{k-1}{i}\\
        &\le b^{-d(k - s + 1) + \mathcal{O}(1)k} \\
    \end{align*} where the last step is due to the fact that $\sum_{i = k - \maxvar + 1}^{k-1} \binom{k-1}{i} \le 2^k$. Now, observe that there is a constant $\xi > 0$ such that $ k-\maxvar+1 \ge \xi k$ for all $k \ge 3$. Thus, the above expression is $\exp(-\Omega(1)dk)\mathcal{O}(1)^k$ and so $\Expected{K_k(W_\ell)} = n\exp(-\Omega(1)dk)\mathcal{O}(k)^{-k}$. It remains to prove \cref{clm:prob}.

    \begin{proof}[Proof of \cref{clm:prob}]
        Fix a vertex $u \in U_k \setminus \{v_1\}$ and observe that
        \begin{align*}
            p = \Pr{w_u \ge w_1c^d \mid (u \sim v_1) \cap (U_k \in W_\ell)}   
        \end{align*}
        because the weight of each $u \in U_k \setminus \{v_1\}$ is an i.i.d.\ random variable conditioned on $\Event{\text{star}}^\ell = \Event{\text{star}} \cap U_k \in W_\ell$ and is, in particular, independent of whether other vertices in $U_k$ are adjacent to $v_1$.
        Now, assume that $w_1 = x \le \ell = n^{1/2 - \varepsilon}$ and observe that \begin{align}\label{eq:cond}
            \Pr{w_u \ge w_1 c^d \mid (u \sim v_1) \cap (w_1 = x)} = \frac{\Pr{(w_u \ge w_1 c^d) \cap (u \sim v_1) \mid w_1 = x}}{\Pr{u \sim v_1 \mid w_1 = x}}.
        \end{align}
        Note that for all $x \le n^{1/2- \varepsilon}$, \begin{align*}
            \Pr{u \sim v_1 \mid w_1 = x} \ge \int_{x}^{\frac{n}{\lambda x}} \frac{\lambda x w_u}{n} \rho(w_u) \d w_u &= \frac{\lambda x}{n} \Theta(1) \left(x^{2-\beta} - \left(\frac{n}{\lambda x}\right)^{2-\beta} \right)\\
            &= \frac{\lambda x^{3-\beta}}{n} \Theta(1)\left(1 - \left(\frac{n}{\lambda x^2}\right)^{2-\beta} \right)\\
            &\ge \left(1 - o(1) \right) \frac{\lambda x^{3-\beta}}{n} \Theta(1),
        \end{align*} where the last step follows from the fact that $x \le n^{1/2-\varepsilon}$. On the other hand, we get \begin{align*}
            \Pr{(w_u \ge w_1 c^d) \cap (u \sim v_1) \mid w_1 = x} \le \int_{xc^d}^{\infty} \frac{\lambda x w_u}{n} \rho(w_u) \d w_u &= \frac{\lambda x}{n} \Theta(1) (xc^d)^{2-\beta}\\
            &= \frac{\lambda x^{3-\beta}}{n} \Theta(1) c^{(2-\beta)d}.
        \end{align*} Thus, we get from \cref{eq:cond} that \begin{align*}
            \Pr{w_u \ge w_1 c^d \mid (u \sim v_1) \cap (w_1 = x)} \le (1+o(1))\Theta(1) c^{(2-\beta)d} = b^{-d + \mathcal{O}(1)},
        \end{align*} if we choose $b = c^{\beta - 2}$, which is greater than $1$, since $c > 1$ and $\beta > 2$.
    \end{proof}
    
    It remains to bound $\Expected{K_k(\overline{W_\ell})}$. We observe that \begin{align*}
        \Expected{K_k(\overline{W_\ell})} \le n^k \Theta(k)^{-k}\!\! \left( \Pr{(U_k \text{ is clique})\! \cap\! (\ell \le w_1 \le \sqrt{n/\lambda})} + \Pr{w_1 \ge \sqrt{n/\lambda}} \right)
    \end{align*} As in the proof of \cref{cor:upperbound}, we have \begin{align*}
        \Pr{w_1 \ge \sqrt{n/\lambda}} \le \Theta(1)^k  n^{\frac{k}{2}(3-\beta) - k}
    \end{align*} and from \cref{lem:starbound}, we get \begin{align*}
        \Pr{(U_k \text{ is clique}) \cap (\ell \le w_1 \le \sqrt{n/\lambda})} &\le \Theta(1)^k n^{-(k-1)}n^{(1/2 - \varepsilon) (k(3-\beta) - 2)}\\  &\le \Theta(1)^k n^{-(k-1)-\delta}
    \end{align*} for some $\delta \ge (1/2 - \varepsilon) (3(\beta-3) + 2) > 0$ since we have $k \ge 3$ and $k(\beta-3) + 2 > 0$ for $k < \frac{2}{3-\beta}$ or $\beta \ge 3$. This implies \begin{align*}
        \Expected{K_k(\overline{W_\ell})} \le n^{1-\delta} \Theta(k)^{-k}.
    \end{align*} Now, we distinguish three cases. In the first case, $2 < \beta < 3$ and $k < \frac{2}{3-\beta}$. Here, as $k$ is at most a constant, $n^{-\delta}$ is asymptotically smaller than $\exp(-\Theta(1)dk)$ (recall that $d = o(\log(n))$), which finishes the proof for this case. If $\beta = 3$, recall that we only have to show the statement for $k = o(\log(n)/d)$ and under this assumption, again, $n^{-\delta}$ is asymptotically smaller than $\exp(-\Theta(1)dk)$. For the case $\beta > 3$, recall that $\delta$ is at least a constant since $ \delta \ge (1/2 - \varepsilon) (3(\beta-3) + 2)$. As $3(\beta-3) + 2$ is strictly greater than $2$, we may choose $\varepsilon > 0$ sufficiently small to get $\delta > 1$. Then, $\Expected{K_k(\overline{W_\ell})} = o(1)$ for all $k \ge 3$, which finishes the proof.
\end{proof}

Summarizing our results on $\Expected{K_k}$, we see that in the case $2 < \beta < 3$, there is a phase transition at $k = \frac{2}{3-\beta}$, which was previously observed (in the case of constant $d$) by Michielan and Stegehuis \cite{Michielan_Stegehuis_2022} and also by Bläsius et al. \cite{Blasius_Friedrich_Krohmer_2018} for HRGs. Regarding the influence of $d$, we find that the number of cliques in the regime $k < \frac{2}{3-\beta}$ or $\beta > 3$ decreases exponentially in $dk$.

\subsection{Bounds on the Clique Number}\label{sec:boundscliquenumber}

Now, we turn our bounds on the expected number of cliques into bounds
on the clique number. The results of this section constitute the first
two columns of \cref{tab:clique-number}.

\paragraph*{Upper Bounds}

We start with the upper bounds stated in the first row
of~\cref{tab:clique-number}.

\begin{theorem}
  \label{thm:clique-number-first-row}
  Let $G = \GIRG(n, \beta, w_0, d)$ be a standard GIRG with
  $\beta < 3$.  Then, $\omega(G) = \mathcal{O}(n^{(3-\beta)/2})$
  asymptotically almost surely.
\end{theorem}
\begin{proof}
  We may upper bound the clique number with Markov's inequality, which
  tells us that
  \begin{align*}
    \Pr{\omega(G) \ge k} = \Pr{K_k \ge 1} \le \Expected{K_k}.
  \end{align*} 
  The goal now is to choose $k$ such that
  $\Expected{K_k} \le n^{-\epsilon}$ for some $\epsilon > 0$, as then
  there is no clique larger than $k$ a.a.s.  For $\beta < 3$, which is
  a prerequisite of this theorem, we have
  $\Expected{K_k} \le n^{\frac{k}{2}(3 - \beta)} (c_1k)^{-k}$ for some
  constant $c_1 > 1$ and large enough $k$, according to
  Theorem~\ref{thm:expectationupperbound}.  If we set
  $k = n^{\frac{1}{2}(3 - \beta)}$, we get
  \begin{align*}
    \Expected{K_k} &\le n^{\frac{1}{2}(3 - \beta) n^{\frac{1}{2}(3 - \beta)}} (c_1n^{\frac{1}{2}(3 - \beta)})^{-n^{\frac{1}{2}(3 - \beta)}}\\
                   &= c_1^{-n^{\frac{1}{2}(3 - \beta)}},
  \end{align*}
  which is asymptotically smaller than $n^{-\epsilon}$ because
  \begin{align*}
    c_1^{-n^{\frac{1}{2}(3 - \beta)}} \le n^{-\epsilon} \Leftrightarrow \exp\left( - \log(c_1) n^{\frac{1}{2}(3 - \beta)} \right) \le \exp \left( - \epsilon \log(n) \right),
  \end{align*}
  which holds for $\beta < 3$.
\end{proof}

Regarding our contribution in the second row
of~\cref{tab:clique-number}, we obtain the following theorem.
 
\begin{theorem}
  \label{thm:clique-number-second-row}
  Let $G = \GIRG(n, \beta, w_0, d)$ be a standard GIRG with
  $\beta = 3$ and $d = \mathcal{O}(\log\log(n))$.  Then,
  $\omega(G) = \mathcal{O}(\log(n) / \log\log(n))$ a.a.s.
\end{theorem}
\begin{proof}
  Analogous to the proof of Theorem~\ref{thm:clique-number-first-row},
  we upper bound the clique number using Markov's inequality, i.e.,
  $\Pr{\omega(G) \ge k} \le \mathbb{E}[K_k]$, and choose $k$ such that
  $\mathbb{E}[K_k] \le n^{-\epsilon}$ for some $\epsilon > 0$.  Now
  for $\beta = 3$, we have $\Expected{K_k} \le n \cdot (c_1 k)^{-k}$
  for some constant $c_1 > 1$ and sufficiently large $k$, which follows by \cref{cor:upperbound} and \cref{lem:binomapprox}. It
  remains to determine the value of $k$ for which our desired upper
  bound on $\mathbb{E}[K_k]$ is valid, which can be done by solving
  the following equation for $k$
  \begin{alignat*}{2}
    && n \cdot (c_1 k)^{-k} &= n^{-\epsilon}\\
    \Leftrightarrow \quad && (c_1 k)^{-k} &= n^{-1-\epsilon}.
  \end{alignat*}
  This yields
  \begin{align*}
    k = \frac{1}{c_1}\exp\left(\mathcal{W}(c_1 \log(n^{1+\epsilon}))\right),
  \end{align*}
  where $\mathcal{W}$ is the Lambert $\mathcal{W}$ function defined by
  the identity $\mathcal{W}(z)e^{\mathcal{W}(z)} = z$. With this
  choice of $k$, indeed
  \begin{align*}
    (c_1 k)^{-k} &= \exp\left(- \frac{1}{c_1}\mathcal{W}(c_1 \log(n^{1+\epsilon})) e^{\mathcal{W}(c_1 \log(n^{1+\epsilon}))} \right)\\
                 &= \exp \left( -\log(n^{1+\epsilon}) \right)\\
                 &= n^{-1-\epsilon}.
  \end{align*}
  In order to simplify the expression for $k$, note that for growing
  $z$, we have that $\mathcal{W}(z) = \log(z) - \log \log(z) + o(1)$ and
  thus
  \begin{align*}
    k &= \frac{1}{c_1} \exp\left( \log(c_1(1+\epsilon)\log(n)) - \log\log(c_1 (1+\epsilon)\ln(n)) + o(1) \right)\\
      &= (1+o(1)) \frac{(1+\epsilon)\log(n)}{\log(c_1(1+\epsilon)\log(n))} = \mathcal{O}\left(\frac{\log(n)}{\log\log(n)} \right).
  \end{align*}
\end{proof}

Finally, using a similar argumentation, we obtain the following for
the upper bounds we contribute to the last row
of~\cref{tab:clique-number}.
\begin{theorem}
  Let $G = \GIRG(n, \beta, w_0, d)$ be a standard GIRG with
  $\beta > 3$ and $d = o(\log(n))$.  Then,
  $\omega(G) = \mathcal{O}(\log(n) / (d + \log\log(n)))$
  a.a.s.
\end{theorem}
\begin{proof}
  The proof is analogous to the one of
  Theorem~\ref{thm:clique-number-second-row}.  However, if
  $\beta > 3$, we may use the stronger upper bound from
  \cref{thm:expectationupperbound}, i.e.,
  \begin{align*}
    \Expected{K_k} = n \exp(-\Theta(1)dk)\Theta(k)^{-k} + o(1).
  \end{align*}
  That is, there are constants $c_1, c_2 > 0$ such that
  $\Expected{K_k} = n \left(c_1\exp(c_2d)k\right)^{-k} + o(1)$.  Just
  like before, setting
  \begin{align*}
    k = \frac{1}{c_1\exp(c_2d)}\exp\left(\mathcal{W}(c_1\exp(c_2d) \log(n^{1+\epsilon}))\right)
  \end{align*}
  yields $\Expected{K_k} \le n^{-\epsilon} + o(1)$.  Using the
  asymptotic expansion of the Lambert $\mathcal{W}$ function
  we then obtain
  \begin{align*}
           k &= (1+o(1)) \frac{(1+\epsilon)\log(n)}{\log(c_1\exp(c_2d)(1+\epsilon)\log(n))}\\
             &= (1+o(1)) \frac{(1+\epsilon)\log(n)}{\log(c_1(1+\epsilon)) + c_2d + \log\log(n)}\\
             &= \mathcal{O}\left(\frac{\log(n)}{d + \log\log(n)}\right).
  \end{align*}
\end{proof}

\paragraph*{Lower Bounds}

To get matching lower bounds, we distinguish once more the cases
$\beta < 3 $ and $\beta \ge 3$.
\begin{theorem}
  Let $G = \GIRG(n, \beta, w_0, d)$ be a standard GIRG with
  $\beta < 3$.  Then, $\omega(G) = \Omega(n^{(3-\beta)/2})$
  a.a.s.
\end{theorem}
\begin{proof}
  We show that there are
  $\Omega\left( n^{\frac{1}{2}(3 -\beta)} \right)$ vertices with
  weight at least $\sqrt{n/\lambda}$ w.h.p.  As all these vertices are
  connected with probability $1$, this implies the existence of an
  equally sized clique.

  Because the weight of each vertex is sampled independently, the
  number of vertices with weight above $\sqrt{n/\lambda}$, denoted by
  $X$, is the sum of $n$ independent Bernoulli random variables with
  success probability
  \begin{align*}
    p = \left( \frac{n}{\lambda w_0^2} \right)^{\frac{1}{2}(1-\beta)} = \Theta\left(n^{\frac{1}{2}(1- \beta)}\right),
  \end{align*}
  which we can infer from \cref{eq:prob-pareto}. Therefore, we get
  \begin{align*}
    \Expected{X} = n p = \Theta\left(n^{\frac{1}{2}(3- \beta)}\right)
  \end{align*}
  and by a Chernoff-Hoeffding bound (\cref{thm:chernoff-hoeffding}),
  we get $X \ge \Expected{X}/2 $ w.h.p., which proves our lower bound.
\end{proof}

For the case $\beta \ge 3$, we use the concentration bounds obtained in the previous section applied to a subgraph of $G$ of bounded weight.

\begin{theorem}\label{thm:cliquenumberlowerbound}
  Let $G = \GIRG(n, \beta, w_0, d)$ be a standard GIRG with
  $\beta \ge 3$ and $d = o(\log(n))$.  Then,
  $\omega(G) = \Omega(\log(n) / (d + \log\log(n)))$ a.a.s.
\end{theorem}
\begin{proof}
  Let $w > w_0$ be some fixed weight and let $M = [w_0, w]$. Furthermore, by \cref{lem:starupperbound} there is a constant $c_1 >0 $ such that
  \begin{align*}
    \Expected{ K_k(M) } \ge n \left( c_1 2^d k \right)^{-k}.
  \end{align*} Now, setting \begin{align*}
    k = \frac{1}{2^{d}c_1} \exp\left(\mathcal{W}(2^d c_1 \log(n^{1-\varepsilon}))\right) = \Theta \left( \frac{\log(n)}{d + \log\log(n)} \right)
  \end{align*} yields $\Expected{K_k(M)} \ge n^\varepsilon$. By
  Lemma~\ref{lem:concentration-bound-low-weight} and as $w$ is constant, \begin{align*}
    \Var{K_k(M)} \le \Expected{K_k(M)} \mathcal{O}(1)^{2k}.
  \end{align*} Thus, the inequality of Chebyshev yields \begin{align*}
    \Pr{|K_k(M) - \Expected{K_k(M)}| \ge \frac{1}{2} \Expected{K_k(G_{\le w_c})}} &\le \frac{\Var{K_k(M)}}{\frac{1}{4} \Expected{K_k(M)}^2} = \frac{\mathcal{O}(1)^{2k}}{\Expected{K_k(M)}}.
  \end{align*}
  As $k = \mathcal{O}\left( \log(n)/\log\log(n) \right)$, we have $\mathcal{O}(1)^{2k} = n^{o(1)}$ and so the above term is at most $n^{o(1) - \varepsilon} = o(1)$. Accordingly, $\Pr{K_k(M) \ge 1} = 1-o(1)$ as desired. 
\end{proof}

\section{Cliques in the High-Dimensional Regime}\label{sec:highdim}

Now, we turn to the high dimensional regime, where $d$ grows faster
than $\log(n)$. We shall see that, for constant $k$ and
$d = \omega(\log^2(n))$, the probability that $U_k$ is a clique only
differs from its counterpart in the IRG model by a factor of
$(1 \pm o(1))$.  However, as it turns out, the asymptotic behavior of
cliques in the case $2 < \beta < 3$ is already the same as in the IRG
model if $d = \omega(\log(n))$.  For $\beta \ge 3$, we show that the
number of triangles in the geometric case remains significantly larger
than in the IRG model as long as $d = \log^{3/2}(n)$.
 
\subsection{Bounding the Clique Probability for Fixed Weights}\label{sec:highdimbounds}

We consider the conditional probability that a set
$U_k = \{v_1, \ldots, v_k \}$ of $k$ independent random vertices with
given weights $w_1, \ldots, w_k$ forms a clique. We derive bounds on
this probability under the assumption that $L_\infty$-norm is used,
which afterwards allows bounding the expected number of cliques and
the clique number. In fact, instead of only bounding the probability
that $U_k$ forms a clique, we bound the probability of the more
general event that an arbitrary set of edges $\mathcal{A}$ is formed
among the vertices of $U_k$. We denote by $E(U_k)$ the random variable
representing the set of edges between the vertices in $U_k$ and
proceed by developing bounds on the probability of the event
$E(U_k) \supseteq \mathcal{A}$.

The main difference to our previous bounds is that the connection
threshold proportional to $(w_uw_v/n)^{1/d}$ now grows with $n$
instead of shrinking, even for constant $w_u, w_v$. This requires us
to pay closer attention to the topology of the torus. That is, we have
to take into account that a single dimension of the torus is in fact a
circle with a circumference of $1$.

The bounds are formalized in Theorem~\ref{thm:chung_lu_eq}, which we
restate for the sake of readability.

\cliqueprobabilityhighdim*

This illustrates that the probability that $U_k$ is a clique is at
most $1 + o(1)$ times its counterpart in the IRG model if
$d = \omega\left(\log^{2}(n)\right)$. For
$d = \omega\left(\log(n)\right)$, we get that these two probabilities
only differ by a factor of $(1 + o(1))n^{o(1)}$, which is not much
compared to the case $d = o(\log(n))$, where this factor is at least
in the order of $n^{\binom{k-1}{2} - \frac{d(k-1)}{\log(n)}}$ among
nodes of constant weight.

Before giving the proof, we derive some lemmas that make certain
arguments easier to follow.  We start with an upper bound on the
probability that the set $U_k = \{ v_1, \ldots, v_k\}$ forms a
clique. Before deriving our bound, we need the following auxiliary
lemma.

\begin{lemma}\label{lem:mathineq}
  Let $\ell \in \mathds{N}$, $\ell \ge 1$. There is a constant
  $x_0 < 1$ such that for all $x_0 \le x \le 1$, we
  have \begin{align*} x^{\ell+1} \le \ell x - \ell + 1.
    \end{align*}
\end{lemma} \begin{proof}
    We substitute $x = 1 - y$ and instead show that there is some $y_0 > 0$ such that for all $0 \le y \le y_0$, \begin{align*}
        (1 - y)^{\ell+1} \le 1 - \ell y. 
    \end{align*} We get from a Taylor series that there is a constant $c \ge 0$ such that for all $0 \le y \le 1$, \begin{align*}
        (1 - y)^{\ell+1} \le 1 - (\ell+1)y + cy^2 = 1 - \ell y + cy^2 - y.
    \end{align*} Now, for all $0 \le y \le 1/c$, the term $cy^2 - y$ is negative and our statement follows.
\end{proof}

In the remainder of this section, we frequently analyze events occurring in a single fixed dimension on the torus and use the following notation. Recall that a single dimension of the torus is a circle of circumference $1$, which we denote by $\mathbb{S}^1$. We define the set of points that are within a distance of at most $r$ around a fixed point $\F{x}$ on this circle as $A(r, \F{x})$, and we denote by $\overline{A}(r,\F{x})$ the complement of $A(r, \F{x})$, i.e., the set $\mathbb{S}^1 \setminus A(r, \F{x})$. Observe that $A(r, \F{x})$ and $\overline{A}(r, \F{x})$ are coherent circular arcs. Assume that the position of $v_i$ in our fixed dimension is $\F{x}_{v_i}$. For any pair of vertices $v_i,v_j$, we define the sets $A_{ij} \coloneqq A(t_{v_iv_j}, \F{x}_{v_i})$ and $\overline{A}_{ij} = \overline{A}(t_{v_iv_j}, \F{x}_{v_i})$. We further define $A_i = A(t_0, \F{x}_{v_i})$ and $\overline{A}_i = \overline{A}(t_0, \F{x}_{v_i})$, whereby we note that $\overline{A}_{ij} \subseteq \overline{A}_i$ for all $i,j$ because $t_0$ is the minimal connection threshold.

In the following, we derive upper and lower bounds on the probability that $U_k$ is a clique. 

\begin{theorem}[Upper Bound]\label{thm:prob-high-dim}
  Let $G = \GIRG(n, \beta, w_0, d)$ be a standard GIRG with
  $d = \omega(\log(n))$, let $U_k = \{v_1, \ldots, v_k\}$ be a set of
  $k$ random vertices.  For any constant $k \in \mathbb{N}_{\ge 3}$,
  any $\mathcal{A} \subseteq \binom{U_k}{2}$ and sufficiently large
  $n$, we have
  \begin{align*}
    \Pr{E(U_k) \supseteq \mathcal{A} \mid \{\kappa\}^{(k)}} \le \left( 1 - \left(\frac{\kappa_0}{n}\right)^{r/d} \sum_{\{i,j\}\in \mathcal{A}}\left( 1 - \left( \frac{\kappa_{ij}}{n} \right)^{1/d} \right)\right)^d,
  \end{align*}
  where $r = (3(k-2) + 1)(k-1)$.
\end{theorem}
\begin{proof}
    In the following, we denote by $t_0$ the minimal connection threshold of any two vertices, i.e., $t_0 = \left( \frac{w_0^2}{\tau n}\right)^{1/d}$. Note that $2t_{uv} = \left( \frac{\kappa_{uv}}{n}\right)^{1/d}$ for any pair of vertices $u,v$. We again consider only one fixed dimension of the torus as they are all independent due to our use of $L_\infty$-norm.

    To get an upper bound on the desired probability, we derive a lower bound on $1 - p$. We define the event $\Event{i}$ as the event that $v_i$ falls into $\overline{A}_{ji}$ for some $\{i, j\} \in \mathcal{A}$ with $ i < j$, and the event $\Event{i}^{\text{dis}}$ as the event that $v_i \notin \bigcup_{j=1}^{i-1} \overline{A}_j$ and the sets $\overline{A}_j$ are disjoint for all $1 \le j \le i$. Note that $\Event{i}$ and $\Event{i}^{\text{dis}}$ are disjoint as $\overline{A}_{ji} \subseteq \overline{A}_j$. Then \begin{align}\label{eq:oneminusp}
        1 - p &\ge \Pr{\Event{2}} + \Pr{\Event{3} \cap \Event{2}^{\text{dis}} } + \Pr{\Event{4} \cap \Event{3}^{\text{dis}} \cap \Event{2}^{\text{dis}} } + \ldots \nonumber \\
        &= \sum_{i=2}^k \Pr{\Event{i} \left| \hspace{.1cm} \bigcap_{j=1}^{i-1}\Event{j}^{\text{dis}}}\Pr{\bigcap_{j=1}^{i-1}\Event{j}^{\text{dis}}}\right.
    \end{align}
    Note that this is a valid bound because all the events we sum over are disjoint.
    
    Now, let $\mathcal{A}_{i}^{-} = \{ \{i, j\}\in \mathcal{A} \mid j \le i  \}$ be the set of edges from vertex $i$ to a lower-indexed vertex. If we condition on $\bigcap_{j=1}^{i-1}\Event{j}^{\text{dis}}$, then the probability of $\Event{i}$ is simply \begin{align*}
        \Pr{\Event{i} \left| \hspace{.1cm} \bigcap_{j=1}^{i-1}\Event{j}^{\text{dis}} \right. } =\!\! \sum_{\{i,j\} \in \mathcal{A}_i^{-}} (1 - 2t_{ij}) = \!\!\sum_{\{i,j\} \in \mathcal{A}_i^{-}} \left(1 - \left( \frac{\kappa_{ij}}{n}\right)^{1/d}\right)
    \end{align*} because all the sets $\overline{A}_{ji}$ are disjoint. It remains to bound $\Pr{\bigcap_{j=1}^{i-1}\Event{j}^{\text{dis}}}$. We obtain \begin{align*}
        \Pr{\bigcap_{j=1}^{i-1}\Event{j}^{\text{dis}}} = \prod_{j=1}^{i-1} \Pr{\Event{j}^{\text{dis}} \left| \hspace{.1cm} \bigcap_{\ell=1}^{j-1}\Event{\ell}^{\text{dis}} \right.}.
    \end{align*}
    Now, the probability $\Pr{\Event{j}^{\text{dis}} \mid \bigcap_{\ell=1}^{j-1}\Event{\ell}^{\text{dis}}}$ is equal to the probability that $v_j$ is placed outside of $\overline{A}_{\ell}$ for all $1 \le \ell < j$ while, at the same time, $\overline{A}_{\ell} \cap \overline{A}_j = \emptyset$. If we consider one fixed set $\overline{A}_{\ell}$, we note that this requires $v_j$ to be of distance at least $1 - 2t_0$ from $v_\ell$ as, otherwise,  $\overline{A}_{\ell}$ and $\overline{A}_j$ overlap. Hence, we may define a ``forbidden'' region around $v_{\ell}$ which includes $\overline{A}_{\ell}$ and all points within distance $1 - 2t_0$ of $v_{\ell}$. This region has volume $3(1-2t_0)$ and so the probability that $v_j$ falls outside the forbidden region is at least $1 - 3(1-2t_0)$. We refer the reader to \cref{fig:torus-dim} for an illustration. Now considering the forbidden region of all $v_{\ell}$ with $1 \le \ell < j$, the combined volume of these forbidden regions is at most $3(j-1)(1-2t_0)$ and hence \begin{align*}
        \Pr{\bigcap_{j=1}^{i-1}\Event{j}^{\text{dis}}} = \prod_{j=1}^{i-1} \Pr{\Event{j}^{\text{dis}} \left| \hspace{.1cm} \bigcap_{\ell=1}^{j-1}\Event{\ell}^{\text{dis}} \right.} &\ge \prod_{j=1}^{i-1} (1 - 3(j-1)(1-2t_0))\\
        &\ge (1 - 3(i-2)(1-2t_0))^{i-1}.
    \end{align*} We get from \cref{eq:oneminusp}, \begin{align*}
        1 - p &\ge \sum_{i=2}^k \Pr{\Event{i} \left| \hspace{.1cm} \bigcap_{j=1}^{i-1}\Event{j}^{\text{dis}} \right.}\Pr{\bigcap_{j=1}^{i-1}\Event{j}^{\text{dis}}}\\ &\ge \sum_{i=2}^k (1 - 3(i-2)(1-2t_0))^{i-1}\sum_{\{i,j\} \in \mathcal{A}_i^{-}} \left(1 - \left( \frac{\kappa_{ij}}{n}\right)^{1/d}\right) \\
        &\ge (1 - 3(k-2)(1-2t_0))^{k-1} \sum_{\{i,j\} \in \mathcal{A}} \left(1 - \left( \frac{\kappa_{ij}}{n}\right)^{1/d}\right)\\
        &= \left( 3(k-2)\left(\frac{\kappa_0}{n}\right)^{1/d} - 3(k-2) + 1 \right)^{k-1} \sum_{\{i,j\} \in \mathcal{A}} \left(1 - \left( \frac{\kappa_{ij}}{n}\right)^{1/d}\right).
    \end{align*} 
    It now remains to show \begin{align*}
        3(k-2)\left(\frac{\kappa_0}{n}\right)^{1/d} - 3(k-2) + 1 \ge \left(\frac{\kappa_0}{n}\right)^{\frac{3(k-2) + 1}{d}}
    \end{align*} for sufficiently large $n$. Recalling that $\left(\frac{\kappa_0}{n}\right)^{1/d}$ tends to $1$ as $n$ grows, this is equvalent to show that there is some $x_0 < 1$ such that for all $x_0 \le x \le 1$ and $\ell = 3(k-2)$, we have \begin{align*}
        \ell x - \ell + 1 \ge x^{\ell + 1}.
    \end{align*} This follows by \cref{lem:mathineq} and the proof is finished.
\end{proof}

\begin{figure}
    \centering
    \scalebox{0.6}{
        \includegraphics{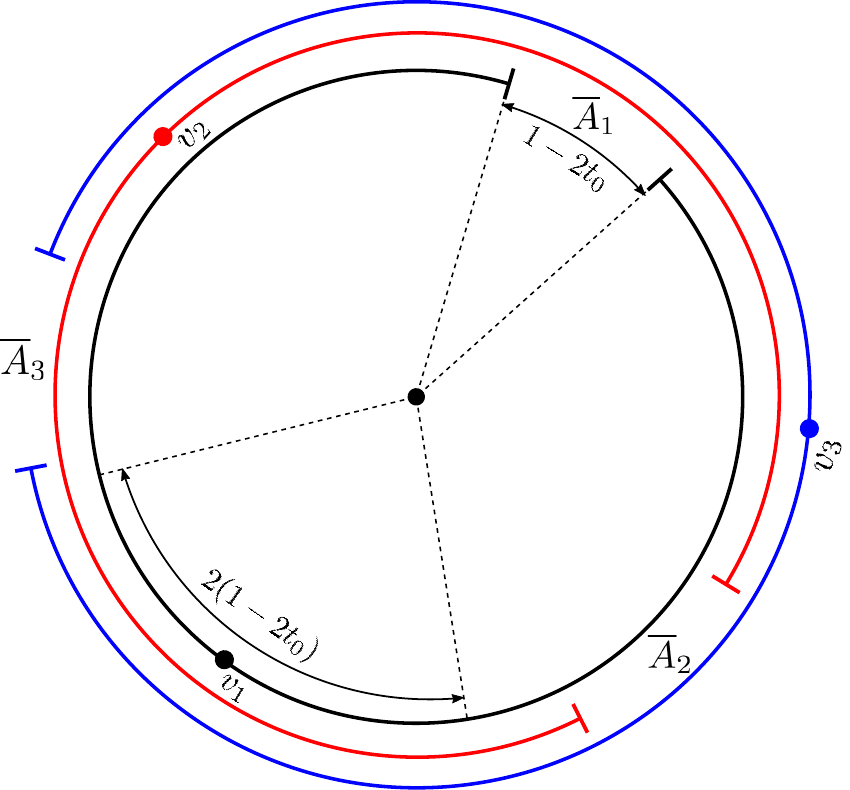}
    }
    \caption{Illustration for the proof of \cref{thm:prob-high-dim}. The colored circles represent the sets $A_i$ for $v_1,v_2,v_3$. The probability that $v_2$ is placed in the region indicated by the black circular arc such that $\overline{A}_1 \cap \overline{A}_2 = \emptyset$ is at least $1 - 3(1-2t_0)$ as indicated by the black arrows.}\label{fig:torus-dim}
\end{figure}  

\begin{theorem}[Lower Bound]\label{thm:lowerboundhighdim}
  Let $G = \GIRG(n, \beta, w_0, d)$ be a standard GIRG with
  $d = \omega(\log(n))$, let $U_k = \{v_1, \ldots, v_k\}$ be a set of
  $k$ random vertices.  Then, for every
  $\mathcal{A} \subseteq \binom{U_k}{2}$,
  \begin{align*}
    \Pr{ E(U_k) \supseteq \mathcal{A} \mid \{\kappa\}^{(k)}} \ge \prod_{i=1}^k \left( 1 - \sum_{\{i,j\} \in \mathcal{A}_i^{-}} \left( 1 - \left(\frac{\kappa_{ij}}{n}\right)^{1/d} \right) \right)^d
  \end{align*}
  where
  $\mathcal{A}_i^{-} = \{ \{i,j\} \in \mathcal{A} \mid j \le i\}$ is
  the set of edges in $\mathcal{A}$ between vertex $i$ and a previous
  vertex.
\end{theorem}\begin{proof}
    We sample the position of $v_2, \ldots, v_k$ one after another and again only consider the probability $p$ that $E(U_k) \supset \mathcal{A}$ in one fixed dimension. When the position of $v_i$ is sampled, the probability that $v_i$ falls into one of $A_{ji}$ for $\{i,j\} \in \mathcal{A}_i$ is at least \begin{align*}
        1 - \sum_{\{i,j\} \in \mathcal{A}_i^{-}} \left( 1 - \left(\frac{\kappa_{ij}}{n}\right)^{1/d} \right),
    \end{align*} where equality holds if the first $i-1$ vertices are placed such that all the sets $\overline{A}_{ji}$ for $1 \le j < i$ are disjoint. Accordingly, we have \begin{align*}
        p \ge \prod_{i=1}^k \left( 1 -\sum_{\{i,j\} \in \mathcal{A}_i^{-}} \left( 1 - \left(\frac{\kappa_{ij}}{n}\right)^{1/d} \right) \right).
    \end{align*} To have $E(U_k) \supset \mathcal{A}$, it is sufficient that this event occurs in all dimensions, so the final bound is $p^d$.
\end{proof}

To see how these bounds behave as $n \rightarrow \infty$, we need the following lemma. 

\begin{lemma}\label{lem:oneminuspsi}
    Let $\Psi, d, \ell$ be non-negative functions of $n$. Assume that $\frac{-\ell\ln(\Psi)}{d} = o(1)$ and $\Psi < 1$ for all $n$. Then, \begin{align*}
        \frac{-\ell\ln(\Psi)}{d} - e\left(\frac{-\ell\ln(\Psi)}{d}\right)^2 \le 1 - \Psi^{\ell/d} \le \frac{-\ell\ln(\Psi)}{d}.
    \end{align*}
\end{lemma} \begin{proof}
    Observe that \begin{align*}
        1 - \Psi^{\ell/d} = 1 - \exp\left(\frac{\ell\ln(\Psi)}{d}\right).
    \end{align*} For the upper bound, we use the well known fact that, for all $x$, we have $\exp(x) \ge 1 + x$, which directly implies our statement. For the lower bound, we use the Taylor series expansion of $\exp$ and bound
    \begin{align}\label{eq:bound}
        1 - \Psi^{\ell/d} = 1 - \exp\left(\frac{\ell\ln(\Psi)}{d}\right) &= -\sum_{i=1}^\infty \frac{1}{i!} \left( \frac{\ell\ln(\Psi)}{d} \right)^i \nonumber \\ &= \sum_{i=1}^\infty \frac{(-1)^{i-1}}{i!} \left( \frac{-\ell\ln(\Psi)}{d} \right)^i \nonumber \\
        &= \frac{-\ell\ln(\Psi)}{d} - \left(\frac{-\ell\ln(\Psi)}{d}\right)^2\sum_{i=2}^\infty\frac{(-1)^{i}}{i!} \left( \frac{-\ell\ln(\Psi)}{d} \right)^{i-2}
    \end{align} Because $\frac{-\ell\ln(\Psi)}{d} = o(1)$, we get that $\left( \frac{-\ell\ln(\Psi)}{d} \right)^{i-2} \le 1$ for all $i \ge 2$ and sufficiently large $n$. Assuming that all the terms in the above sum are positive yields \begin{align*}
        1 - \Psi^{\ell/d} &\ge \frac{-\ell\ln(\Psi)}{d} - \left(\frac{-\ell\ln(\Psi)}{d}\right)^2\sum_{i=0}^\infty \frac{1}{i!}\\
        &= \frac{-\ell\ln(\Psi)}{d} - e\left(\frac{-\ell\ln(\Psi)}{d}\right)^2
    \end{align*} as desired.
\end{proof}

\begin{lemma}\label{lem:lnpupperbound}
    Let \begin{align*}
        P \coloneqq \left( 1 - \left(\frac{\kappa_0}{n}\right)^{r/d} \sum_{\{i,j\} \in \mathcal{A}} \left(1 - \left(\frac{\kappa_{ij}}{n}\right)^{1/d} \right) \right)^d.
    \end{align*} Then there is a constant $\delta \ge 0$ such that \begin{align*}
        \ln(P) \le   \sum_{\{i,j\} \in \mathcal{A}} \ln \left( \frac{\kappa_{ij}}{n} \right) \left( 1 - \frac{\delta \ln(n)}{d} \right).
    \end{align*}
\end{lemma}\begin{proof}
    We have \begin{align*}
        \ln(P) = d \ln \left( 1 - \left(\frac{\kappa_0}{n}\right)^{r/d} \sum_{\{i,j\} \in \mathcal{A}} \left(1 - \left(\frac{\kappa_{ij}}{n}\right)^{1/d} \right) \right).
    \end{align*} We get from a Taylor series that \begin{align*}
        \ln(x) = \sum_{i = 1}^\infty (-1)^{i-1} \frac{(x - 1)^i}{i} = - \sum_{i = 1}^\infty \frac{(1 - x)^i}{i}
    \end{align*} and thus \begin{align*}
        \ln(P) &= -d \sum_{i = 1}^\infty \frac{1}{i} \left(\left(\frac{\kappa_0}{n}\right)^{r/d} \sum_{\{i,j\} \in \mathcal{A}} \left(1 - \left(\frac{\kappa_0}{n}\right)^{1/d} \right) \right)^i\\
        &\le -d\left(\frac{\kappa_0}{n}\right)^{r/d} \sum_{\{i,j\} \in \mathcal{A}} \left(1 - \left(\frac{\kappa_{ij}}{n}\right)^{1/d} \right).
    \end{align*} We now apply the lower bound from \cref{lem:oneminuspsi} with $\Psi = \kappa_{ij}/n$. Note that this fulfills the condition $-\ln(\Psi)/d = o(1)$ as $d = \omega(\log(n))$ and $\kappa_{ij}/n = \Omega(1/n)$. We obtain \begin{align*}
        \ln(P) &\le d\left(\frac{\kappa_0}{n}\right)^{r/d} \sum_{\{i,j\} \in \mathcal{A}} \left( \frac{\ln\left( \frac{\kappa_{ij}}{n}\right) }{d} + \frac{e\ln^2\left( \frac{\kappa_{ij}}{n}\right) }{d^2} \right)\\
        &= \left(\frac{\kappa_0}{n}\right)^{r/d} \sum_{\{i,j\} \in \mathcal{A}} \ln \left( \frac{\kappa_{ij}}{n} \right) \left( 1 - \frac{e}{d}\ln \left( \frac{n}{\kappa_{0}} \right) \right).\end{align*} 
        Note that the above term is negative, so we need to lower bound the term $\left( \frac{\kappa_0}{n} \right)^{r / d}$ to proceed. We get from \cref{lem:oneminuspsi},
        \begin{align*} 
            \left( \frac{\kappa_0}{n} \right)^{r / d} \ge 1 - \frac{r}{d} \ln \left( \frac{n}{\kappa_0} \right) - \frac{er^2}{d^2} \ln \left( \frac{n}{\kappa_0} \right)^2 = 1 - \frac{c\ln(n)}{d}
        \end{align*} for some constant $c > 0$ and sufficiently large $n$. With this, 
        \begin{align*}    
        \ln(P) &\le \sum_{\{i,j\} \in \mathcal{A}} \ln \left( \frac{\kappa_{ij}}{n} \right) \left( 1 - \frac{e}{d}\ln \left( \frac{n}{\kappa_{0}} \right)  \right) \left( 1 - \frac{c\ln(n)}{d} \right) \\
        &\le \sum_{\{i,j\} \in \mathcal{A}} \ln \left( \frac{\kappa_{ij}}{n} \right) \left( 1 - \frac{\delta \ln(n)}{d} \right)
    \end{align*} for some $\delta > 0$ and sufficiently large $n$. 
\end{proof}

\begin{lemma}\label{lem:lowerboundbound}
    Let \begin{align*}
        Q \coloneqq \prod_{i = 1}^k \left( 1 - \sum_{\{i,j\} \in \mathcal{A}_i^{-}} \left( 1 - \left(\frac{\kappa_{ij}}{n}\right)^{1/d} \right) \right)^d.
    \end{align*} Then there is a constant $\delta$ such that \begin{align*}
        \ln(Q) \ge \sum_{\{i,j\} \in \mathcal{A}} \ln\left( \frac{\kappa_{ij}}{n} \right) \left( 1 + \frac{\delta \ln(n)}{d} \right).
    \end{align*}
\end{lemma} \begin{proof}
    By the upper bound from \cref{lem:oneminuspsi} and the Taylor series of $\ln$, we obtain \begin{align*}
        \ln(Q) &\coloneqq d \sum_{i=1}^k \ln \left( 1 - \sum_{\{i,j\} \in \mathcal{A}_i^{-}} \left( 1 - \left(\frac{\kappa_{ij}}{n}\right)^{1/d} \right) \right)\\
        &\ge d \sum_{i=1}^k \ln \left( 1 - \sum_{\{i,j\} \in \mathcal{A}_i^{-}} \left( \frac{\ln\left(\frac{n}{\kappa_{ij}}\right)}{d} \right) \right)\\
        &= -d \sum_{i=1}^k \sum_{\ell=1}^\infty \frac{1}{\ell} \left( \sum_{\{i,j\} \in \mathcal{A}_i^{-}} \left( \frac{\ln\left(\frac{n}{\kappa_{ij}}\right)}{d} \right)\right)^\ell \\
        &= 
        -d \sum_{i=1}^k \sum_{\{i,j\} \in \mathcal{A}_i^{-}} \left( \frac{\ln\left(\frac{n}{\kappa_{ij}}\right)}{d} \right) \sum_{\ell=1}^\infty \frac{1}{\ell} \left( \sum_{\{i,j\} \in \mathcal{A}_i^{-}} \left( \frac{\ln\left(\frac{n}{\kappa_{ij}}\right)}{d} \right)\right)^{\ell - 1} \\
        &= \sum_{\{i,j\} \in \mathcal{A}} \ln\left( \frac{\kappa_{ij}}{n} \right) \sum_{\ell=1}^\infty \frac{1}{\ell} \left( \sum_{\{i,j\} \in \mathcal{A}_i^{-}} \left( \frac{\ln\left(\frac{n}{\kappa_{ij}}\right)}{d} \right)\right)^{\ell - 1}\\
        &\ge \sum_{\{i,j\} \in \mathcal{A}} \ln\left( \frac{\kappa_{ij}}{n} \right) \sum_{\ell=1}^\infty \frac{1}{\ell} \left(  \frac{k\ln\left(\frac{n}{\kappa_{0}}\right)}{d} \right)^{\ell - 1}.
\end{align*} Now, we have \begin{align*}
    \sum_{\ell=1}^\infty\frac{1}{\ell} \left(  \frac{k\ln\left(\frac{n}{\kappa_{0}}\right)}{d} \right)^{\ell - 1} \le 1 + \frac{k\ln\left(\frac{n}{\kappa_{0}}\right)}{d} \sum_{\ell=0}^\infty \left(  \frac{k\ln\left(\frac{n}{\kappa_{0}}\right)}{d} \right)^{\ell} \le 1 + \frac{2k\ln\left(\frac{n}{\kappa_{0}}\right)}{d}
\end{align*} because  $\frac{k\ln\left(\frac{n}{\kappa_{0}}\right)}{d} = o(1)$ and so the above geometric sum converges to $1 + o(1)$. 
\end{proof}

\begin{proof}[Proof of \cref{thm:chung_lu_eq}]
    Combining \cref{thm:prob-high-dim}, \cref{thm:lowerboundhighdim}, \cref{lem:lnpupperbound} and \cref{lem:lowerboundbound}, we get that there are constants $\delta, \delta' > 0$ such that \begin{align*}
        \prod_{\{i,j\} \in \mathcal{A}} \left( \frac{\kappa_{ij}}{n} \right)^{1 + \frac{\delta\ln(n)}{d}} \le \Pr{E(U_k) \supseteq \mathcal{A} \mid \{\kappa\}^{(k)}} \le \prod_{\{i,j\} \in \mathcal{A}} \left( \frac{\kappa_{ij}}{n} \right)^{1 - \frac{\delta'\ln(n)}{d}}.
    \end{align*} If $d = \omega(\log(n))$, $\frac{\delta \ln(n)}{d} = o(1)$ and the proof of this case is finished. If $d = \omega(\log^2(n))$, observe that \begin{align*}
        \prod_{\{i,j\} \in \mathcal{A}} \left( \frac{\kappa_{ij}}{n} \right)^{1 + \frac{\delta\ln(n)}{d}} \ge \left(\frac{n}{\kappa_0}\right)^{-\binom{k-1}{2} \frac{\delta \ln(n)}{d}}
        \prod_{\{i,j\} \in \mathcal{A}} \left( \frac{\kappa_{ij}}{n} \right).\end{align*} Similarly, 
        \begin{align*}
            \prod_{\{i,j\} \in \mathcal{A}} \left( \frac{\kappa_{ij}}{n} \right)^{1 - \frac{\delta'\ln(n)}{d}} \le \left(\frac{n}{\kappa_0}\right)^{\binom{k-1}{2} \frac{\delta' \ln(n)}{d}}
            \prod_{\{i,j\} \in \mathcal{A}} \left( \frac{\kappa_{ij}}{n} \right).\end{align*}
    Since $d = \omega(\log^2(n))$, we get that \begin{align*}
        \left(\frac{n}{\kappa_0}\right)^{\binom{k-1}{2} \frac{\delta \ln(n)}{d}} = (1 + o(1)) n^{1/\omega(\log(n))} = (1 + o(1)) \exp\left( \frac{\log(n)}{\omega(\log(n))} \right) = 1 + o(1)
    \end{align*} as we assume $k$ to  be constant. 
\end{proof}

\paragraph{Bounds for Triangles}\label{sec:boundsfortriangles}

We proceed by deriving a lower bound for the probability that three vertices form a triangle.

\begin{lemma}\label{lem:highdimtriangles}
  Let $G = \GIRG(n, \beta, w_0, d)$ be a standard GIRG and let
  $v_1, v_2, v_3$ be three random vertices.  Then,
  \begin{align*}
    \Pr{v_2 \sim v_3 \mid v_1 \sim v_2, v_3} \ge \left(3 - 3 \left(\frac{\kappa_0}{n}\right)^{-1/d} + \left(\frac{\kappa_0}{n}\right)^{-2/d}\right)^d.
  \end{align*}
\end{lemma}
\begin{proof}
    We consider one fixed dimension and give an upper bound on the probability that $v_2 \nsim v_3$ conditioned on $v_1 \sim v_2, v_3$. In the following, we abbreviate $p \coloneqq \frac{\kappa_0}{n}$. Note that we may assume that all three vertices are of weight $w_0$ as it is easy to verify that larger weights only increase the probability of forming a triangle. Conditioned on the event $v_1 \sim v_2, v_3$, the vertices $v_2,v_3$ are uniformly distributed within a circular arc of length $p^{1/d}$ around $v_1$. In order for $v_2 \nsim v_3$ to occur, $v_3$ needs to be placed within a circular arc of length $1-p^{1/d}$ opposite to $v_2$. Hence, the probability that $v_2 \nsim v_3$ conditioned on the event that $v_2$ is placed at distance $x$ of $v_1$ is \begin{align*}
        \Pr{v_2 \nsim v_3 \mid |x_{v_1} - x_{v_2}|_C = x} = \begin{cases}
            x/p^{1/d} & \text{if } x \le 1 - p^{1/d}\\
            (1-p^{1/d})/p^{1/d} & \text{otherwise}  
        \end{cases}
    \end{align*} where $|\cdot|_C$ denotes the distance on the circle. Since $v_2$ is distributed uniformly within distance $\frac{1}{2}p^{1/d}$ around $v_1$, we obtain \begin{align*}
        \Pr{v_2 \nsim v_3 \mid v_1 \sim v_2, v_3} &\le \frac{1}{\frac{1}{2}p^{1/d}} \int_{0}^{\frac{1}{2}p^{1/d}} \Pr{v_2 \nsim v_3 \mid |x_{v_1} - x_{v_2}|_C = x} \d x \\
        &= \frac{1}{\frac{1}{2}p^{1/d}} \left( \int_{0}^{1-p^{1/d}}\frac{x}{p^{1/d}} \d x + \int_{1-p^{1/d}}^{\frac{1}{2}p^{1/d}} \frac{1-p^{1/d}}{p^{1/d}}\d x \right).
    \end{align*} Solving the integrals then yields \begin{align*}
        \Pr{v_2 \nsim v_3 \mid v_1 \sim v_2, v_3} \le 3p^{-1/d} -2 - p^{-2/d}
    \end{align*} implying the desired bound.
\end{proof}

We use this finding to show in the following lemma that, although
cliques of size at least $4$ already behave like in the IRG model if
$d = \omega(\log(n))$, the number of triangles in the geometric case
is still larger than that in the non-geometric case.

\begin{lemma}\label{lem:highdimlowerbound}
  Let $G = \GIRG(n, \beta, w_0, d)$ be a standard GIRG with
  $d = \omega(\log(n))$ and let $v_1, v_2, v_3$ be three random
  vertices.  Then,
  \begin{align*}
    \Pr{v_2 \sim v_3 \mid v_1 \sim v_2, v_3} \ge (1+o(1))\left( \frac{\kappa_0}{n} \right)^{1 - \frac{\ln(n)^2}{d^2} \pm \mathcal{O}\left(\frac{\ln(n)^3}{d^3} \right)}.
  \end{align*}
\end{lemma}
\begin{proof}
    By \cref{lem:highdimtriangles}, we have $\Pr{v_2 \sim v_3 \mid v_1 \sim v_2, v_3} = q^d$ with \begin{align*}
        q \coloneqq 3 \left( 1 - \left(\frac{\kappa_0}{n}\right)^{-1/d}\right) + \left(\frac{\kappa_0}{n}\right)^{-2/d}.
    \end{align*} In the following, we once again abbreviate $p = \frac{\kappa_0}{n}$ and start by decomposing $q$ into a sum. \begin{align*} 
        q &= 3 \left( 1 - p^{-1/d}\right) + p^{-2/d}\\
        &= 3 \left( 1 - \exp\left(\frac{ -\ln(p)}{d} \right) \right) + \exp\left(\frac{ -2\ln(p)}{d} \right) \\ 
        & = -3 \sum_{i=1}^\infty \frac{1}{i!} \left( \frac{-\ln(p)}{d} \right)^i + \sum_{i=0}^\infty \frac{1}{i!} \left( \frac{-2\ln(p)}{d} \right)^i\\
        &= 1 + \sum_{i=1}^\infty \frac{1}{i!} \left( \left( \frac{-2\ln(p)}{d} \right)^i - 3\left( \left( \frac{-\ln(p)}{d} \right)^i \right) \right)\\
        &= 1 + \sum_{i=1}^\infty \frac{1}{i!} \left( \frac{-\ln(p)}{d} \right)^i (2^i - 3).
    \end{align*} We now proceed by bounding $d\ln(q)$ and apply the Taylor series of $\ln$ to get \begin{align}\label{eq:dlnq}
        d \ln(q) = -d \sum_{i=1} \frac{(1-q)^i}{i} = \sum_{i=1} (-1)^{i+1} d \frac{(q-1)^{i}}{i}.
    \end{align} By abbreviating $x \coloneqq \ln(p)/d$ and our above sum, we get \begin{align}\label{eq:defz}
        q -1 = \sum_{i=1}^\infty \frac{1}{i!} \left( \frac{-\ln(p)}{d} \right)^i (2^i - 3) &= x + \frac{1}{2}x^2 - \frac{5}{6}x^3 + \ldots \nonumber \\
        &=  x \underbrace{ \left( 1 + \frac{1}{2}x - \frac{5}{6}x^2 + R \right) }_{Z},
    \end{align} where $|R| = \mathcal{O}(|x|^3)$ because $|x| = o(1)$.
    Inserting this into the first terms of the sum in \cref{eq:dlnq} yields \begin{align*}
        d \ln(q) &= d x Z - d \frac{1}{2}x^2Z^2 + d \frac{1}{3} x^3 Z^3 + \ldots\\
        &= d x\left( Z - \frac{1}{2}xZ^2 + \frac{1}{3}x^2Z^3 + \sum_{i=4} (-1)^{i+1} x^{i-1}Z^{i} \right).
    \end{align*}
    Using the definition of $Z$ in \cref{eq:defz}, we may carefully compute the leading terms of the above sum to obtain \begin{align*}
        d \ln(q) = dx \left( 1 - x^2 + R' \right),
    \end{align*} with $|R'| = \mathcal{O}(|x|^3)$. Recalling that $x = \ln(p) / d$ gives \begin{align*}
        d \ln(q) = \ln(p) \left( 1 - \left(\frac{\ln(p)}{d}\right)^2 \pm \mathcal{O} \left( \left|\frac{\ln(p)}{d}\right|^3 \right) \right).
    \end{align*} As $p = \kappa_0/n$, this shows \begin{align*}
        q^d = \left( \frac{\kappa_0}{n} \right)^{1 - \frac{ \ln(n/\kappa_0)^2}{d^2} \pm \mathcal{O}\left( \frac{\ln(n)^3}{d^3} \right)} = (1+o(1))\left( \frac{\kappa_0}{n} \right)^{1 - \frac{ \ln(n)^2}{d^2} \pm \mathcal{O}\left( \frac{\ln(n)^3}{d^3} \right)}.
    \end{align*}
    Here, the factor of $1 + o(1)$ derives from the fact that 
    \[\frac{ \ln(n/\kappa_0)^2}{d^2} = \frac{\ln(n)^2 - 2\ln(n)\ln(\kappa_0) + \ln(\kappa_0)^2}{d^2} = \frac{\ln(n)^2}{d^2} - \frac{1}{\omega(\log(n))}.\]
\end{proof}

The above bound implies that the probability that a triangle forms
among vertices of constant weight is still significantly larger than
in the non-geometric models as long as
$d = o\left( \log^{3/2}(n) \right)$. We summarize this behavior in the
following theorem.
\begin{theorem}
  Let $G = \GIRG(n, \beta, w_0, d)$ be a standard GIRG with
  $\beta > 3$ and $d = \omega(\log(n))$.  Then, there is a constant
  $\delta \ge 0$ such that the expected number of triangles
  fulfills
  \begin{align*}
    \Expected{K_3} = \begin{cases} 
      \Omega\left( \exp\left( \frac{\ln^3(n)}{d^2} \right)\right) & \text{if } 3 < \beta < \infty \\ 
      \Theta\left( \exp\left( \frac{\ln^3(n)}{d^2} \right)\right) & \text{if } \beta = \infty
    \end{cases}
  \end{align*} where $\beta = \infty$ refers to the case where every vertex has the same weight.  
\end{theorem}
\begin{proof}
    We note that, by \cref{lem:highdimlowerbound} \begin{align*}
        \Pr{U_3 \text{ is clique}} \ge (1+o(1)) \left( \frac{\kappa_0}{n} \right)^2\left( \frac{\kappa_0}{n} \right)^{1 - \frac{ \ln(n)^2}{d^2} \pm \mathcal{O}\left( \frac{\ln(n)^3}{d^3} \right)}
    \end{align*} and thus \begin{align*}
        \Expected{K_3} = \binom{n}{3} \Pr{U_3 \text{ is clique}} \ge (1+o(1))\kappa_0^3 n^{\Omega\left( \frac{\ln^2(n)}{d^2} \right)} = \Omega\left(\exp\left(   \frac{ \ln^3(n)}{d^2} \right)\right)
    \end{align*} and the first part of the statement is shown. For the second part, we note that in the case of constant weights, the bound from \cref{lem:highdimtriangles} is the exact (conditional) probability that a triangle is formed and the transformations from \cref{lem:highdimlowerbound} still apply for obtaining an upper bound. 
\end{proof}

\subsection{Characterizing Cliques by Vertex Weights}\label{sec:clique-weight-high-dim}

In this section, we extend the bounds on $q_k$ obtained above to the entire graph and characterize it by the weights of the associated vertices assuming $d = \omega(\log(n))$ and $\beta \in (2,3)$ (the other parameter regimes are discussed in the subsequent section). To this end, we will mainly be concerned with bounding the following integral. 

\begin{lemma}\label{lem:integrationbound}
  Let $\rho$ be the density function of the Pareto distribution, let $\varepsilon = \varepsilon(n) = o(1)$ and define \begin{align*}
    \Lambda(k) \coloneqq \int_{0}^\infty \cdots \int_{0}^\infty \rho(w_1) \cdots \rho(w_k) \prod_{1 \le i < j \le k} \left( \frac{\kappa_{ij}}{n} \right)^{1-\varepsilon} \d w_k \ldots \d w_1.
  \end{align*} For every constant $k$, we have \begin{align*}
    \Lambda(k) = \mathcal{O} \left( n^{\frac{k}{2}(1- \beta)} \right).
  \end{align*}
\end{lemma} 

\begin{remark}
  Clearly, $\Lambda(k)$ is an upper bound on the clique probability if we used $w_0$ as the lower integration limit instead of $0$. However, we show the more general statement where the integration limit is $0$ as it will be useful later in the proof of \cref{lem:weightboundhighdim}.
\end{remark}

\begin{proof}[Proof of \cref{lem:integrationbound}]
  We use a similar technique like Daly et al. \cite{Daly_Haig_Shneer_2020} who bound \begin{align*}
    \Lambda(k) \le \sum_{m = 0}^{k} \binom{k}{m} \Lambda(k, m)
  \end{align*} where \begin{align*}
    &\Lambda(k, m) \coloneqq \\
    & \hspace{.2cm} \underbrace{\int_{0}^{\sqrt{n/\lambda}} \cdots \int_{0}^{\sqrt{n/\lambda}}}_{m \text{ times}} \underbrace{\int_{\sqrt{n/\lambda}}^\infty \cdots \int_{\sqrt{n/\lambda}}^\infty}_{k-m \text{ times}} \hspace{.1cm} \rho(w_1) \cdots \rho(w_k) \prod_{1 \le i < j \le k} \left( \frac{\kappa_{ij}}{n} \right)^{1-\varepsilon} \d w_k \ldots \d w_1.
  \end{align*}
  We start with the first extreme case \begin{align*}
    \Lambda(k, k) \!=\!\! \int_{0}^{\sqrt{n/\lambda}}\!\!\!\!\! \cdots \int_{0}^{\sqrt{n/\lambda}}\!\! \rho(w_1) \cdots \rho(w_k) \prod_{1 \le i < j \le k} \left( \frac{\kappa_{ij}}{n} \right)^{1-\varepsilon} \d w_1 \ldots \d w_k.
\end{align*} Since $\rho(w) = c w^{-\beta}$ for some constant $c$ and as $k$ is constant, \begin{align*}
    \begin{split}
        &\Lambda(k,k) = \\ &\hspace{1cm} \Theta(1) \int_{0}^{\sqrt{n/\lambda}} \cdots \int_{0}^{\sqrt{n/\lambda}}  \prod_{1 \le i < j \le k}\left( \frac{\kappa_{ij}}{n} \right)^{1- \varepsilon} w_1^{-\beta} \cdots w_k^{-\beta} \d w_1 \ldots \d w_k
    \end{split}\\
    &= \Theta(1) n^{-\binom{k}{2}(1-\varepsilon)} \left( \int_{0}^{\sqrt{n/\lambda}} w^{(k-1)(1-\varepsilon)-\beta} \d w \right)^k \\
    &= \Theta(1) n^{-\binom{k}{2}(1-\varepsilon)} n^{\frac{1}{2}k(k-1)(1-\varepsilon) + \frac{k}{2}(1-\beta)} \\
    &= \Theta(1) n^{\frac{k}{2}(1-\beta)}
\end{align*} as $\binom{k}{2} = \frac{1}{2}k(k-1)$. Note that the second step holds for sufficiently large $n$ since $k \ge 3$ and $\beta < 3$ as this leads to an exponent in the integral that is strictly greater than $-1$. Above bounds only hold for sufficiently large $n$ as $\varepsilon = o(1)$. For the other extreme case, let $\wmin = \min\{w_1, \ldots, w_k\}$, then \Cref{cor:upperbound} yields that  \begin{align*}
    \Lambda(k, 0) = \Pr{\wmin \ge \sqrt{n/\lambda}} = \Theta(1) n^{\frac{k}{2}(1-\beta)}
\end{align*} because $\kappa_{ij} = 1$ if $w_i,w_j \ge \sqrt{n/\lambda}$. 

If $m \ge 3$, we bound \begin{align*}
  \Lambda(k, m) &\le \Lambda(m, m) \cdot \Lambda(k-m, 0) \\
  &= \Theta(1) n^{\frac{m}{2}(1-\beta)} n^{\frac{k-m}{2}(1-\beta)} \\
  &=\Theta(1) n^{\frac{k}{2}(1-\beta)} 
\end{align*} as  desired. It thus remains the case $m = \{1,2\}$. To this end, we use that \begin{align*}
  \Lambda(k, m) &\le \Lambda(3, m) \Lambda(k-3, 0)\\
  &= \Theta(1) \Lambda(3, m) n^{\frac{k-3}{2}(1-\beta)}
\end{align*} and so it suffices to show that $\Lambda(3, m) = \Theta(1)n^{\frac{3}{2}(1-\beta)}$ for $m \in \{1,2\}$. For $m = 1$, we bound \begin{align*}
  \begin{split}
      &\Lambda(3, 1) \le \\ &\hspace{1cm} \Theta(1) \int_{0}^{\sqrt{n/\lambda}} \int_{\sqrt{n/\lambda}}^{\infty} \int_{\sqrt{n/\lambda}}^{\infty} \left( \frac{w_1^2w_2w_3}{n} \right)^{2(1-\varepsilon)} w_1^{-\beta}w_2^{-\beta}w_3^{-\beta} \d w_1 \d w_2 \d w_3
  \end{split}\\
  &=  \Theta(1) n^{-2(1-\varepsilon)} \int_{0}^{\sqrt{n/\lambda}} w^{2(1-\varepsilon)-\beta} \d w \left( \int_{\sqrt{n/\lambda}}^{\infty} w^{1-\varepsilon-\beta} \d w\right)^2\\
  &= \Theta(1) n^{-2(1-\varepsilon) + (1-\varepsilon) + \frac{1}{2}(1-\beta) + (1-\varepsilon) + (1-\beta) }\\
  &= \Theta(1) n^{\frac{3}{2}(1 -\beta)}
\end{align*} as desired. Again, this works because for sufficiently large $n$, the exponent in the integral starting at $0$ is greater than $-1$. For the case $m = 2$, we use \begin{align*}
  \begin{split}
      &\Lambda(3, 2) \le \Theta(1) \underbrace{\int_{0}^{\sqrt{n/\lambda}}\int_{0}^{w_1} \int_{\sqrt{n/\lambda}}^{n/w_1}  \prod_{1 \le i < j \le 3} \left( \frac{w_iw_j}{n} \right)^{1-\varepsilon} (w_1w_2w_3)^{-\beta} \d w_3 \d w_2 \d w_1}_{I_1} \\
      & \hspace*{.2cm} + \Theta(1) \underbrace{ \int_{0}^{\sqrt{n/\lambda}}\int_{0}^{w_1} \int_{n/w_1}^\infty \left( \frac{w_1w_2^2w_3}{n^2} \right)^{1-\varepsilon} (w_3w_2w_1)^{-\beta} \d w_3 \d w_2 \d w_1}_{I_2}.  
  \end{split}
\end{align*} 
Note that this is a valid bound since the vertices $v_1, v_2$ are
interchangeable and thus $P_2$ is at most twice as large as the two
integrals above, which is captured by the $\Theta(1)$
terms. Calculations now show that \begin{align*}
  I_1 &=  n^{-3(1-\varepsilon)} \int_{0}^{\sqrt{n/\lambda}}\int_{0}^{w_1} \int_{\sqrt{n/\lambda}}^{n/w_1}  \prod_{1 \le i < j \le 3} \left( w_iw_j\right)^{1-\varepsilon} (w_1w_2w_3)^{-\beta} \d w_3 \d w_2 \d w_1\\
  &= \Theta(1) n^{-3(1-\varepsilon)}  \int_{0}^{\sqrt{n/\lambda}} w_1^{2(1-\varepsilon) - \beta} \int_{0}^{w_1} w_2^{2(1-\varepsilon) - \beta} \int_{\sqrt{n/\lambda}}^{n/w_1} w_3^{2(1-\varepsilon) - \beta} \d w_3 \d w_2 \d w_1\\
  &\le \Theta(1) n^{-3(1-\varepsilon)}  \int_{0}^{\sqrt{n/\lambda}} w_1^{2(1-\varepsilon) - \beta} \int_{0}^{w_1} w_2^{2(1-\varepsilon) - \beta} \left( \frac{n}{w_1} \right)^{2(1-\varepsilon) - \beta + 1} \d w_2 \d w_1\\
  &= \Theta(1) n^{-(1-\varepsilon) - \beta + 1}  \int_{0}^{\sqrt{n/\lambda}} w_1^{-1} \int_{0}^{w_1} w_2^{2(1-\varepsilon) - \beta} \d w_2 \d w_1 \\
  &= \Theta(1) n^{\varepsilon - \beta}  \int_{0}^{\sqrt{n/\lambda}} w_1^{2(1-\varepsilon) - \beta} \d w_1 = \Theta(1) n^{\varepsilon - \beta + 1-\varepsilon - \frac{1}{2}\beta + \frac{1}{2} } = \Theta(1)n^{\frac{3}{2}(1-\beta)}
\end{align*}
and \begin{align*}
  I_2 &= \int_{0}^{\sqrt{n/\lambda}}\int_{0}^{w_1} \int_{n/w_1}^\infty \left( \frac{w_1w_2^2w_3}{n^2} \right)^{1-\varepsilon} (w_3w_2w_1)^{-\beta} \d w_3 \d w_2 \d w_1\\
  &\le \Theta(1) n^{-2(1-\varepsilon)} \int_{0}^{\sqrt{n/\lambda}} w_1^{1-\varepsilon - \beta} \int_{0}^{w_1} w_2^{2(1-\varepsilon) - \beta} \left( \frac{n}{w_1} \right)^{1-\varepsilon - \beta + 1} \d w_2 \d w_1\\
  &= \Theta(1) n^{-(1-\varepsilon) - \beta + 1} \int_{0}^{\sqrt{n/\lambda}} w_1^{-1} \int_{0}^{w_1} w_2^{2(1-\varepsilon) - \beta} \d w_2 \d w_1\\
  &= \Theta(1) n^{\frac{3}{2}(1-\beta)}
\end{align*}
as desired.

The above lemma does not only imply our claimed bounds on the expected number of cliques but also allows us to show that cliques of all sizes form dominantly among vertices of weight in the order of $\sqrt{n}$. 

\begin{lemma}\label{lem:weightboundhighdim}
  Let $d = \omega(\log(n))$ and $\beta \in (2,3)$. Then for all (potentially superconstant) $k$ and all $p \in (0,1)$, there is an $\varepsilon > 0$ such that \begin{align*}
    \Pr{\wmin \ge \varepsilon \sqrt{n} \mid U_k \text{ is clique}} \ge p.
  \end{align*}
\end{lemma}
  For $k > \frac{2}{3-\beta}$, we already showed the statement in \cref{lem:sqrtnboundsufficientlylargek}. For $k \le \frac{2}{3-\beta}$, we use an argument inspired by the techniques introduced in \cite{Hofstad_Leeuwaarden_Stegehuis_2021}. Note that in the following, we assume $k$ to be constant since we are in the case $k \le \frac{2}{3-\beta}$. We write \begin{align*}
    \Pr{\wmin < \varepsilon \sqrt{n} \mid U_k \text{ is clique}} = \frac{\Pr{\wmin < \varepsilon \sqrt{n} \cap U_k \text{ is clique}}}{\Pr{U_k \text{ is clique}}}
  \end{align*} and proceed by showing that this probability can be made arbitrarily small by choosing $\varepsilon$ large enough.
  By considering the event that $\wmin \ge \sqrt{n/\lambda}$, we immediately get $\Pr{U_k \text{ is clique}} = \Omega\left( n^{\frac{k}{2}(1-\beta)} \right)$. To show our statement, we proceed by showing that \begin{align*}
    \Pr{\wmin < \varepsilon \sqrt{n} \cap U_k \text{ is clique}} \le f(\varepsilon) n^{\frac{k}{2}(1- \beta)}  
  \end{align*}
  for a function $f$ that tends to $0$ as $\varepsilon \rightarrow 0$. To this end, note that by \cref{thm:chung_lu_eq}, there is a function $\varepsilon = \varepsilon(n) = o(1)$ such that \begin{align*}
    \Pr{U_k \text{ is clique}} &\le \Lambda(k)\\
    &=  \int_{0}^\infty \cdots \int_{0}^\infty \rho(w_1) \cdots \rho(w_k) \prod_{1 \le i < j \le k} \left( \frac{\kappa_{ij}}{n} \right)^{1-\varepsilon} \d w_k \ldots \d w_1.
  \end{align*} Substituting $w_i = y_i \sqrt{n/\lambda}$ and recalling that $\rho(w) = c w^{-\beta}$ for some constant $c$ yields\begin{align*}
    &\Lambda(k) = \\
    & \hspace{.1cm} \int_{0}^\infty \cdots \int_{0}^\infty c^k (y_1 \cdots y_k)^{-\beta} \sqrt{\frac{n}{\lambda}}^{-k\beta}  \prod_{1 \le i < j \le k} \left( \frac{\min\{ n, \lambda n y_iy_j \}}{n} \right)^{1-\varepsilon} \sqrt{\frac{n}{\lambda}}^k \d y_k \ldots \d y_1\\
    &= \left( \frac{n}{\lambda} \right)^{\frac{k}{2}(1-\beta)} \underbrace{\int_{0}^\infty \cdots \int_{0}^\infty c^k (y_1 \cdots y_k)^{-\beta}  \prod_{1 \le i < j \le k} \left( \min\{ 1, \lambda y_iy_j \} \right)^{1-\varepsilon} \d y_k \ldots \d y_1}_{\coloneqq I}. 
  \end{align*} By \cref{lem:integrationbound}, we have $\Lambda(k) = \mathcal{O}\left( n^{\frac{k}{2}(1-\beta)} \right)$ and thus $I < \infty$. Now \begin{align*}
    &\Pr{U_k \text{ is clique} \cap \wmin < \varepsilon\sqrt{n}} \le \\
    & \hspace{1cm} k \int_{0}^{\varepsilon \sqrt{n}}\int_{0}^\infty \cdots \int_{0}^\infty \rho(w_1) \cdots \rho(w_k) \prod_{1 \le i < j \le k} \left( \frac{\kappa_{ij}}{n} \right)^{1-\varepsilon} \d w_k \ldots \d w_1.
  \end{align*} Using the same subsitution as above, this yields \begin{align*}
    &\Pr{U_k \text{ is clique} \cap \wmin < \varepsilon\sqrt{n}} \le \\
    & \hspace{.1cm} k \left( \frac{n}{\lambda} \right)^{\frac{k}{2}(1-\beta)} \underbrace{\int_{0}^{\varepsilon \sqrt{\lambda}} \cdots \int_{0}^\infty c^k (y_1 \cdots y_k)^{-\beta}  \prod_{1 \le i < j \le k} \left( \min\{ 1, \lambda y_iy_j \} \right)^{1-\varepsilon} \d y_k \ldots \d y_1}_{\coloneqq J(\varepsilon)}.
  \end{align*} As $\lim_{\varepsilon \rightarrow \infty} J(\varepsilon) = I < \infty$, $J(\varepsilon)$ can be made arbitrarily small by choosing $\varepsilon > 0$ small enough, which yields the desired statement.
\end{proof}

\subsection{Bounds on $\Expected{K_k}$ and $\omega(G)$}\label{sec:highdimboundscliquenumber}

We use the findings from section \cref{sec:highdimbounds} to prove the third column of \cref{tab:clique-number} and \cref{tab:expectedcliques}, respectively. We treat the cases $2 < \beta <3$, $\beta  = 3$, and $\beta >3 $ separately.

For $\beta \in (2, 3)$, we show that if $d = \omega(\log(n))$, then
the expected number of cliques is
$n^{\frac{k}{2}(3-\beta)} \Theta(k)^{-k}$ for all $k \ge 3$ just like
in the IRG model.  This is despite the fact that the probability that
a clique forms among vertices of constant weight is still
significantly higher than in the IRG model if
$\log(n) \ll d \ll \log^{3/2}(n)$. The reason for this is that the
probability of forming a clique among vertices of weight close to
$\sqrt{n}$ behaves like that of the IRG model if $d = \omega(\log(n))$
and because cliques forming among these high-weight vertices dominate
all the others.

\begin{theorem}
  Let $G = \GIRG(n, \beta, w_0, d)$ be a standard GIRG with
  $\beta < 3$ and $d = \omega(\log(n))$.  Then,
  $\mathbb{E}[K_k] = n^{\frac{k}{2}(3-\beta)} \Theta(k)^{-k}$ for
  $k \ge 3$ and $\omega(G) = \Theta(n^{(3 - \beta)/2})$ a.a.s.
\end{theorem}
\begin{proof}
Observe that \cref{cor:upperbound} and \cref{cor:lowerbound} imply our desired bounds if $k > \frac{2}{3-\beta}$. Otherwise, just considering the event that $\wmin \ge \sqrt{n/\lambda}$, gives us the desired lower bound on $q_k$ and thus on $\mathbb{E}[K_k]$. To get an upper bound, note that $q_k \le \Lambda(k)$ by \cref{thm:chung_lu_eq}, and that $\Lambda(k) = \mathcal{O}\left(n^{\frac{k}{2}(3-\beta)}\right)$ by \cref{lem:integrationbound}, which directly implies our claimed bounds on $\mathbb{E}[K_k]$.
To bound $\omega(G)$, now the same argumentation as in \Cref{thm:clique-number-first-row} applies.
\end{proof}

We note that the phase transition at $k = \frac{2}{3-\beta}$ is not
present anymore.  We continue with the case where $\beta = 3$.

\begin{theorem}
  Let $G = \GIRG(n, \beta, w_0, d)$ be a standard GIRG with
  $\beta = 3$ and $d = \omega(\log(n))$.  Then,
  $\omega(G) = \mathcal{O}(1)$.
\end{theorem}
\begin{proof}
  
We show that the number of cliques of size $k = 4$ is such that for every $\varepsilon > 0$ there is a constant $c > 0$ such that \begin{align*}
   \Pr{K_4 \ge c} \le \varepsilon.
\end{align*} That is, the probability that $K_4 \ge c$ can be made arbitrarily small by choosing $c$ large enough. This fact is sufficient to show that the clique number is $\Theta(1)$.

We start by observing that $\beta = 3$ implies that the maximum weight $w_{max}$ is a.a.s. in the order $\sqrt{n}$. More precisely, denoting by $X_w$ the number of vertices with weight at least $w$, we get by Markov's inequality that for every $w$, \begin{align}\label{eq:maxweight}
    \Pr{w_{max} \ge w} = \Pr{X_w \ge 1} \le \Expected{X_w} = \Theta(1) n w^{1-\beta} = \Theta(1) n w^{-2}.
\end{align} Thus, for every $c \ge 0$, we have \begin{align}\label{eq:weightdecay}
    \Pr{w_{max} \ge c \sqrt{n}} = \mathcal{O}(1)c^{-2}.
\end{align} With this and Markov's inequality, we may bound for every $t, c \ge 0$ \begin{align}\label{eq:markovbound}
    \Pr{K_4 \ge t} &\le \Pr{K_4 \ge t \mid w_{max} \le c \sqrt{n}} + \Pr{w_{max} \ge c\sqrt{n}} \nonumber \\
    &\le \frac{\Expected{K_4 \mid w_{max} \le c \sqrt{n}}}{t} + \Pr{w_{max} \ge c\sqrt{n}}.
\end{align} To bound $\Expected{K_4 \mid w_{max} \le c \sqrt{n}}$, we note that a random weight $w_i$ fulfills \begin{align*}
    \Pr{w_i \le x \mid w_{max} \le c \sqrt{n}} &= \frac{\Pr{\bigcap_{j} w_j \le c \sqrt{n} \cap w_i \le x}}{\Pr{\bigcap_{j} w_j \le c \sqrt{n}}}\\
    &= \frac{\Pr{w_i \le x}}{\Pr{w_i \le c \sqrt{n}}}
\end{align*} and since $\Pr{w_i \le c \sqrt{n}} = 1 - o(1)$, we get that the density of $w_i$ is $\rho_{w_i}(x) = c' x^{-\beta}$ for some $c' \ge 0$ independent of $x$. Using that, we use \cref{thm:chung_lu_eq} to bound \begin{align*}
    \begin{split}
        &\Pr{U_k \text{ is clique } \mid w_k \le c\sqrt{n}} \le \\ &\hspace{1cm} \Theta(1) \int_{w_0}^{c\sqrt{n}} \cdots \int_{w_0}^{c\sqrt{n}}  \prod_{1 \le i < j \le k}\left( \frac{\kappa_{ij}}{n} \right)^{1- \varepsilon} w_1^{-\beta} \cdots w_k^{-\beta} \d w_1 \ldots \d w_k 
    \end{split}\\
    &= \Theta(1) n^{-\binom{k}{2}(1-\varepsilon)} \left( \int_{w_0}^{c\sqrt{n}} w^{(k-1)(1-\varepsilon)-\beta} \d w \right)^k
\end{align*} for some function $\varepsilon(n) = o(1)$. Because $k \ge 4$ and $\beta = 3$, we observe that the exponent $(k-1)(1-\varepsilon)-\beta$ is greater than $-1$ for sufficiently large $n$, and hence, the above integral evaluates to \begin{align*}
    \Pr{U_k \text{ is clique } \mid w_k \le c\sqrt{n}} = c^{k((k-1)(1-\varepsilon) + 1 - \beta)} \Theta(1) n^{\frac{k}{2}(1-\beta)}. 
\end{align*} For $k = 4$, we obtain \begin{align*}
    \Pr{U_k \text{ is clique } \mid w_k \le c\sqrt{n}} \le \Theta(1) c^4 n^{-4}
\end{align*} and accordingly \begin{align*}
    \Expected{K_4 \mid w_{max} \le c \sqrt{n}} = \binom{n}{4} \Pr{U_k \text{ is clique } \mid w_k \le c\sqrt{n}} \le \Theta(1) c^4.  
\end{align*} By (\ref{eq:markovbound}) and (\ref{eq:weightdecay}), this implies \begin{align*}
    \Pr{K_4 \ge t} = \Theta(1)c^4t^{-1} + \mathcal{O}(1)c^{-2}.
\end{align*} Setting $c = t^{1/5}$ yields $\Pr{K_4 \ge t} = \mathcal{O}(1)(t^{-1/5} + t^{-2/5})$ and shows that the probability that $K_4 \ge t$ can be made arbitrarily small by increasing $t$. To bound the clique number, we note that the existence of a clique of size $k$ implies the existence of $\binom{k}{4}$ cliques of size $4$ and so, \begin{align*}
    \Pr{\omega(G) \ge k} \le \Pr{K_4 \ge \binom{k}{4}},
\end{align*} which can be made arbitrarily small by choosing $k$ large enough. Hence the probability that the clique number grows as any superconstant function $f(n) = \omega(1)$ is in $o(1)$, which shows that the clique number is in $\mathcal{O}(1)$ a.a.s.

\end{proof}

Finally, we deal with the case where $\beta > 3$, where we show that,
in this case, there are no cliques of size $4$ or larger
a.a.s. 
\begin{theorem}
  Let $G = \GIRG(n, \beta, w_0, d)$ be a standard GIRG with
  $\beta > 3$ and $d = \omega(\log(n))$.  Then,
  $\mathbb{E}[K_k] = o(1)$ for $k \ge 4$ and, thus, $\omega(G) \le 3$
  a.a.s.
\end{theorem}
\begin{proof}
  We use a similar strategy as in the above paragraph. In analogy to
  \cref{eq:maxweight}, we now have \begin{align*} \Pr{w_{max} \ge
      n^{\alpha}} \le \Theta(1) n^{1 - (1-\beta)\alpha},
\end{align*} which is $o(1)$ if $\alpha > \frac{1}{\beta-1}$. For some $\alpha$ in the range $1/2 > \alpha > \frac{1}{\beta-1}$, we get \begin{align*}
    \begin{split}
        &\Pr{U_k \text{ is clique } \mid w_k \le n^\alpha} \le \\ &\hspace{1cm} \Theta(1) \int_{w_0}^{n^\alpha} \cdots \int_{w_0}^{n^\alpha}  \prod_{1 \le i < j \le k}\left( \frac{\kappa_{ij}}{n} \right)^{1- \varepsilon} w_1^{-\beta} \cdots w_k^{-\beta} \d w_1 \ldots \d w_k 
    \end{split}\\
    &= \Theta(1) n^{-\binom{k}{2}(1-\varepsilon)} \left( \int_{w_0}^{n^\alpha} w^{(k-1)(1-\varepsilon)-\beta} \d w \right)^k\\
    &= \Theta(1) n^{\left(\alpha - \frac{1}{2} \right) k(k-1)(1-\varepsilon) + \alpha k (1 - \beta) } = o\left( n^{\alpha k (1 - \beta)} \right).
\end{align*} Accordingly, \begin{align*}
    \Expected{K_4 \mid w_{max} \le c \sqrt{n}} = \binom{n}{4} \Pr{U_k \text{ is clique } \mid w_k \le c\sqrt{n}} = \Theta(1) n^{4 + 4 \alpha (1-\beta)} = o(1).  
\end{align*} as $\alpha(1-\beta) < -1$. By Markov's inequality this implies \begin{align*}
    \Pr{K_4 \ge 1} \le \Expected{K_4 \mid w_{max} \le n^{\alpha}} + \Pr{w_{max} \ge n^\alpha} = o(1)
\end{align*} as desired.
\end{proof}

\section{Concentration Bounds}\label{sec:concentration-bounds}

We use the insights gained so far to obtain the following concentration bounds on the total number of cliques that hold for almost all parameter regimes and clique sizes we consider. 

\concentration*
\begin{proof}
  We start with the regimes where cliques dominantly form among vertices of weight $\sqrt{n}$. Recall that this covers case (ii), and case (i) if we additionally assume $k < \frac{2}{3-\beta}$. We write \begin{align*}
    K_k = K_k(M_\varepsilon^{(-)}(\sqrt{n})) + K_k(\overline{M_\varepsilon^{(-)}(\sqrt{n})}).
  \end{align*} Since we are in the regime where cliques form dominantly in $M_\varepsilon(\sqrt{n})$, as shown in \Cref{sec:cliquesbyvertexweightslowdim} and \Cref{sec:clique-weight-high-dim}, there is some function $\varepsilon = o(1)$ such that \begin{align*}
    \Expected{K_k(M_\varepsilon^{(-)}(\sqrt{n}))} = (1-o(1))\Expected{K_k} \text{ and } \Expected{K_k(\overline{M_\varepsilon^{(-)}(\sqrt{n})})} = o(1) \Expected{K_k}.
  \end{align*} Furthermore, by \cref{lem:concentration-bound-low-weight} and Chebyshev's inequality, we have \begin{align*}
    \Pr{ \left| \frac{K_k(M_\varepsilon^{(-)}(\sqrt{n}))}{\Expected{K_k}} - 1 \right| \ge \delta } &\le (1+o(1))\delta^{-2} \sum_{\ell = 1}^k \mathcal{O}\left( \frac{\varepsilon^{1 - \beta} k^2}{n^{(3-\beta)/2}} \right)^{\ell}\\
    &= \mathcal{O} \left( \frac{\delta^{-2} \varepsilon^{1-\beta} k^2}{n^{(3-\beta)/2}} \right).
  \end{align*} If we choose an $\varepsilon = o(1)$ that decays sufficiently slowly, this tends to $0$ for every $\delta > 0$ due to our assumption on $k$. Furthermore, \begin{align*}
    K_k(\overline{M_\varepsilon^{(-)}(\sqrt{n})})/\Expected{K_k} \rightarrow_p 0
  \end{align*} by Markov's inequality and the proof of this case is finished.

  On the other hand, if cliques dominantly form among low-weight vertices (this is the case for case (iii), and case (1) for $k > \frac{2}{3-\beta}$), write \begin{align*}
    K_k = K_k(M_\varepsilon^{(+)}(w)) + K_k(\overline{M_\varepsilon^{(+)}(w)})
  \end{align*} for some $w = e^{\Theta(1)d} k^{\frac{1}{2-\beta}}$ and note that this covers covers the cases (iii) and (iv). By \cref{lem:generalvarbound}, \begin{align*}
    \Pr{ \left| \frac{K_k(M_\varepsilon^{(+)}(w))}{\Expected{K_k}} - 1 \right| \ge \delta } \le (1+o(1))\delta^{-2} \frac{\mathcal{O}(w/\varepsilon)^k}{\Expected{K_k}}.
  \end{align*} Note that due to $k = o\left( \frac{\log(n)}{\log\log(n) + d} \right) = n^{o(1)}$, and $w = n^{\Theta(1)d / \log(n)} k ^{\frac{1}{2-\beta}}$, we have $\mathcal{O}(w)^k = n^{o(1)}$, and further $\Expected{K_k} \ge n^c$ for some constant $c > 0$ (cf. \Cref{thm:cliquenumberlowerbound}). Accordingly, \begin{align*}
    \Pr{ \left| \frac{K_k(M_\varepsilon^{(+)}(w))}{\Expected{K_k}} - 1 \right| \ge \delta } \le (1+o(1))\delta^{-2} \varepsilon^{-k} n^{-c + o(1)}.
  \end{align*} If we choose an $\varepsilon = o(1)$ that decays sufficiently slowly, the above term is $o(1)$. Moreover -- as above -- for $\varepsilon = o(1)$, we have \begin{align*}
    \Expected{K_k(M_\varepsilon^{(+)}(w))} = (1-o(1))\Expected{K_k} \text{ and } \Expected{K_k(\overline{M_\varepsilon^{(+)}(w)})} = o(1) \Expected{K_k}
  \end{align*}  and the rest of the proof is analogous to the previous case.
\end{proof}

\section{Asymptotic Equivalence with IRGs}\label{sec:asympeq}
We continue by studying the infinite-dimensional limit of our
model, i.e., the case where $n$ is fixed and $d$ goes to infinity. We
prove that in this situation, the GIRG model becomes in a strong sense
equivalent to the non-geometric IRG model. That is, we prove that the
total variation distance of the distribution over all possible graphs
with $n$ vertices goes to $0$ as $d \rightarrow \infty$. We prove the
following

\asymptoticchunglueq*

We split the proof of this theorem by considering the case of an
$L_p$-norm for $1 \le p < \infty$ and the case of $L_\infty$-norm
separately. Our investigations in \cref{apx:torusvscube} further show
why RGGs on the torus become equivalent to non-geometric models as $d$
tends to infinity and why this is not the case if we use the hypercube
instead as previously observed in
\cite{Erba_Ariosto_Gherardi_Rotondo_2020,dc-r-02}.

\subsection{Equivalence for $L_p$-norms with $1 \le p < \infty$}
Our argument builds upon a multivariate central-limit-theorem similar
to the one used by Devroye et al. who establish a similar statement
for SRGGs~\cite{Devroye_Gyoergy_Lugosi_Udina_2011} .

Before starting the proof, we introduce some necessary auxiliary
statements. Our argumentation builds upon the following Berry-Esseen
theorem introduced in~\cite{Raic_2019}.

\begin{theorem}[Theorem 1.1 in \cite{Raic_2019}]\label{thm:centrallimit}
    Let $Z_1, \ldots, Z_d$ be independent zero-mean random variables taking values in $\mathbb{R}^m$. Let further $Z \coloneqq \sum_{i=1}^d Z_i$ and assume that the covariance matrix of $Z$ is the identity matrix. Let $X \in \mathbb{R}^m$ be a random variable following the standard $m$-variate normal distribution $\mathcal{N}(0, I)$. Then for any convex set $A \subseteq \mathbb{R}^m$, we have
    \begin{equation*}
        \abs{\Pr{Z\in A} - \Pr{X \in A} } \le (42d^{1/4} + 16) \sum_{i=1}^d\Expected{\norm{Z_i}^3}
    \end{equation*}
    where $\norm{Z_i}$ is the $L_2$-norm of $Z_i$.
\end{theorem}

This illustrates that for $d \rightarrow \infty$, the (random) distance between two vertices behaves like a Gaussian random variable. Throughout this section, we use the following notation. For any $u,v \in V$, we define $\Delta_{(u,v)} \in \mathbb{R}^d$ as the component-wise distance of $u$ and $v$, i.e., $\Delta_{(u,v), i} \coloneqq  \abs{\F{x}_{ui} - \F{x}_{vi}}_C = \min\{\abs{\F{x}_{ui} - \F{x}_{vi}}, 1 - \abs{\F{x}_{ui} - \F{x}_{vi}}\}$. Recall that $t_{uv}$ is the connection threshold of the vertices $u,v$, i.e., $u,v$ are adjacent if and only if their distance is at most $t_{uv}$. We may express \begin{align*}
    \Pr{u \sim v} = \Pr{\norm{\Delta_{(u,v)}}_p \le t_{uv}} = \Pr{\sum_{i=1}^d \Delta_{(u,v),i}^p \le t_{uv}^p}
\end{align*}
and we further note that $\Delta_{(u,v), i}$ and $\Delta_{(u,v), j}$ are i.i.d. random variables. Define \begin{align*}
    \mu \coloneqq \Expected{\Delta_{(u,v), i}^p} \hspace{.5cm} \text{and} \hspace{.5cm} \sigma^2 \coloneqq \Var{\Delta_{(u,v), i}^p}. 
\end{align*} and let the random variable $Z_{(u,v),i}$ be defined as \begin{align*}
    Z_{(u,v),i} \coloneqq \frac{\Delta_{(u,v),i}^p- \mu}{\sqrt{d}\sigma}.
\end{align*}
Now define $Z_{(u,v)} \coloneqq \sum_{i=1}^d Z_{(u,v),i}$ and observe that this allows us to express \begin{align*}
    \Pr{u \sim v} = \Pr{\sum_{i=1}^d \Delta_{(u,v),i}^p \le t_{uv}^p} = \Pr{Z_{(u,v)} \le \frac{t_{uv}^p - d \mu}{\sqrt{d}\sigma}}.
\end{align*} 
Working with $Z_{(u,v)}$ instead of $\Delta_{(u,v)}$ has the advantage that we have $\Expected{Z_{(u,v)}} = 0$ and \begin{align*}
    \Var{Z_{(u,v)}} = \sum_{i=1}^d\Var{Z_{(u,v),i}} = \frac{1}{d\sigma^2} \sum_{i=1}^d\Var{\Delta_{(u,v),i}^p} = 1.
\end{align*} These properties are useful when applying \cref{thm:centrallimit}. Now recall that $t_{uv}$ is defined so that the marginal connection probability $\Pr{u \sim v}$ is equal to $\min\left\{1,\frac{\lambda w_uw_v}{n}\right\}$, which is required in order to ensure that $\Expected{\deg(v)} \propto w_v$ for all $v$. We use this to establish the following lemma describing the asymptotic behavior of the threshold $t_{uv}$.

\begin{restatable}[]{lemma}{limesthreshold}\label{lem:limes}
  Let $G = \GIRG(n, \beta, w_0, d) = (V, E)$ be a standard GIRG with
  $L_p$-norm for $p \in [1, \infty)$.  Denote by $\Phi$ the cumulative
  density function of the standard Gaussian distribution, i.e.,
  $\Phi(x) = \sqrt{2\pi} \int_{-\infty}^x e^{-t^2/2} dt$.  Then for
  any $u,v \in V$ with $u \neq v$ and $\frac{\lambda w_uw_v}{n} < 1$,
  we have
  \begin{equation*}
    \lim_{d\rightarrow \infty} \frac{t_{uv}^p - d\mu}{\sqrt{d}\sigma} = \Phi^{-1}\left(\frac{\lambda w_uw_v}{n}\right).
  \end{equation*}
\end{restatable} 
\begin{proof}
  Let $u,v$ be fixed. In the remainder of this proof, we abbreviate
  $Z_{(u,v)}$ with $Z$, and $Z_{(u,v),i}$ with $Z_i$. For every
  $c\in\mathbb{R}$, define the set
  $A_c = \left\{ x \in \mathbb{R} \mid x \le c \right\}$. Let further
  $X \sim \mathcal{N}(0, 1)$ be a standard Gaussian random
  variable. We get from \cref{thm:centrallimit}
  that \begin{align}\label{eq:clim}
         \abs{\Pr{Z \le c } - \Pr{X \le c }} &\le (42d^{1/4} + 16)\sum_{i=1}^d \Expected{\abs{Z_i}^3} \nonumber \\
                                             &= \frac{42d^{1/4} + 16}{d^{3/2}\sigma^3} \sum_{i=1}^d \Expected{\abs{\Delta_{(u,v), i}^p - \mu}^3} \nonumber \\
                                             &\le \frac{42d^{1/4} + 16}{d^{1/2}\sigma^3}\\
                                             &= o_d(1)
    \end{align} because $\Delta_{(u,v), i}^p - \mu \in [-1, 1]$,
    which shows that $Z$ converges to a standard Gaussian random variable as $d \rightarrow \infty$. In particular, this statement is true for $c = \Phi^{-1}\left(\frac{\lambda w_uw_v}{n}\right)$. At the same time, \begin{align*}
        \Pr{X \le \Phi^{-1}\left(\frac{\lambda w_uw_v}{n}\right)}  = \frac{\lambda w_uw_v}{n} = \Pr{Z \le \frac{t_{uv}^p - d \mu}{\sqrt{d}\sigma}},
    \end{align*} where the second step follows by the definition of $t_{uv}$ and $Z$. Hence, by (\ref{eq:clim}), \begin{align*}
        \lim_{d \rightarrow \infty }\abs{\Pr{Z \le \Phi^{-1}\left(\frac{\lambda w_uw_v}{n}\right) } - \Pr{Z \le \frac{t_{uv}^p - d \mu}{\sqrt{d}\sigma}}} = 0.
    \end{align*} 
    Since the function $f(c) = \Pr{Z \le c}$ converges to the cumulative density function $\Phi$ of the standard Gaussian distribution and since this function is continuous and strictly monotonically increasing, we infer that \begin{align*}
        \lim_{d\rightarrow \infty} \frac{t_{uv}^p - d \mu}{\sqrt{d}\sigma} = \Phi^{-1}\left(\frac{\lambda w_uw_v}{n}\right).
    \end{align*}
\end{proof}
With this, we prove the main theorem of this section. 
\begin{proof}[Proof of \cref{thm:chunglueq} for $p \in [1, \infty)$]
  As $n$ is fixed, the set $\mathcal{G}(n)$ is finite and so it
  suffices to show that for all $H \in \mathcal{G}(n)$, we have
\begin{equation}\label{eq:limes}
  \lim_{d \rightarrow \infty}\Pr{G_\text{GIRG} = H} = \Pr{G_\text{IRG} = H}.
\end{equation}
First of all, we note that for any $u,v\in V$ with
$\frac{\lambda w_uw_v}{n} \ge 1$, $u$ and $v$ are guaranteed to be
connected in both $G_\text{GIRG}$ and $G_\text{IRG}$. Hence, for every
$H \in \mathcal{G}(n)$ in which $u$ and $v$ are not connected, we get
$\Pr{G_\text{GIRG} = H} = \Pr{G_\text{IRG} = H} = 0$. For this reason,
it suffices to show \cref{eq:limes} for all $H \in \mathcal{G}(n)$ in
which all $u,v\in V$ with $u\neq v$ and
$\frac{\lambda w_uw_v}{n} \ge 1$ are connected.

Let $H$ be an arbitrary but fixed such graph. We define the set
$\mathcal{Q} = \{ (u,v) \mid 1 \le u < v \le n, \frac{\lambda w_u
  w_v}{n} < 1 \}$, which contains all pairs of vertices that are not
connected with probability $1$. For any event $\Event{}$, we denote by
$\mathds{1}(\Event{})$ the random variable that is $1$ if $\Event{}$
occurs and $0$ otherwise. Similarly, we denote by $H_{(u,v)}$ an
indicator variable that is $1$ if the edge $\{u,v\}$ is present in $H$
and $0$ otherwise. Furthermore, for every $(u,v) \in \mathcal{Q}$, we
define $\mathcal{N}_{(u,v)}$ to be an independent standard Gaussian
random variable. Then, \begin{equation*} \Pr{G_\text{IRG} = H} =
  \Pr{\bigcap_{(u,v)\in\mathcal{Q}} \left(
      \mathds{1}\left(\mathcal{N}_{(u,v)} \le \Phi^{-1}\left(
          \frac{\lambda w_uw_v}{n} \right) \right) = H_{(u,v)} \right)
  }.
\end{equation*}
Furthermore, recall the definition of $Z_{(u,v)}$ and observe \begin{equation*}
    \Pr{G_\text{GIRG} = H} = \Pr{\bigcap_{(u,v)\in\mathcal{Q}} \left( \mathds{1}\left( Z_{(u,v)} \le \frac{t_{uv}^p - d\mu}{\sqrt{d}\sigma} \right) = H_{(u,v)} \right) }.
\end{equation*} In addition, we define the random graph $\tilde{G}$ in which all $u\neq v\in V$ with $\frac{\lambda w_uw_v}{n} \ge 1$ are guaranteed to be connected, and in which for every $(u,v) \in \mathcal{Q}$, the edge $\{u,v\}$ is present if and only if $Z_{(u,v)} \le \Phi^{-1}\left(\frac{\lambda w_uw_v}{n}\right)$. Accordingly, \begin{equation*}
    \Pr{\tilde{G} = H} = \Pr{\bigcap_{(u,v)\in\mathcal{Q}} \left( \mathds{1}\left( Z_{(u,v)} \le \Phi^{-1}\left( \frac{\lambda  w_uw_v}{n} \right) \right) = H_{(u,v)} \right) }.
\end{equation*}
From \cref{lem:limes}, we get that \begin{align*}
    \lim_{d \rightarrow \infty} \frac{t_{uv}^p - d\mu}{\sqrt{d}\sigma} = \Phi^{-1}\left(\frac{\lambda w_uw_v}{N}\right)
\end{align*} and so, \begin{equation*}
    \lim_{d\rightarrow \infty} \abs{\Pr{G_\text{GIRG} = H} - \Pr{\tilde{G} = H} } = 0.
\end{equation*} It therefore only remains to show $\lim_{d \rightarrow \infty} \Pr{\tilde{G} = H} = \Pr{G_\text{IRG} = H}$.

For this, we let $m = |\mathcal{Q}|$ and we define the random vector $Z_i \in \mathbb{R}^{m}$ that has the random variables $Z_{(u,v),i}$ as its components for all $(u,v) \in \mathcal{Q}$. We further define $Z \coloneqq \sum_{i = 1}^d Z_i$.
We use \cref{thm:centrallimit} and define the set $A \subseteq
\mathbb{R}^{m}$ such that
\begin{align*}
  x \in A \Leftrightarrow \forall (u,v) \in
  \mathcal{Q}: \begin{cases}
                 x_{(u,v)} \le \Phi^{-1}\left(\frac{\lambda  w_uw_v}{n} \right) & \text{if } H_{(u,v)} = 1 \\
                 x_{(u,v)} > \Phi^{-1}\left( \frac{\lambda w_uw_v}{n}\right) & \text{otherwise.}
               \end{cases} 
\end{align*}
It is easy to observe that $A$ is convex. We get
$\Pr{\tilde{G} = H} = \Pr{Z \in A}$ and
$\Pr{G_\text{IRG} = H} = \Pr{X \in A}$ where $X$ is a random variable
from the standard $m$-variate normal distribution.

We further note that for all $(u,v) \in \mathcal{Q}$, the random variables $\Delta_{(u,v),1}, \ldots, \Delta_{(u,v),d}$ are independent, which implies that $Z_1, \ldots, Z_d$ are independent as well. Furthermore, they have expectation $0$ and for all $1 \le i, j \le d, (u,v),(s,t) \in \mathcal{Q}$ with $(u,v) \neq (s, t)$, the random variables $\Delta_{(u,v),i}$ and $\Delta_{st,j}$ are independent, even if $i = j$ and $\{u,v\} \cap \{s,t\} \neq \emptyset$ (because the torus is a homogeneous space). This implies that also $Z_{(u,v),i}$ and $Z_{(s,t),j}$ as well as $Z_{(u,v)}$ and $Z_{(s,t)}$ are independent. Hence, $\text{Cov}\left[ Z_{(u,v)},Z_{(s,t)} \right] = 0$. Accordingly, the covariance matrix of $Z$ is the identity matrix. Thus, \cref{thm:centrallimit} implies \begin{align*}
    \abs{\Pr{Z \in A}\! -\! \Pr{X \in A}} &\le (42d^{1/4} + 16)\sum_{i=1}^{d}\Expected{\norm{Z_i}^3}\\
    &= (42d^{1/4} + 16)\sum_{i=1}^{d} \Expected{\left(\sum_{(u,v)\in \mathcal{Q}} Z_{(u,v),i}^2\right)^{3/2}}  \\
    &= (42d^{1/4} + 16)\sum_{i=1}^{d}\frac{1}{d^{3/2}\sigma^3} \Expected{\left( \sum_{(u,v)\in \mathcal{Q}}\!\!\!\left(\Delta_{(u,v),i}^p - \mu\right)^2\right)^{3/2}}\\
    &\le (42d^{1/4} + 16)\sum_{i=1}^{d} \frac{m^{3/2}}{d^{3/2}\sigma^3}\\
    &= (42d^{1/4} + 16) \frac{dm^{3/2}}{d^{3/2}\sigma^3}\\
    &= \frac{m^{3/2}}{\sigma^3} \left(\frac{42}{d^{1/4}} + \frac{16}{d^{1/2}}\right) = o_d(1),
\end{align*}
as $\Delta_{(u,v), i}^p - \mu \in [-1, 1]$ and $m \le
\binom{n}{2}$. This shows \begin{align*} \lim_{d \rightarrow \infty}
  \Pr{\tilde{G} = H} = \Pr{G_\text{IRG} = H},
\end{align*} as desired.
\end{proof} 

As mentioned before, the above result helps in getting an intuition
for how the choice of the underlying ground space of geometric random
graphs affects the impact of an increasing dimensionality.  Recall
that RGGs on the hypercube do \emph{not} converge to Erdős–Rényi
graphs as $n$ is fixed and $d \rightarrow \infty$~\cite{dc-r-02,
  Erba_Ariosto_Gherardi_Rotondo_2020}.  However, our results imply
that they do when choosing the torus as ground space.  These apparent
disagreements are despite the fact that we apply the central limit
theorem similarly.

As discussed before, the above proof relies on the fact that, for
independent zero-mean variables $Z_1, \dots, Z_d$, the covariance
matrix of the random vector $Z = \sum_{i=1}^d Z_i$ is the identity
matrix.  This is due to the fact that the torus is a homogeneous
space, implying that the probability measure of a ball of radius $r$
is the same, regardless of where this ball is centered.  This makes
the random variables $Z_{(u,v)}$ and $Z_{(u, s)}$ independent.  As a
result their covariance is $0$ although both ``depend'' on the
position of $u$.

For the hypercube, this is not the case.  Although the distance of two
vertices can analogously be defined as a sum of independent, zero-mean
random vectors over all dimensions just like we do above, (the only
difference being that $\Delta_{(u,v),i}$ is now the distance between
$u, v$ in dimension $i$ in the hypercube, leading to different values
of $\mu$ and $\sigma^2$) the random variables $Z_{(u,v)}$ and
$Z_{(u, s)}$ do \emph{not} have a covariance of $0$.

In fact, one can verify that for every $1 \le i \le d$, there is a
slightly positive covariance between $\Delta_{(u,v),i}$ and
$\Delta_{(u,s),i}$ (equal to $1/180$).  This transfers to a covariance
between $Z_{(u, v)}$ and $Z_{(u, s)}$, which stays constant as $d$
grows, since
\begin{align*}
  \text{Cov}[Z_{(u,v)}, Z_{(u,s)}] &= \Expected{Z_{(u,v)} \cdot Z_{(u,s)}}\\
                                   &= \sum_{i=1}^d\sum_{j=1}^d \Expected{Z_{(u,v),i} \cdot Z_{(u,s),j}}\\
                                   &= \frac{1}{d\sigma^2} \sum_{i=1}^d \text{Cov}[\Delta_{(u,v),i}^p\Delta_{(u,s),i}^p]\\
                                   &=
                                     \frac{\text{Cov}[\Delta_{(u,v),1}^p\Delta_{(u,s),1}^p]}{\sigma^2},
\end{align*}
where we used that $\Expected{Z_{(u,v),i} \cdot Z_{(u,s),j}} = 0$ if
$i \neq j$.  Accordingly, the covariance matrix $\Sigma$ of $Z$ is not
the identity matrix. Nevertheless, our proof from the previous section
still applies if we replace $Z$ by $Y = \sum_{i=1}^d \Sigma^{-1} Z_i$.
Now $Y$ is the sum of independent random vectors and has the identity
matrix as its covariance matrix, so \cref{thm:centrallimit} remains
applicable.  Furthermore, $\Expected{||\Sigma^{-1} Z_i ||^3}$ is still
proportional to $d^{-3/2}$ and thus remains bounded such that $Y$
converges to a standard $m$-dimensional normal vector. This implies
that $Z$ converges to a random vector from $\mathcal{N}(0, \Sigma)$
showing that RGGs on the hypercube converge to a non-geometric model
where the probability that any fixed graph is sampled can be described
-- like above -- as the probability that
$Z \sim \mathcal{N}(0, \Sigma)$ falls into the convex set $A$. In this
model, however, the edges are not drawn independently, as $\Sigma$ is
not the identity matrix.  In fact, for any three vertices $s, u, v$,
the components $Z_{(u,v)}$ and $Z_{(u,s)}$ are slightly positively
correlated, so there is a slightly higher probability that $s, u, v$
form a triangle than in a corresponding Erdős–Rényi graph.  This leads
to a higher tendency to form cliques, which is in accordance with the
observations from Erba et
al.~\cite{Erba_Ariosto_Gherardi_Rotondo_2020}.

\subsection{Asymptotic Equivalence for $L_\infty$-norm}

In this section, we prove that our model also loses its geometry if $L_\infty$-norm is used. We use a different technique to prove this theorem, as the $L_\infty$-distance between two vertices is not a sum of independent random variables anymore and central limit theorems do not apply anymore. Instead our argument builds upon the bounds we establish in \cref{sec:highdimbounds}.

\begin{proof}[Proof of \cref{thm:chunglueq} for $L_\infty$-norm]
  We show that for all $H \in
  \mathcal{G}(n)$,
  \begin{align}\label{eq:toshow}
    \lim_{d \rightarrow \infty} \Pr{G_{\text{GIRG}} = H} = \Pr{G_{\text{IRG}} = H}.
  \end{align}
  We start by establishing a way to compute $\Pr{G = H}$ for any
  random variable $G$ representing a distribution over all graphs in
  $\mathcal{G}(n)$. For this, we denote by $E(H)$ the set of edges of
  a graph $H = (V, E) \in \mathcal{G}(n)$. We further let
  $\binom{V}{2}$ be the set of all possible edges on the vertex set
  $V$. Now, for any $H \in \mathcal{G}$, we have \begin{align*} \Pr{G
      = H} = \Pr{E(G) \supseteq E(H)} - \sum_{E(H) \subset \mathcal{A}
      \subseteq \binom{V}{2} } \Pr{E(G) = \mathcal{A}}.
    \end{align*} That is, we may express the probability that $G$ is sampled as the probability that a supergraph of $G$ is sampled minus the probability that a any proper supergraph of $G$ is sampled. Now, for any $E(H) \subset \mathcal{A} \subseteq \binom{V}{2}$, the probability $\Pr{E(G) = \mathcal{A}}$ is the probability that $G$ is a specific graph with at least $|E(H)| + 1$ edges. Now, we may repeatedly substitute terms of the form $\Pr{E(G) = \mathcal{A}}$ in the same way until we have an (alternating) sum consisting only of terms that have the form $\Pr{E(G) \supseteq \mathcal{A}}$ for some $E(H) \subset \mathcal{A} \subseteq \binom{V}{2}$. 
    That is, we may calculate the probability $\Pr{G = H}$ even if we only know $\Pr{E(G) \supseteq \mathcal{A}}$ for any $\mathcal{A} \subseteq \binom{V}{2}$.

    As $n$ is fixed, in order to prove our statement in
    \cref{eq:toshow}, it suffices to prove that, for each
    $\mathcal{A} \subseteq \binom{V}{2}$, we have \begin{align*}
      \lim_{d \rightarrow \infty} \Pr{E(G_{\text{GIRG}}) \supseteq
        \mathcal{A}} = \Pr{E(G_{\text{IRG}}) \supseteq \mathcal{A}} =
      \prod_{\{i,j\} \in \mathcal{A}} \frac{\kappa_{ij}}{n}.
    \end{align*} Using \cref{thm:chung_lu_eq}, we get that \begin{align*}
        \Pr{E(G_{\text{GIRG}}) \supseteq \mathcal{A}} = \prod_{\{i,j\} \in \mathcal{A}} \left( \frac{\kappa_{ij}}{n} \right)^{1 \mp \mathcal{O}_d\left( \frac{\ln(n)}{d} \right) }.
   \end{align*} For $d \rightarrow \infty$, this clearly converges to $\Pr{E(G_{\text{IRG}}) \supseteq \mathcal{A}}$ and the proof is finished.
\end{proof}


\bibliographystyle{siamplain}
\bibliography{bibliography.bib}
\end{document}